%% file: ccs2020-main.tex
\documentclass[sigconf]{acmart}

\input{macros}

\newif\ifshowchanges
\showchangesfalse

\ifshowchanges
\newcommand{\changed}[1]{\textcolor{red}{#1}}
\else 
\newcommand{\changed}[1]{#1}
\fi

\setcopyright{none} 
\acmYear{2020}

\begin{document}

\input{ccs2020}

\bibliographystyle{ACM-Reference-Format}
\bibliography{biblio}

\appendix

 \section{Proofs}
 \input{proof_t_eq}

\input{proof_ooo_consistency}
 \input{proof-spec-consistency}
 \input{proof_constant_time}

\end{document}

%% file: macros.tex
\usepackage{algorithmic}
\usepackage{graphicx}
\usepackage{textcomp}
\usepackage{xcolor}
\def\BibTeX{{\rm B\kern-.05em{\sc i\kern-.025em b}\kern-.08em
    T\kern-.1667em\lower.7ex\hbox{E}\kern-.125emX}}
\usepackage{adjustbox} 

\usepackage{hyperref}    

\usepackage{tcolorbox} 
\usepackage{latexsym} 
\usepackage{syntonly} 
\usepackage{amsmath} 
\usepackage{stmaryrd} 
\usepackage{color}
\usepackage{graphics,epsfig, subfigure} 
\usepackage[disable]{todonotes}
\usetikzlibrary{shapes,arrows,shadows,positioning,shapes.multipart,fit,automata}
\usepackage{comment}
\usepackage{paralist}
\usepackage{thmtools,thm-restate} 
\usepackage{multicol} 
 
\usepackage{newfloat} 
\DeclareFloatingEnvironment[
fileext=los,
listname={List of Examples},
name=Example,
placement=tbhp,
within=none,
]{pomset}

\definecolor{block-gray}{gray}{0.85}
\newtcolorbox{newnotation}{colback=block-gray,grow to right
  by=-0mm,grow to left by=-0mm,boxrule=0pt,boxsep=0pt,
  left=0pt,right=0pt,top=0pt,bottom=0pt}

\newcommand{\Regs}{\mathcal{R}}
\newcommand{\Pc}{\mathcal{PC}}
\newcommand{\Locs}{\mathcal{M}}
\newcommand{\ANames}{\mathit{N}}
\newcommand{\Val}{\mathit{V}}
\newcommand{\pc}{\mathit{pc}}
\newcommand{\subst}[3]{#1[#2 \mapsto #3]}
\newcommand{\Stores}{\mathit{Stores}}
\newcommand{\Observations}{\mathit{Obs}}
\newcommand{\trace}{\mathit{trace}}

\newcommand{\inst}{\iota}
\newcommand{\MILlang}{\mathcal{I}}
\newcommand{\dom}[1]{\mathrm{dom}(#1)}

\newcommand{\ass}[3]{#2 \leftarrow #1 ? #3}
\newcommand{\asstrue}[2]{#1 \leftarrow #2}
\newcommand{\load}[2]{ld\ #1\ #2}

\newcommand{\store}[3]{st\ #1\ #2 \ #3}

\newcommand{\dl}[1]{dl\ #1}
\newcommand{\ds}[1]{ds\ #1}
\newcommand{\il}[1]{il\ #1}

\newcommand{\expr}{e}

\newcommand{\instr}{\iota}

\newcommand{\conf}{\mathit{States}}

\newcommand{\step}[2]{\xrightarrow[#2]{#1}}
\newcommand{\singlestep}[2]{\xrightarrow[#2]{#1}}
\newcommand{\doublestep}[2]{\xrightarrow[#2]{#1}\mathrel{\mkern-14mu}\rightarrow}
\newcommand{\triplestep}[2]{\xrightarrow[#2]{#1}\mathrel{\mkern-14mu}\rightarrow\mathrel{\mkern-14mu}\rightarrow}

\newcommand{\singlestepdown}[2]{\downarrow}

\newcommand{\storage}{s}
\newcommand{\Inst}{\mathit{I}}
\newcommand{\commits}{\mathit{C}}
\newcommand{\decodes}{\mathit{F}}
\newcommand{\restype}{\tau}
\newcommand{\newstate}{\sigma}

\newcommand{\completed}[1]{\mathcal{C}(#1)}

\newcommand{\exc}{\footnotesize \textsc{Exe}}
\newcommand{\pexe}{\footnotesize \textsc{Pexe}}
\newcommand{\exe}{\exc}
\newcommand{\cmt}{\footnotesize\textsc{Cmt}}
\newcommand{\ftc}{\footnotesize\textsc{Ftc}}
\newcommand{\ret}{\footnotesize\textsc{Ret}}
\newcommand{\prd}{\footnotesize\textsc{Prd}}
\newcommand{\rbk}{\footnotesize\textsc{Rbk}}

\usepackage{listings}


%% file: ccs2020.tex
\title{InSpectre: Breaking and Fixing Microarchitectural Vulnerabilities by Formal Analysis
}

\author{Roberto Guanciale}
\email{robertog@kth.se}
\orcid{}
\affiliation{%
  \institution{KTH Royal Institute of Technology}
  \city{Stockholm}
  \postcode{SE-100 44}
  \country{Sweden}
}
\author{Musard Balliu}
\orcid{0001-6005-5992}
\email{musard@kth.se}
\affiliation{%
  \institution{KTH Royal Institute of Technology}
  \city{Stockholm}
  \postcode{SE-100 44}
  \country{Sweden}
}
\author{Mads Dam}
\orcid{0001-5432-6442}
\email{mfd@kth.se}
\affiliation{%
  \institution{KTH Royal Institute of Technology}
  \city{Stockholm}
  \postcode{SE-100 44}
  \country{Sweden}
}

\begin{abstract}
The recent Spectre attacks have  demonstrated the fundamental insecurity of current computer microarchitecture. 
The attacks use features like pipelining, out-of-order and speculation to extract arbitrary information about the memory contents of a process. 
A comprehensive formal microarchitectural model capable of representing the forms of out-of-order and speculative behavior that can meaningfully
be implemented in a high performance pipelined architecture has not yet emerged. Such a model would be very useful, as it would allow
the existence and non-existence of vulnerabilities, and soundness of countermeasures to be formally established. 

This paper
 presents such a model targeting single core processors. 
%
%
The model is intentionally very general and provides an infrastructure to define models of real CPUs. 
It incorporates microarchitectural features that underpin all known Spectre vulnerabilities. We use the model to elucidate the security 
of existing and new vulnerabilities, as well as to formally analyze the effectiveness of proposed countermeasures.  
Specifically, we discover three new (potential) vulnerabilities, including a new variant of Spectre v4, a vulnerability on speculative fetching, and a vulnerability on out-of-order execution,  
and analyze the effectiveness of  existing countermeasures including constant time and serializing instructions.

\end{abstract}


\maketitle

\section{Introduction}
The wealth of vulnerabilities that have followed on from Spectre and Meltdown \cite{DBLP:conf/sp/KocherHFGGHHLM019,lipp2018meltdown} have provided ample evidence of the fundamental insecurity of current computer microarchitecture. The extensive use of instruction level parallelism in the form of out-of-order (OoO) and speculative execution has produced designs with side channels that  can be exploited by attackers to learn sensitive information about the memory contents of a process. One witness of the subtlety of the issues is the more than 50 years passed since pipelining, caching, and OoO execution, cf. IBM S/360, was first introduced. 

Another witness is the fact that two years after the discovery of Spectre, a comprehensive understanding of the security implications of pipeline related microarchitecture features has yet to emerge. One result is the ongoing arms race between researchers discovering new Spectre-related vulnerabilities \cite{DBLP:conf/uss/CanellaB0LBOPEG19}, and CPU vendors providing 
patches followed by informal arguments~\cite{cortexA53}. The security and effectiveness of the currently proposed
countermeasures is unknown, and there are continuously new  vulnerabilities appearing that exploit specific microarchitecture features. 

It is important to note that side channels and functional correctness are to a large extent orthogonal.
The latter is
usually proved by reducing pipelined behaviour to sequential behaviour through some form of refinement-based argument. The past decades have seen a significant body of work in this area, 
cf.  \cite{BurchD94,sawada2002verification,AagardCDJ01,ManoliosS05}, addressing rich sets of features of concrete pipeline designs such as OoO, speculation, and self-modifying code. 
Functional correctness, however, focuses on programs' input-output behaviour and fails to adequately capture the differential aspects of speculation and instruction reordering that are at the root of Spectre-like vulnerabilities. For a systematic study of the latter we argue that new tools that are not necessarily tied to any specific pipeline architecture are needed. 

Along this line, several recent works \cite{secspec,guarnieri2018spectector,mcilroy2019spectre,Disselkoen2019TheCT} have started to propose formal microarchitectural models using information flow analysis to identify information leaks arising from speculative execution in a principled manner. 
These models capture specific speculation features, e.g, branch prediction, and
variants of Spectre, in particular variant 1, and design analyses that detect known attacks~\cite{secspec,guarnieri2018spectector,wang2019kleespectre}. While these approaches illustrate the usefulness of  formal models in analyzing microarchitecture leaks, features lying at
the heart of modern CPUs such as OoO execution and many forms of speculation remain largely
unexplored, implying that new vulnerabilities may still exist. 

\paragraph*{Contributions}
This work presents InSpectre, the first comprehensive model capable of
 capturing OoO execution and all forms of speculation that can be meaningfully 
 implemented in the context of a high performance pipeline. The model is intentionally very general and
provides an infrastructure to define models of real CPUs
(Section~\ref{sec:fmm}), which can be used to analyze effectiveness of
countermeasures for a given processor. 

Our first contribution is a novel semantics  supporting microarchitectural 
features such as OoO execution, non-atomicity of instructions, and various forms of
speculation, including branch prediction, jump target prediction, return address prediction, and dependency prediction.
Additionally, the semantics supports features such as
address aliasing, dynamic references, store forward, and OoO memory commits,
which are necessary to model all known variants of Spectre. The semantics implements the stages of an abstract
pipeline supporting OoO (Section~\ref{sec:ooo}) and
speculative execution (Section~\ref{sec:spec}).
In line with existing work \cite{secspec,guarnieri2018spectector}, our security condition
formalizes the intuition that optimizations should not
introduce additional information leaks (conditional
noninterference, Section~\ref{sec:secmod}).
We use this condition to 
show that InSpectre can reproduce all four variants of Spectre.

As a second contribution, we use InSpectre to discover three new potential vulnerabilities. 
The first vulnerability shows that CPUs supporting only OoO may leak sensitive information. 
We discovered the second vulnerability while attempting to validate a CPU vendor's claim that microarchitectures like Cortex A53 are immune to Spectre 
vulnerabilities because they support only speculative fetching \cite{cortexA53}. Our model reveals that this may not be the case. The third vulnerability is a variant
of Spectre v4 showing that speculation of a  dependency, rather than speculation of a non-dependency as in Spectre v4, between a load and a store operation
may also leak sensitive information.

Finally, as a third contribution,  we leverage InSpectre to analyze
the effectiveness of some existing  countermeasures.
We found that
constant-time~\cite{Bernstein05} analysis is unsound for processors supporting only OoO, 
and propose a provably secure fix that enables constant-time analysis
to ensure security for such processors.

\newcommand{\stylecode}[1]{\texttt{#1}}







\section{Security Model}\label{sec:secmod}

Our security model has the following ingredients: 
(i) an \emph{execution} model which is given by the execution
semantics of a program; 
(ii) an \emph{attacker} model specifying the observations of an attacker;
(iii) a \emph{security policy} specifying the parts of the program
state that contain sensitive/high information, and the parts that
contain public/low information; 
(iv) a \emph{security condition} capturing  a program's security with respect
to an execution model, an attacker model, and a security policy.

First, we consider a general model of attacker that observes the interaction between the CPU and the memory subsystem. 
This model has been used (e.g.,~\cite{almeida2016verifying}) to capture information leaks via
cache-based side channels transparently without an explicit cache
model. It can capture trace-driven attackers that can interleave with the victim's execution and
indirectly observe, for instance using Flush+Reload \cite{gruss2016flush}, the victim's cache footprint via latency jitters.
The attacker can observe the address of a memory load $\dl{v}$ (data load from memory address $v$), the address of a memory store $\ds{v}$ (data store to memory address $v$), as well as 
the value of the program counter $\il{v}$ (instruction load from
memory address $v$) \cite{DBLP:conf/icisc/MolnarPSW05}.

We assume a transition relation $\mbox{$\singlestep{}{}$}\subseteq\conf\times\Observations\times\conf$ to model the execution semantics of a program as 
a state transformer  producing observations $l \in \Observations$. The reflexive and transitive closure of $\singlestep{}{}$  induces a set of executions $\pi \in \Pi$.  
The function $\trace: \Pi \mapsto \Observations^\ast$ extracts the sequence of observations of an execution.

The security policy is defined by an \emph{indistinguishability
  relation}  $\mbox{$\sim$}\subseteq\conf\times\conf$.
The relation $\sim$  determines the security of information that is initially stored in a state, modeling the set of initial states that an attacker is not allowed to discriminate.
These states represent the initial \emph{uncertainty} of an attacker about sensitive information.  

The security condition defines the security of a program on the
target execution model (e.g., the speculation model) $\singlestep{}{}_t$ 
conditionally on the security of the same program on the reference, i.e. sequential, model $\singlestep{}{}_r$, by requiring that the target model does not leak more information than the reference model for a  policy $\sim$. 

 %
%
%
%
 \begin{definition}[Conditional Noninterference]\label{def:ci}
 Let $\sim$ be a security policy and
 $\singlestep{}{}_t$ and $\singlestep{}{}_r$ be  transition relations for the
 target and reference  models of a system. 
 The system is conditionally noninterferent if for all
  $\newstate_1, \newstate_2 \in \conf $ such that $\newstate_1 \sim \newstate_2$, if for every  $\pi_1 = \newstate_1 \singlestep{}{}_r
  \cdots$ there exists 
 $\pi_2 = \newstate_2 \singlestep{}{}_r  \cdots$ such that $\trace(\pi_1) =
 \trace(\pi_2)$ then
 for every  $\rho_1 = \newstate_1 \singlestep{}{}_t  \cdots$  there exists  $\rho_2
 = \newstate_2 \singlestep{}{}_t  \cdots$
 such that  $\trace(\rho_1) = \trace(\rho_2)$. 
\end{definition}


Conditional noninterference captures only the new information leaks
that may be introduced by model $\singlestep{}{}_t$, and ignores any leaks already present in model $\singlestep{}{}_r$. 
\changed{
The target model is constructed in two steps. First, we present an OoO model that extends the sequential model, which is deterministic, by allowing evaluation to proceed out-of-order. Then the OoO model is further extended by adding speculation. At each step the traces of the abstract model are included in the extended model, and a memory consistency result demonstrates that the per location sequence of memory stores is the same for both models. This establishes functional correctness. Conditional noninterference then establishes security of each extension. Each such step strictly increases the set of possible traces by adding nondeterminism. Since refinement is often viewed as essentially elimination of nondeterminism, one can think of the extensions as ``inverse refinements''. 
Since conditional noninterference considers a possibilistic setting, it does not account for information leaks through the number of initial indistinguishable states.
}


\changed{We now elucidate the advantages of conditional noninterference as compared to standard notions of  noninterference and declassification. 
Suppose we define the security condition directly on the target model, in the style of standard noninterference.}
\changed{\begin{definition}[Noninterference]
 Let $P$ be a program with transition relation $\singlestep{}{}$ and
 $\sim_P$ a security policy. $P$ satisfies noninterference \changed{up to $\sim_P$} if for all
 $\newstate_1, \newstate_2 \in \conf $ such that $\newstate_1 \sim_P
 \newstate_2$  
 and executions $\pi_1 = \newstate_1 \singlestep{}{}  \cdots$, 
 there exists an execution  $\pi_2 = \newstate_2 \singlestep{}{}  \cdots$ such that $\trace(\pi_1) = \trace(\pi_2)$.
\end{definition}
}

\changed{Noninterference ensures that if the observations do not enable an attacker 
to refine his knowledge  of sensitive information beyond what is allowed by the policy $\sim_P$,  the program can be considered secure.
Noninterference can accommodate partial release of sensitive information by refining the definition of the indistinguishability relation $\sim_P$. 
In our context, a precise definition of $\sim_P$ can be challenging to define. 
However,  we ultimately aim at showing that the OoO/speculative model does not leak more information than the in-order (sequential) model, thus  
capturing the intuition that microarchitectural features like OoO and speculation should not introduce additional leaks. 
Therefore, instead of defining the policy $\sim_P$ explicitly, we split it into two relations $\sim$ (as in Def.~\ref{def:ci}) and $\sim_D$, where the former models information of the initial state that 
is known by the attacker, i.e., the public resources,  and the latter models information that the attacker is allowed to learn during the  execution via observations. Hence, $\mbox{$\sim_P$} = \mbox{$\sim$} \cap \mbox{$\sim_D$}$. 
This characterization allows for a simpler formulation of the security condition that is transparent on the definition of $\sim_D$, as described in Def.~\ref{def:ci}. 
}


\section{Formal Microarchitectural Model}
\label{sec:fmm}

We introduce a Machine Independent Language (MIL) which we use to define the semantics of microarchitectural features such as
OoO and speculative execution. We use MIL as a form of abstract microcode language: A target language for translating ISA instructions and reasoning
about features that may cause vulnerabilities like Spectre. Microinstructions in MIL 
represent atomic actions that can be executed by
the CPU, emulating the pipeline phases in an abstract manner.
This model is intentionally very general and provides an infrastructure to define 
models of real microarchitectures.

We consider a domain of values $v\in \Val$,
a program counter $\pc \in \Pc$,
a finite set of register/flag identifiers $r_0,\ldots,r_n, f, z \in \Regs \subseteq \Val$, and
a finite set of memory addresses $a_0,\ldots,a_m \in \Locs \subseteq
\Val$.
The language can be easily extended to support other type of
resources, e.g., registers for vector operations. 
We assume a total order $<$ on a set of names $t_0,t_1,\ldots \in \ANames$, which we use to uniquely identify  microinstructions. 
We write $\ANames_1 < \ANames_2$ if
for every pair $(t_1, t_2) \in \ANames_1 \times \ANames_2$ it holds that
$t_1 < t_2$.

\newcommand{\fundefined}[2]{#1(#2)\mbox{$\uparrow$}}
\newcommand{\fdefined}[2]{#1(#2)\mbox{$\downarrow$}}
\newcommand{\proj}[2]{#1\mbox{$\mid_{#2}$}}

Microinstructions $\instr \in \MILlang$ are  conditional atomic single assignments.
A microinstruction $\instr=\ass{c}{t}{o}$ is uniquely identified by its name $t \in
\ANames$ and consists of 
a boolean guard $c$, which determines if the assignment
should be executed, and an operation $o \in Op$.
\todo{RG: removed
A MIL \emph{program} $I$ is a set of microinstructions
$\ass{c_i}{t_i}{o_i} \in \mathcal{I}$.
}
The MIL language has three types of operations: 
\[
\begin{array}{l}
e ::= v \ \mid \ t \ \mid \ e_1 + e_2 \ \mid \ e_1 > e_2 \ \mid  \cdots \\
o ::= \expr  \mid\ \load{\restype}{t_a} \mid \ \store{\restype}{t_a}{t_v}
\end{array}
\]

An internal operation $e$ is an expression over standard finite arithmetic and
can additionally refer to names in $\ANames$ and values in $\Val$.
A resource load operation $\load{\restype}{t_a}$, where $\restype \in \{\Pc, \Regs,
\Locs\}$, loads the value of resource $\restype$ addressed by $t_a$. We support three types of resources: The program counter $\Pc$, registers $\Regs$, and memory locations $\Locs$. 
A resource store operation $\store{\restype}{t_a}{t_v}$ uses the value
of $t_v$ to update the resource $\restype$ addressed by $t_a$.

\newcommand{\fn}{\mathit{fn}}
\newcommand{\bn}{\mathit{bn}}
\newcommand{\translate}{\mathit{translate}}

The free names $\fn(\instr)$ of an instruction $\instr = \ass{c}{t}{o}$ is the set of names occurring in $c$ or $o$,
the bound names, $\bn(\instr)$, is the singleton $\{t\}$, and the names $n(\instr)$ is $\fn(\instr) \cup \bn(\instr)$.

To model the internal state of a CPU pipeline, we can translate an ISA instruction as multiple microinstructions.
For an ISA instruction at address $v \in \Locs$ and a name $t\in\ANames$, the function $\translate(v, t)$
returns the MIL translation of the instruction at address $v$, ensuring that the names of the microinstructions thus generated are greater than $t$. 
Because we assume code to not be self-modifying, an instruction can be statically identified by its address in memory.
We assume that the translation function satisfies the properties:
\label{axiom:tranpilation}
\begin{inparaenum}[(i)]
    \item for all $\iota_1,\iota_2\in \translate(v,t)$, if $\iota_1 \neq \iota_2$ then $\bn(\iota_1) \cap
      \bn(\iota_2) = \emptyset$; for all $\iota \in \translate(v,t)$,
    \item  $\fn(\iota) < \bn(\iota)$, and 
    \item $\{t\} < n(\iota)$.
\end{inparaenum}

These properties ensure that names uniquely identify
microinstructions, the name parameters of a
single instruction form a Directed Acyclic Graph, 
the translated microinstructions are assigned names greater than $t$, and
the translation of two different ISA instructions does not
have direct inter-instruction dependencies (but may have indirect ones).


\subsection{MIL Program Examples}

We introduce some illustrative examples of MIL programs, using their
graph representation.
For clarity, we omit conditions whenever they are true 
and visualize only the immediate dependencies between graph elements.


  Consider an ISA instruction
  that increments the value of register $r_1$ by one, i.e.,
  \stylecode{$r_1$:= $r_1$+$1$}. The instruction  
can be translated in MIL as follows:
\[
  \left \{
  \begin{array}{l}
    \asstrue{t_1}{r_1},\:
    \asstrue{t_2}{\load{\Regs}{t_1}},\:
    \asstrue{t_3}{t_2+1},\:
    \asstrue{t_4}{\store{\Regs}{t_1}{t_3}},\\
    \asstrue{t_5}{\load{\Pc}{}}, \:
    \asstrue{t_6}{t_5+4},\:
    \asstrue{t_7}{\store{\Pc}{}{t_6}}
  \end{array}
  \right \}
\]

Intuitively, $t_1$ refers to the identifier of target register $r_1$, $t_2$
loads the current value of register $r_1$, $t_3$ executes the increment, and
$t_4$ stores the result of $t_3$ in the register store. 
The translation of an ISA instruction also updates the program
counter to enable the execution of the next instruction. In this case, the program counter
is increased by 4, unconditionally. Notice that we omit the program
counter's address, since there is only one such resource.
We can graphically represent this set of microinstructions using the
following graph:

\newcommand{\picexpT}[3][xshift=.0cm]{\node[pblock,label={$#2$},name=#2,#1]{$#3$};}
\newcommand{\picloadT}[4][xshift=.0cm]{\node[pblock,label={$#2$},name=#2,#1]{$\load{#3}{#4}$};}
\newcommand{\picloadTT}[4][xshift=.0cm]{\node[pblock,label={[xshift=-0.2cm]$#2$},name=#2,#1]{$\load{#3}{#4}$};}

\newcommand{\picstoreT}[5][xshift=.0cm]{\node[pblock,label={$#2$},name=#2,#1]{$\store{#3}{#4}{#5}$};}
\newcommand{\picstoreTT}[5][xshift=.0cm]{\node[pblock,label={[xshift=-0.2cm]$#2$},name=#2,#1]{$\store{#3}{#4}{#5}$};}
\newcommand{\picexp}[4][xshift=.0cm]{\node[pblock,label={[xshift=-.2cm]$#2$},name=#2,rectangle split, rectangle split horizontal, rectangle split parts=2, #1]{\nodepart{one}  $#3$ \nodepart{two} $#4$};}
\newcommand{\picstore}[6][xshift=.0cm]{\node[pblock,label={[xshift=-.2cm]$#2$},name=#2,rectangle split, rectangle split horizontal, rectangle split parts=2, #1]{\nodepart{one}  $#3$ \nodepart{two} $\store{#4}{#5}{#6}$};}
\newcommand{\picload}[5][xshift=.0cm]{\node[pblock,label={[xshift=-.2cm]$#2$},name=#2,rectangle split, rectangle split horizontal, rectangle split parts=2, #1]{\nodepart{one}  $#3$ \nodepart{two} $\load{#4}{#5}$};}

\newcommand{\pomsetnewline}{0.5cm}

\newcommand{\pomsetstyle}{
  \tikzstyle{every node}=[font=\scriptsize]
  \tikzset{pblock/.style = {rectangle, draw=black!50, top
    color=white,bottom color=white!20, align=center}}
}


\begin{pomset}
  \begin{center}
\begin{tikzpicture}[node distance=2.3cm,auto,>=latex']

\pomsetstyle                      
\picexp{t_1}{\top}{r_1}
\picload[right of=t_1]{t_2}{\top}{\Regs}{t_1}
\picexp[right of=t_2]{t_3}{\top}{t_2+1}
\picstore[right of=t_3]{t_4}{\top}{\Regs}{t_1}{t_3}

\picload[below = \pomsetnewline of t_2]{t_5}{\top}{\Pc}{}
\picexp[right of=t_5]{t_6}{\top}{t_5+4}
\picstore[right of=t_6]{t_7}{\top}{\Pc}{}{t_6}

\draw[->]  (t_1) edge (t_2) (t_2) edge (t_3) (t_3) edge (t_4);
\draw[->]  (t_5) edge (t_6) (t_6) edge (t_7);
\end{tikzpicture}
\end{center}
\vspace*{-.3cm}
\caption{\stylecode{$r_1$ := $r_1$+$1$}}
\label{example:1}
\end{pomset}


In the following we adopt syntactic sugar to use expressions, in place of  names, for the address
and value of load and store operations. This can be eliminated by introducing
the proper intermediary internal assignments. 
This permits to rewrite the previous example as:

\vspace*{-0.3cm}
\begin{pomset}
\begin{center}
\begin{tikzpicture}[node distance=2.4cm,auto,>=latex']
  \pomsetstyle
\picloadT{t_2}{\Regs}{r_1}
\picstoreT[right of=t_2]{t_4}{\Regs}{r_1}{t_2+1}
\picloadT[right of= t_4]{t_5}{\Pc}{}
\picstoreT[right of=t_5]{t_7}{\Regs}{}{t_5+4}
\draw[->]  (t_2) edge (t_4) ;
\draw[->]  (t_5) edge (t_7) ;
\end{tikzpicture}
\end{center}
\end{pomset}
\vspace*{-0.1cm}
The translation of multiple ISA instructions results in disconnected graphs. 
This reflects the fact that inter-instruction dependencies may not be
statically identified due to dynamic references and must be 
dynamically resolved by the MIL semantics.  
When translating multiple instructions, we use the following convention for generated  names: the name $t_{ij}$  identifies the $j$-th microinstruction 
resulting from the translation of the $i$-th instruction. Our convention induces a total (lexicographical) order over names (i.e., $t_{ij} < t_{i'j'}$ iff $(i < i') \vee (i = i' \wedge j < j')$), 
which respects the properties of the translation function.

MIL is expressive enough to
support conditional instructions
like conditional
arithmetic and conditional move. Conditional branches can be
modeled in MIL via microinstructions that are guarded by complementary
conditions. For instance,  the  \stylecode{beq a} instruction, which jumps to address $a$ if the $z$ flag is set,  can be translated as in Example~\ref{example:branch}.


\begin{pomset}
\begin{center}
\begin{tikzpicture}[node distance=2.5cm,auto,>=latex']

  \pomsetstyle

  \picloadT{t_1}{\Regs}{z}
  \picexpT[right of=t_1]{t_2}{t_1 = 1}
  \picloadT[right of=t_2]{t_3}{\Pc}{}
  \picstore[below = \pomsetnewline of t_1]{t_4}{t_2}{\Pc}{}{a}
  \picstore[below = \pomsetnewline of t_3]{t_5}{\neg t_2}{\Pc}{}{t_3+4}

 \draw[->]
 (t_1) edge (t_2)
 (t_2) edge (t_4)
 (t_2) edge (t_5) (t_3) edge (t_5);

\end{tikzpicture} 
\end{center}
\caption{\stylecode{beq a}}
\label{example:branch}
\end{pomset}


%

%
\newcommand{\den}[1]{[#1]}
\newcommand{\lblempty}{\cdot}

\section{Out-of-Order Semantics}\label{sec:ooo}
\changed{
This section presents an OoO semantics for MIL programs, which is extended in
Section~\ref{sec:spec} to account for speculation. The in-order semantics in Section \ref{sec:inorder} is obtained by constraining the OoO semantics to enforce program order evaluation. We prove memory consistency  for both the OoO and speculative semantics. The results show that, for each semantics, the sequences of per location memory stores are the same, thus establishing functional correctness of the OoO and speculative semantics.
}


\subsection{States, Transitions, Observations}
We formalize the semantics via a transition relation $\newstate \doublestep{l}{} \newstate'$, which maps a state $\newstate$ to a state $\newstate'$, and 
produces a (possibly empty, represented by a dot ($\lblempty$)) observation $l \in \Observations$, eliding the dot when convenient.
As in Section~\ref{sec:secmod}, $\Observations = \{\lblempty,\dl{v}, \ds{v}, \il{v}\}$ 
captures the attacker model. 

States $\newstate$ are tuples $(\Inst, \storage, \commits, \decodes)$ where:
\begin{inparaenum}[(i)]
\item $\Inst$ is a set of MIL microinstructions,
\item $\storage \in \Stores =\ANames \rightharpoonup \Val$ is a (partial) storage function from names to values recording
microinstructions' execution results,
\item $\commits \subseteq \ANames$ is a set of names of store operations that have been committed  to the memory subsystem, 
\item $\decodes \subseteq \ANames$ is a set of names of program counter store  operations that have been processed, causing the ISA instruction at the stored location to be fetched and decoded.
\end{inparaenum}

In the following we write $\subst{\storage}{t}{v}$ for substitution of value $v$ for
name $t$ in  store $\storage$. We use $\fdefined{f}{x}$
to represent that the partial function $f$ is defined on $x$, and
$\fundefined{f}{x}$ if not $\fdefined{f}{x}$. We write $\dom{f}$ for the
domain of a partial function $f$.
We also use  $\proj{f}{D}$ to represent the restriction of  function $f$ to  domain $D$.
The semantics of expressions is  $[e]: \Stores \rightharpoonup \Val$ and
is defined as expected. An expression  is undefined if at least one name is
undefined in a storage, i.e., $\fundefined{[e]}{\storage}
\Leftrightarrow \fn(e) \not \subseteq \dom{\storage}$.
  For $\newstate=(\Inst, \storage, \commits, \decodes)$ we use
  $[e]\newstate$ for $[e]\storage$,
  $\fundefined{\newstate}{t}$ for $\fundefined{\storage}{t}$, and
  $\inst \in \newstate$ for $\inst \in \Inst$.

\subsection{Microinstruction Lifecycle}\label{sec:ooo:cycle}
%

\newcommand{\statestyle}{
  \tikzstyle{every node}=[font=\scriptsize]
  \tikzset{pblock/.style = {circle, draw=black!50, top
      color=white,bottom color=white!20, align=center}}
  \tikzset{edge/.style = {->,> = latex'}}
}

\newcommand{\runempty}[3][xshift=.0cm]{
  \node[#1, name=#2, label=$#3$, draw, circle,  minimum width=0.03cm, fill=gray!50] {};
}
\newcommand{\runemptyplain}[3][xshift=.0cm]{
  \node[#1, name=#2, label=#3, draw, circle,  minimum width=0.03cm, fill=gray!50] {};
}

\newcommand{\runpred}[4][xshift=.0cm]{
  \node[#1, name=#2, label=$#3$, draw, circle, dotted,  minimum width=0.03cm] {$#4$};
}
\newcommand{\runpredplain}[4][xshift=.0cm]{
  \node[#1, name=#2, label=#3, draw, circle, dotted,  minimum width=0.03cm] {$#4$};
}

\newcommand{\runexec}[4][xshift=.0cm]{
  \node[#1, name=#2, label=$#3$, draw, circle, minimum width=0.03cm] {$#4$};
}
\newcommand{\runexecplain}[4][xshift=.0cm]{
  \node[#1, name=#2, label=#3, draw, circle, minimum width=0.03cm] {$#4$};
}
\newcommand{\runspec}[4][xshift=.0cm]{
  \node[#1, name=#2, label=$#3$, draw, circle, dashed, minimum width=0.03cm] {$#4$};
}
\newcommand{\runspecplain}[4][xshift=.0cm]{
  \node[#1, name=#2, label=#3, draw, circle, dashed, minimum width=0.03cm] {$#4$};
}
\newcommand{\runcomplete}[4][xshift=.0cm]{
\node[state,accepting,inner sep=1pt, minimum size=0.5cm] (#2) [#1] {$#4$};
}
\newcommand{\runcommit}[4][xshift=.0cm]{
  \node[#1, name=#2, label=#3, draw, circle, minimum width=0.03cm,
  very thick] {$#4$};
}
\newcommand{\rundecode}[4][xshift=.0cm]{
\node[state,accepting,inner sep=1pt, minimum size=0.5cm,label=#3] (#2) [#1] {$#4$};
}
\newcommand{\rundecodeplain}[4][xshift=.0cm]{
\node[state,accepting,inner sep=1pt, minimum size=0.5cm,label=#3] (#2) [#1] {$#4$};
}
\newcommand{\runspecdecodeplain}[4][xshift=.0cm]{
\node[state,accepting,dashed,inner sep=1pt, minimum size=0.5cm,label=#3] (#2) [#1] {$#4$};
}
\newcommand{\runspecdecode}[4][xshift=.0cm]{
\runspecdecodeplain{#1}{#2}{#3}{$#4$}
}

\newcommand{\specstate}{h}
\newcommand{\tracerightarrow}[2]{
  \node[below right =0.2cm and 0.05cm of #1, name=tr#1-a] {};
  \node[right =0.6cm of tr#1-a, name=tr#1-b] {};
  \draw[->] (tr#1-a) edge node{$#2$} (tr#1-b);
}
\newcommand{\traceleftarrow}[2]{
  \node[below left =0.2cm and -0.05cm of #1, name=tr#1-a] {};
  \node[left =0.6cm of tr#1-a, name=tr#1-b] {};
  \draw[->] (tr#1-a) edge node[above]{$#2$} (tr#1-b);
}
\newcommand{\tracedownarrow}[2]{
  \node[below =0.3cm of #1, name=tr#1-a] {};
  \node[below =0.45cm of tr#1-a, name=tr#1-b] {};
  \draw[->] (tr#1-a) edge node[right]{$#2$} (tr#1-b);
}

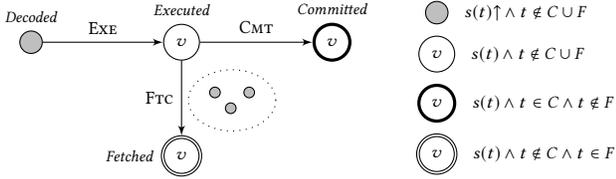
\begin{figure}
\begin{center}
\begin{tikzpicture}[node distance=2cm,auto,>=latex']

\statestyle                      

\runempty{1}{\textit{Decoded}}
\runexec[right of = 1]{2}{\textit{Executed}}{v}
\runcommit[right of = 2]{3}{$\textit{Committed}$}{v}
\rundecode[below =1cm of 2]{4}{left:$\textit{Fetched}$}{v}
\draw[->] (1) edge node{$\exe$} (2);
\draw[->] (2) edge node{$\cmt$} (3);
\draw[->] (2) edge node[left]{$\ftc$} (4);

\node[above right =0.6 cm and 0.2cm of 4, name=11, draw, circle,  inner sep=0.05cm, minimum width=0.03cm, fill=gray!50] {};
\node[right =0.3cm of 11, name=12, draw, circle,  inner sep=0.05cm, minimum width=0.03cm, fill=gray!50] {};
\node[below right =0.1cm and 0.1cm of 11, name=13, draw, circle,  inner sep=0.05cm, minimum width=0.03cm, fill=gray!50] {};

\node[draw,ellipse,dotted,fit = (11)(12)(13)] {};

\runempty[above right =0.1cm and 1.1cm of 3]{l1}{}
\node[right =0.1cm of l1, name=l2] {$\fundefined{\storage}{t} \land t \not \in \commits \cup \decodes$};
\runexec[below =0.15cm of l1]{l3}{}{v}
\node[right =0.1cm of l3, name=l4]{$\storage(t) \land t \not \in \commits \cup \decodes$};
\runcommit[below =0.15cm of l3]{l5}{}{v}
\node[right =0.1cm of l5, name=l6]{$\storage(t) \land t \in \commits \land t \not \in \decodes$};
\rundecode[below =0.15cm of l5]{l7}{}{v}
\node[right =0.1cm of l7, name=l8]{$\storage(t) \land t \not \in \commits \land t \in \decodes$};
\end{tikzpicture}
\end{center}
\caption{
  OoO semantics: Microinstruction lifecycle
}
\label{fig:ooo}
\end{figure}


\todo{changed picture}
Figure \ref{fig:ooo} represents the microinstruction lifecycle in the
OoO semantics.
For a given state $(\Inst, \storage, \commits, \decodes)$, a
microinstruction $\inst =\ass{c}{t}{o} \in \Inst$ can be in one of four different
states.
A microinstruction  in state \text{Decoded} (represented by a gray circle) has not been
executed, committed or fetched ($\fundefined{s}{t}$, $t \not \in
\commits$, $t \not \in \decodes$), and its guard is either true ($[c]s$)
or undefined ($\fundefined{[c]}{s}$). If the guard is false, i.e,
$\neg[c]s$, the instruction is considered as \text{Discarded} (not shown). 
A microinstruction is able to move to state \text{Executed} (represented by a simple circle whose
content is $\storage(t)$)
if its guard evaluates to true 
and all dependencies have been executed. 
Subsequently, an \text{Executed} 
store microinstruction can either be committed to the memory subsystem
(\text{Committed}: $t \in \commits$, represented by a bold circle),
or, if it is a program counter store, assign the program counter, causing a new ISA
instruction to be fetched and decoded (\text{Fetched}: $t \in \decodes$, represented by a double
circle).
Accordingly, the transition of a program counter store to
state \text{Fetched} leads to the spawn of a collection of newly decoded microinstructions
(i.e., the translation of the subsequent ISA instruction) in state
\text{Decoded}.
The labels of the edges in the diagram correspond to the names of the
transition rules of Section~\ref{sec:ooo:semantics}.

\subsection{Semantics of Single Microinstructions}
The semantics is defined in two steps: we first define the 
semantics of single microinstructions, then  introduce the operational semantics
of MIL programs.
The semantics of a microinstruction $\den{\inst} : \conf \rightarrow (\Val
\times \Observations) \cup \{\perp\}$ returns either a value and an
observation, \changed{or $\perp$ if the microintruction cannot be executed}. 

\noindent\textbf{(Internal operations)}
The semantics of internal operations is straightforward:
\begin{newnotation}
\[\den{\ass{c}{t}{e}}\newstate = \left \{
      \begin{array}{ll}
        ([e]\newstate, \lblempty) & \mbox{if } \fdefined{[e]}{\newstate} \\
        \perp & \mbox{otherwise}
      \end{array}
    \right .
  \]
\end{newnotation}
  \newcommand{\emptystore}{\emptyset}
An internal operation can be executed as soon as its
dependencies are available.
In Example~\ref{example:1}, the semantics of internal operation $t_1$ is defined for the
empty storage $\emptystore$, since it does not refer to any names. However, the semantics of $t_3$ is undefined in $\emptystore$, since
it depends on the value of $t_2$ that is not available in $\emptystore$.

\noindent\textbf{(Store operations)}
The semantics of store operations is defined as follows:
\begin{newnotation}
\[\den{\ass{c}{t}{\store{\restype}{t_a}{t_v}}}\newstate =
    \left \{
      \begin{array}{ll}
        (\den{t_v}\newstate, \lblempty) & \mbox{if } \fdefined{\den{t_v}}{\newstate} \land \fdefined{\den{t_a}}{\newstate}\\
        \perp & \mbox{otherwise}
      \end{array}
    \right .\]
\end{newnotation}
A resource update can be executed as soon
as both the address of the resource and the value are available. Observe that this rule models the \emph{internal} execution of a
resource update and not its commit to the memory subsystem. 
These internal updates are not observable by a programmer, therefore there is no restriction on their execution order.
As an example, the ISA program
\stylecode{$r_1$:= 0; $r_2$:= $r_1$; $r_1$:= 1} can be implemented by the
following microinstructions:
%

\begin{center}
\begin{tikzpicture}[node distance=4cm,auto,>=latex']

  \pomsetstyle
  
\node[pblock,label={$t_{11}$},name=a1]{$\store{\Regs}{r_1}{0}$};                      

\node[pblock,label={[xshift=-.2cm]$t_{21}$},name=a2,  right  = 1cm of a1]{$\load{\Regs}{r_1}$};                      

\node[pblock,label={[xshift=-.2cm]$t_{22}$},name=a3, below = \pomsetnewline of a2 ]{$\store{\Regs}{r_2}{t_{21}}$};                      

\node[pblock,label={$t_{31}$},name=a4,  right  = 1cm of a2]{$\store{\Regs}{r_1}{1}$};



\draw[->]  (a2) edge (a3) 
;

\end{tikzpicture}
\end{center}
%

The semantics of  $t_{11}$, i.e., $\store{\Regs}{r_1}{0}$,  and $t_{31}$, i.e., $\store{\Regs}{r_1}{1}$, is defined in $\emptystore$, and yields
$(0,\lblempty)$ and $(1,\lblempty)$, respectively. As we will see,
the operational semantics is in charge of ordering resource updates  to
preserve consistency and dependencies.

\newcommand{\strfor}{\textit{str-may}}
\newcommand{\strcert}{\textit{str-crt}}
\newcommand{\stract}{\textit{str-act}}

\noindent\textbf{(Load operations)}
While the semantics of  internal operations and store operations only depends on 
the execution of their operands, load operations may depend on past store operations. This requires identifying the previous resource update that determines 
the correct value to be loaded.  We use the following definitions to compute the set of store operations that may affect a load operation.

\begin{definition}
Consider a load operation $\ass{c}{t}{\load{\restype}{t_a}}$
\begin{itemize}
\item
    $\strfor(\newstate,t) = \{
    \ass{c'}{t'}{\store{\restype}{t'_a}{t'_v}} \in \newstate \mid
    t' < t \land
(\den{c'}{\newstate} \vee \fundefined{\den{c'}}{\newstate})
\land
(\den{t'_a}{\newstate} = \den{t_a}{\newstate} \vee
\fundefined{\newstate}{t'_a}{} \vee
\fundefined{\newstate}{t_a})
\}$
 is the set of stores
 that \emph{may} affect the load address of $t$ in state $\newstate$.
\item
   $\stract(\newstate, t) = \{
\ass{c'}{t'}{\store{\restype}{t'_a}{t'_v}} \in \strfor(\newstate, t) \mid\  \neg \exists \ass{c''}{t''}{\store{\restype}{t''_a}{t''_v}} \in \strfor(\newstate, t)
.\
t'' > t' \land
\den{c''}{\newstate} \land \den{t''_a}{\newstate} \in \{\den{t_a}{\newstate}, \den{t'_a}{\newstate}\}
\}$ 
 is the set of \emph{active} stores. 
\end{itemize}
\end{definition}
The stores that may affect the address of $t$ are the stores that: $(i)$ have not been discarded, namely they can be executed ($[c]\newstate$) or may be executed ($ \fundefined{c}{\newstate}$), and $(ii)$ the store address in $t'_a$ may result 
in the same address as the load address in $t_a$, namely either they both evaluate to the same address ($\newstate(t'_a) = \newstate(t_a)$), or the store address is unknown ($\fundefined{\newstate}{t'_a}$), or the load
address is unknown ($\fundefined{\newstate}{t_a}$).

The active stores of $t$ are  the stores that may affect
the load address computed by $t_a$, and, there  are no subsequent stores $t''$ on the same address as the load address in $t_a$, or on the same address as the store address in $t'_a$.
This set determines the ``minimal'' set of store operations that may
affect a load operation from address $t_a$.

The definitions of $\stract(\newstate, t)$ and $\strfor(\newstate,t)$ are
   naturally extended to stores
   $\ass{c}{t}{\store{\restype}{t_a}{t_v}}$.
   These definitions allow us to define the semantics of loads:
   \begin{newnotation}
\[
  \begin{array}{l}
    \den{\ass{c}{t}{\load{\restype}{t_a}}}\newstate  =\\
    \qquad
    \left \{
      \begin{array}{ll}
        (\den{t_s}{\newstate}, l) & \mbox{if } \bn(\stract(\newstate, t)) =
                                                    \{t_s\} \land \\
        & \fdefined{\newstate}{t_a} \land 
          \fdefined{\newstate}{t_s}
        \\
        \perp & \mbox{otherwise}
      \end{array}
                \right .
                \\
    \mbox{where }l = \\
    \qquad
                \left \{
                \begin{array}{ll}
                  \dl{\newstate(t_a)} & \mbox {if } t_s\in\commits \wedge \restype = \Locs \\
                  \lblempty & \mbox{otherwise}
                \end{array}
                              \right .
      \end{array}
    \]
    \end{newnotation}
A load operation can be executed if the set of active stores consists of a singleton set with bound name  $t_s$, i.e., the store causing $t_a$ to be assigned is uniquely determined, 
and both the address $t_a$ of the load  and the address $t_s$ of the store 
can be evaluated in state $\newstate$.

Note that the semantics allows forwarding the result of a store
to another microinstruction before it is committed to 
memory. In fact, if the
active store is yet to be committed to
memory, i.e., $t_s\notin\commits$, it is possible for the
store to forward its data to the load, without causing an interaction
with the memory subsystem (i.e., $l=\lblempty$).
Otherwise, the load yields an observation of
a data load from address $\newstate(t_a)$.
%
%
    


\begin{pomset}
  \begin{center}
\begin{tikzpicture}[node distance=4cm,auto,>=latex']

\pomsetstyle
  

\node[pblock,label={$t_{11}$},name=a1]{$1$};

\node[pblock,label={$t_{12}$},name=a2, right of = a1]{$\store{\Locs}{t_{11}}{1}$};

\node[pblock,label={[xshift=-.2cm]$t_{21}$},name=a3,  below  = \pomsetnewline of a1]{$0$};                      

\node[pblock,label={[xshift=-.2cm]$t_{22}$},name=a4, right of=a3 ]{$\store{\Locs}{t_{21}}{2}$};                      

\node[pblock,label={[xshift=-.2cm]$t_{31}$},name=a5,  below  = \pomsetnewline of a3]{$1$};                      

\node[pblock,label={[xshift=-.2cm]$t_{32}$},name=a6, right of=a5 ]{$\store{\Locs}{t_{31}}{3}$}; 

\node[pblock,label={[xshift=-.2cm]$t_{41}$},name=a7,  below  = \pomsetnewline of a5]{$1$};                      

\node[pblock,label={[yshift=-0.4cm, xshift=1.0cm]$t_{42}$},name=a8, right of=a7 ]{$\load{\Locs}{t_{41}}$}; 

\draw[->]  (a1) edge (a2) (a3) edge (a4) (a5) edge (a6) (a7) edge (a8);

\node[draw,rectangle,minimum width=3.2cm,minimum height=3cm,fit = (a2)(a6)] {};
\node[draw,dashed,minimum width=2.5cm,minimum height=1.8cm,rectangle,fit = (a4)(a6), yshift=0.07cm] {};
\node[draw,dotted,rectangle,minimum width=1.5cm,minimum height=.85cm,fit = (a6), yshift=0.15cm] {};

\end{tikzpicture}
  \end{center}
 \caption{\stylecode{*(1):=1; *(0):=2; *(1):=3; *(1);}}
   \label{example:active-store}
\end{pomset}


Example~\ref{example:active-store} illustrates the semantics of loads.
The program writes $1$ into address $1$, then writes $2$ in $0$,
overwrites address $1$ with $3$, and finally loads from address $1$.
We use active stores to dynamically compute the dependencies of
load operations. Let $\newstate_0$
be a state containing microinstructions as in the example, and having empty storage.
For this state, the active store for the load  $t_{42}$, i.e., $\stract(\newstate_0, t_{42})$, consists of all stores of the example, as depicted by the solid rectangle. 
Since none of microinstructions that compute the addresses  have 
been executed, the address $t_{41}$ of the load is unknown, hence, we cannot exclude any store from affecting the
address that will be used by $t_{42}$. Therefore, the load 
cannot be executed in $\newstate_0$. This set of active stores will
shrink during execution as more information becomes available through the storage.  

Let the storage of $\newstate_1$ be $\{t_{11} \mapsto 1; t_{31}
\mapsto 1\}$, i.e., the result of executing $t_{11}$ and $t_{31}$. The active stores $\stract(\newstate_1, t_{42})$
consist of microinstructions depicted by the dashed rectangle. Observe that the store $t_{12}$ is in $\strfor(\newstate_1, t_{42})$,
however there exists a subsequent store, namely $t_{32}$, that
overwrites the effects of $t_{12}$ on the same memory address.
Therefore, $t_{12}$ is no longer an active store and it can safely be discarded.

Let the storage of $\newstate_2$ be $\{t_{11} \mapsto 1; t_{31} \mapsto 1, t_{41} \mapsto
1\}$, i.e., the result of executing  $t_{11}, t_{31}$ and $t_{41}$. The active stores
$\stract(\newstate_2, t_{42})$ now consist of the singleton set $\{t_{32}\}$ as depicted by the dotted rectangle. This is because the address $t_{41}$ 
of the load can be computed in state $\newstate_2$.  
Although $t_{22}$ is still in $\strfor(\newstate_2, t_{42})$, there is a subsequent store, $t_{32}$, that will certainly affect
the address of the load. Therefore, $t_{22}$ is no longer
an active store.

Finally, let the storage of $\newstate_3$ be $\{t_{11} \mapsto 1; t_{31} \mapsto 1, t_{41} \mapsto
1, t_{32} \mapsto 3\}$, i.e., the result of executing
$t_{11}, t_{31}$, $t_{41}$, and  $t_{32}$.
Once $\stract$ has been reduced to a singleton set ($\{t_{32}\}$),
and the active-store has been executed ($\fdefined{\newstate_3}{t_{32}}$),
the semantics of the load is defined. This yields
the same value as  the store  in $t_{32}$. If the store  $t_{32}$ has been committed to memory, the execution of the load yields the observation $\dl{1}$.



\subsection{Operational Semantics}
\label{sec:ooo:semantics}
We can now define the microinstructions' transition relation  $\newstate
\doublestep{l}{}\newstate'$, implementing the lifecycle of Section~\ref{sec:ooo:cycle}.

\noindent \textbf{(Execute)} A microinstruction can be executed if it hasn't already been executed
($\fundefined{\storage}{t}$),  the guard holds ($\den{c}{\storage}$), and the
dependencies have been resolved ($\fdefined{\den{\inst}}{\storage}$): 
\begin{newnotation}
  \[
    \begin{array}{cc}
(\exe)
&
\begin{array}{c}
  \inst =  \ass{c}{t}{o}  \in \Inst \qquad \fundefined{\storage}{t} \qquad
\den{c}{\storage} \qquad \den{\inst}\newstate = (v, l)
\\ \hline
\newstate = (\Inst, \storage, \commits, \decodes) \doublestep{l}{} (\Inst,
\subst{\storage}{t}{v}, \commits, \decodes)
\end{array}\\
\end{array}
\]
\end{newnotation}
Observe that if $\inst$ is a load from the memory subsystem, the rule
can produce the observation of a data load. 

\noindent \textbf{(Commit)} Once a memory store has been executed ($\fdefined{\storage}{t}$), it can be
committed to memory, yielding an observation. The rule  ensures that stores can only be committed once ($t \not \in
\commits$) and that stores on the same
address are committed in  program order, by checking that all past
stores are in $\commits$, i.e., $\bn(\strfor(\newstate, t))
\subseteq \commits$.
\begin{newnotation}
\[
    \begin{array}{cc}
(\cmt)
&
\begin{array}{c}
  \ass{c}{t}{\store{\Locs}{t_a}{t_v}} \in \Inst \qquad
  \fdefined{\storage}{t} \qquad t \not \in \commits \\ 
  \bn(\strfor(\newstate, t)) \subseteq \commits \\ \hline
\newstate = (\Inst, \storage, \commits, \decodes) \doublestep{\ds{s(t_a)}{}}{×} (\Inst,
  \storage, \commits \cup \{t\}, \decodes)
\end{array}
\end{array}
\]
\end{newnotation}
In summary, stores can be executed internally in any order, however, they are
committed in order.
In  Example~\ref{example:active-store}, if $\newstate$ has storage $\storage =
\{t_{11} \mapsto 1; t_{12} \mapsto 1;
t_{31} \mapsto 1; t_{32} \mapsto 3\}$ and commits $\commits = \emptyset$, then only $t_{12}$ can be committed,
since $t_{22}$ has not been executed and $\bn(\strfor(\newstate, t_{32})) \not \subseteq \commits$.
Notice that $t_{22}$ is in the may stores since its address has not been
resolved. Therefore, $t_{32}$ can be committed only after $t_{12}$
has been committed and $t_{21}$ has been executed. However, the commit
of $t_{32}$ does not have to wait for the commit or execution of
$t_{22}$. In fact, if
$\newstate'$ has storage $\storage' = \storage \cup \{t_{21} \mapsto 0\}$
then $\bn(\strfor(\newstate', t_{32})) = \{t_{12}\}$. That is, the order of store commits is only enforced
per location, as expected.

\noindent \textbf{(Fetch-Decode)} A program counter store enables the fetching and decoding (i.e., translating) of
a new ISA instruction. The rule for fetching is similar to the rule for
commit, since instructions are fetched in order. The set
$\decodes$ keeps track of program counter updates whose resulting
instruction has been fetched and ensures that 
instructions are not fetched or decoded twice. Fetching the result of
a program counter update yields
the observation of an instruction load from address $a$.
\begin{newnotation}
\[
    \begin{array}{cc}
(\ftc)
&
  \begin{array}{c} 
    \ass{c}{t}{\store{\Pc}{}{t_v}} \in \Inst \qquad
    \storage(t) = a \qquad
    t \not \in \decodes
    \\
    \bn(\strfor(\newstate, t)) \subseteq \decodes
    \\ \hline
\newstate = (\Inst, \storage, \commits, \decodes) \doublestep{\il{a}}{} (\Inst
    \cup \Inst', \storage, \commits, \decodes \cup
    \{t\})
    \\
    \mbox{where }\Inst' = \translate(a, max(\Inst))
  \end{array}
      \end{array}
\]
\end{newnotation}
Write $max(\Inst)$ for the largest name $t$ in  $\Inst$ and
$\translate(a, max(\Inst))$ for the translation of the instruction at address $a$, ensuring that the names of the
microinstructions thus generated are greater than $max(\Inst)$.

\noindent \textbf{(Remarks on OoO semantics)}
\mathchardef\mhyphen="2D
\newcommand{\stepparam}{\mathit{step\mhyphen param}}
The three rules of the semantics reflect the atomicity of
MIL microinstructions: A transition can affect a single microinstruction by
either assigning a value to the storage, extending the set of commits,
or extending the set of fetches.
In the following, we use $\stepparam(\newstate, \newstate') = (\alpha, t)$
to identify the rule $\alpha \in \{\exe, \cmt(a,v), \ftc(\Inst)\}$
that enables $\newstate \doublestep{}{} \newstate'$ and the name $t$ of the
affected microinstruction. In case of commits we also extract the
modified address $a$ and the saved value $v$, in case of fetches we
extract the newly decoded microinstructions $\Inst$.
The semantics preserves several invariants:
Let $(\Inst, \storage, \commits, \decodes) = \newstate$ 
if $\alpha = \exe$ then $\fundefined{\newstate}{t}$;
if $\alpha = \cmt(a,v)$ then $t \not \in \commits$ and free names
(i.e., address and value) of the corresponding microinstruction are defined
in $\storage$;
if $\alpha = \ftc(\Inst')$ then $t \not \in \decodes$;
all state components are monotonic.

\noindent \textbf{(Initial state)}
In order to bootstrap the computation, we assume that
the set of microinstructions of the initial state contains
one store for each memory address and register, the value
of these stores is the initial value of the corresponding resource,
and that these stores are in the storage and commits of the
initial state.


\newcommand{\spec}{\delta}
\newcommand{\guesses}{P}
\newcommand{\newspec}[4]{EXT(#1, #2, #3, #4)}
\newcommand{\specPars}[5]{SPAS(#1, #2, #3, #4, #5)}

\section{Speculative Semantics}\label{sec:spec}

We now extend the OoO semantics to support speculation. 
We add two new components to the states: a set of names $\guesses \subseteq n(\Inst)$ 
whose values have been predicted as result of speculation, and a partial
 function $\spec:\ANames \rightharpoonup S$ 
recording, for each name $t$, the storage
dependencies at time of execution of the microinstruction identified by $t$. 
\changed{Therefore, a state in the speculative semantics is a tuple $h = (\Inst, \storage, \commits, \decodes, \spec, \guesses)$ where $\sigma = (\Inst, \storage, \commits, \decodes)$
is the corresponding state in the OoO semantics. Abusing notation we write $(\sigma, \spec, \guesses)$ to denote a state in the specutative semantics, 
and  use $\specstate, \specstate_1, \dots$  to range over these states.} 
Informally, $\spec(t)$ is a \emph{snapshot} of the storage that affects the value of $t$ due to speculative predictions.
As we will see, these \changed{snapshots} are needed in order to match speculative  states with non-speculative  states, and to
restore the state of the execution in case of misspeculation.


\newcommand{\deps}{\mathit{deps}}
\newcommand{\depsI}{\mathit{depsX}}
\newcommand{\asn}{\mathit{asn}}
\newcommand{\sources}{\mathit{srcs}}

\subsection{Managing Microinstruction Dependencies}
The execution of a microinstruction may depend on local (intra-)
instruction dependencies, the names appearing freely in a
microinstruction, as well as cross (inter-) instruction
dependencies, caused by memory or register loads.
\begin{definition}
Let $\ass{c}{t}{o} \in \newstate$.
The dependencies of $t$ in $\newstate$ are $$\deps(t,\newstate) =
\fn(\ass{c}{t}{o})\cup\depsI(t,\newstate)$$
where the cross-instruction
dependencies are defined as
\[
\depsI(t,\newstate) = \left\{\begin{array}{ll}
\emptyset, & \mbox{if }t\mbox{ is not a load} \\
\asn(\newstate, t)\cup\sources(\newstate, t), & \mbox{otherwise}.\
\end{array}
\right.
\]
Cross-dependencies are nonempty only for loads
and consist of the names of active stores affecting $t$ in state $\newstate$,
$\asn(\newstate, t) = \bn(\stract(\newstate,t))$, plus, 
the names of stores potentially intervening between the earliest
active store and $t$ (we call $\sources(\newstate, t)$ the \emph{potential sources of}
$t$), which are defined as 
\[
\sources(\newstate, t)  =
\begin{array}{rcl}
\bigcup \{ \fn(c'), \{t'_a\} & | &
                             \mathit{min}(\asn(\newstate, t))\leq t' <
                             t, \\
    & &                    \ass{c'}{t'}{\store{\restype}{t'_a}{t'_v}}
                                    \in \newstate \}
                                    \end{array}
\]
\end{definition}

Intuitively, a load depends on the execution of active stores
that may affect the address of that load. Moreover, the fact that a
name $t^\ast$ is in the set of active stores $asn$ depends on the addresses and guards of all
stores between $t^\ast$ and $t$. This is because their values will determine the actual 
store that affects the address of the load  $t$. Thanks to our ordering relation $<$ between names, we can use  
the minimum name $min(asn)$ in  $asn$ to compute all stores between any name in $asn$ and $t$, thus
extracting the free names of their guards and addresses. 

The following figure illustrates  dependencies of the load  from Example~\ref{example:active-store}:
\begin{center}
\begin{tikzpicture}[node distance=4cm,auto,>=latex']

  \pomsetstyle
  

\node[pblock,label={$t_{11}$},name=a1]{$1$};

\node[pblock,label={[xshift=1.2cm]$t_{12}$},name=a2, right of = a1]{$\store{\Locs}{t_{11}}{1}$};

\node[pblock,label={[xshift=-.2cm]$t_{21}$},name=a3,  below  = \pomsetnewline of a1]{$0$};                      

\node[pblock,label={[xshift=1.2cm]$t_{22}$},name=a4, right of=a3 ]{$\store{\Locs}{t_{21}}{2}$};                      

\node[pblock,label={[xshift=-.2cm]$t_{31}$},name=a5,  below  = \pomsetnewline of a3]{$1$};                      

\node[pblock,label={[xshift=1.2cm]$t_{32}$},name=a6, right of=a5 ]{$\store{\Locs}{t_{31}}{3}$}; 

\node[pblock,label={[xshift=-.4cm, yshift=-.1cm]$t_{41}$},name=a7,  below  = \pomsetnewline of a5]{$1$};                      

\node[pblock,label={[xshift=1.2cm]$t_{42}$},name=a8, right of=a7 ]{$\load{\Locs}{t_{41}}$}; 

\draw[->]  (a1) edge (a2) (a3) edge (a4) (a5) edge (a6) (a7) edge (a8);

\node[draw,ellipse,minimum width=2.5cm,fit = (a2)] {};
\node[draw,ellipse,minimum width=2.5cm,fit = (a6)] {};
\node[draw,dashed,rectangle,minimum width=2cm,minimum height=2.8cm,fit = (a1)(a3)(a5)] {};
\node[draw,dotted,ellipse,minimum width=2cm,fit = (a7)] {};
\end{tikzpicture}
\end{center}
If $\storage = \{t_{11} \mapsto 1; t_{21} \mapsto 0; t_{41} \mapsto 1\}$
%
then the set of active stores names $asn$ for $t_{42}$ is $\bn(\stract(\newstate,t_{42})) = \{t_{12}, t_{32}\}$,
as depicted by the solid ellipses. In particular, $min(asn)=t_{12}$. We consider all stores between $t_{12}$ and the load $t_{42}$ (i.e., $t_{12}, t_{22}$, and $t_{32}$),
and add to the set of cross-dependencies the names in their guards and addresses, namely $t_{11}, t_{21}$ and $t_{31}$, as depicted by the dashed rectangle.
Observe that $t_{21}$ is in the set of cross-dependencies, although  $t_{22}$ is not an active store. This is because membership of $t_{12}$  in the 
active stores' set depends on the address $t_{21}$ being set to $0$, i.e., $\storage(t_{21}) = 0$. Therefore, the set of cross-dependencies $\depsI(t_{42},\newstate) = \{t_{12}, t_{32}, t_{11}, t_{21}, t_{31}\}$. 
Finally, the local dependencies of the load $t_{42}$ consist of its parameter $t_{41}$ (the dotted ellipsis), such
that $\deps(t_{42},\newstate) = \{t_{12}, t_{32}, t_{11}, t_{21}, t_{31},t_{42}\}$.

We verify that the dependencies $\deps$ are computed correctly.

  \begin{definition}[$t$-equivalence]
    Let $\newstate_1$ and $\newstate_2$ be states with storage
    $\storage_1$ and $\storage_2$, and
    $\inst_1$ and $\inst_2$ be the microinstructions identified by $t$.
    Then  $\newstate_1$ and $\newstate_2$ are $t$-equivalent,
    $\newstate_1 \sim_t \newstate_2$,  if $\inst_1 = \inst_2$,
      $\proj{\storage_1}{\mathit{fn}(\inst_1)} =
      \proj{\storage_2}{\mathit{fn}(\inst_2)}$,
      and if $t$'s microinstruction is a load  with dependencies
  $\mbox{ } T_{i} = \deps(t, \newstate_i)$ and active stores $\mbox{ }$
  $SA_i = \stract(\proj{\newstate_i}{T_i}, t)$ for $i \in \{1,2\}$
  then 
  $SA_1 = SA_2$ and
  $
  \proj{\storage_1}{SA_1} =
  \proj{\storage_2}{SA_2}.
  $
\end{definition}
Intuitively, $t$-equivalence states that, if the microinstruction named with $t$  depends (in the sense of $\deps$) in both states on the same active stores and these stores assign the same value to $t$, 
then the microinstruction has the same dependencies, it is enabled, and it produces the same result in both states. 

We use three possible states of the example above to illustrate
$t$-equivalence: $\newstate_1$ is a state reachable in the OoO
semantics, $\newstate_2$ and $\newstate_3$ may result from misspeculating the value of $t_{31}$ to be
$0$ and $5$ respectively.
\begin{center}
\begin{tikzpicture}[node distance=1cm,auto,>=latex']

\statestyle                      

\runexecplain                   {1-11}{left:$t_{11}$}{1}
\runexecplain [right of = 1-11] {1-12}{left:$t_{12}$}{1}
\runexecplain [below = 0.1cm of 1-11] {1-21}{left:$t_{21}$}{0}
\runexecplain [right of = 1-21] {1-22}{left:$t_{22}$}{2}
\runexecplain[below = 0.1cm of 1-21] {1-31}{left:$t_{31}$}{1}
\runexecplain[right of = 1-31] {1-32}{left:$t_{32}$}{3}
\runexecplain [below = 0.1cm of 1-31] {1-41}{left:$t_{41}$}{1}
\runemptyplain[right of = 1-41] {1-42}{left:$t_{42}$}

\node    [below right =0.1cm and 0.15cm of 1-41, name=l1] {$\newstate_1$};
\node    [below right =0.0cm and 0.15cm of 1-22, name=leq1] {$\not \sim_{t_{42}}$};

\runexecplain  [right =1.5cm of 1-12]{2-11}{left:$t_{11}$}{1}
\runexecplain [right of = 2-11] {2-12}{left:$t_{12}$}{1}
\runexecplain [below = 0.1cm of 2-11] {2-21}{left:$t_{21}$}{0}
\runexecplain [right of = 2-21] {2-22}{left:$t_{22}$}{2}
\runexecplain[below = 0.1cm of 2-21] {2-31}{left:$t_{31}$}{0}
\runexecplain[right of = 2-31] {2-32}{left:$t_{32}$}{3}
\runexecplain [below = 0.1cm of 2-31] {2-41}{left:$t_{41}$}{1}
\runemptyplain[right of = 2-41] {2-42}{left:$t_{42}$}

\node    [below right =0.1cm and 0.15cm of 2-41, name=l2] {$\newstate_2$};
\node    [below right =0.0cm and 0.15cm of 2-22, name=leq2] {$\sim_{t_{42}}$};

\runexecplain  [right =1.5cm of 2-12]{3-11}{left:$t_{11}$}{1}
\runexecplain [right of = 3-11] {3-12}{left:$t_{12}$}{1}
\runexecplain [below = 0.1cm of 3-11] {3-21}{left:$t_{21}$}{0}
\runexecplain [right of = 3-21] {3-22}{left:$t_{22}$}{2}
\runexecplain[below = 0.1cm of 3-21] {3-31}{left:$t_{31}$}{5}
\runexecplain[right of = 3-31] {3-32}{left:$t_{32}$}{3}
\runexecplain [below = 0.1cm of 3-31] {3-41}{left:$t_{41}$}{1}
\runemptyplain[right of = 3-41] {3-42}{left:$t_{42}$}

\node    [below right =0.1cm and 0.15cm of 3-41, name=l2] {$\newstate_3$};

\end{tikzpicture}
\end{center}

The states $\newstate_1$ and $\newstate_2$ are not
$t_{42}$-equivalent.
In particular, 
$T_{1} = \deps(t_{42}, \newstate_1) = \{t_{31}, t_{32}\}$ (notice that
$t_{12}$ and $t_{22}$ are not in the dependencies because by we know
that $t_{31}\mapsto 1$ and $t_{41} \mapsto 1$) ,
$T_{2} = \deps(t_{42}, \newstate_2) = \{t_{11}, t_{21}, t_{31},
t_{12}\}$.
Notice that $\proj{\newstate_1}{T_1}$ and $\proj{\newstate_2}{T_2}$
contain all the information needed to evaluate the semantics of
$t_{42}$ in $\newstate_1$ and $\newstate_2$ respectively.
In this case $SA_1 = \stract(\proj{\newstate_1}{T_1}, t_{42}) = \{t_{32}\}$,
and $SA_2 = \stract(\proj{\newstate_2}{T_2}, t_{42}) = \{t_{12}\}$ hence
$SA_1 \neq SA_2$: the two states lead the load $t_{42}$ to take the
result produced by two different memory stores.

The states $\newstate_2$ and $\newstate_3$ are 
$t_{42}$-equivalent.
In fact, 
$T_{3} = \deps(t_{42}, \newstate_3) = \{t_{11}, t_{21}, t_{31},
t_{12}\}$ and $SA_3 = \stract(\proj{\newstate_3}{T_3}, t_{42}) =
\{t_{12}\}$.
Therefore, $SA_2 = SA_3$ and   $
  \proj{\storage_2}{SA_2} =
  \proj{\storage_3}{SA_3}
  $: The two states lead the load $t_{42}$ to take the
result produced by the same memory stores.

\begin{lemma}
\label{lem:deps-eq}
  If $\newstate_1  \sim_t  \newstate_2$ and $t$'s microinstruction  
  in $\newstate_1$ is $\instr = \ass{c}{t}{o}$, then
  $\deps(t, \newstate_1) = \deps(t, \newstate_2)$,
  $\den{c}\newstate_1 = \den{c}\newstate_2$,
  and if $\den{\inst}\newstate_1 = (v_1, l_1)$ and $\den{\inst}\newstate_2 =
  (v_2, l_2)$ then $v_1 = v_2$. 
\end{lemma}
\begin{proof}
See Appendix~\ref{proof:teq}.
\end{proof}

\subsection{Microinstruction Lifecycle}

\begin{figure}
\begin{center}
\begin{tikzpicture}[node distance=2.2cm,auto,>=latex']

\statestyle                      

\runempty{dec}{Decoded}{}{\textit{Decoded}}
\runpredplain[below =1.15cm of dec]{pred}{below:$\textit{Predicted}$}{v}
\runspec[right of = dec]{spec}{\textit{Speculated}}{v}
\runexec[right of = spec]{exec}{\textit{Retired}}{v}
\rundecode[below =1cm of exec]{fetch}{below:$\textit{Fetched}$}{v}
\runspecdecodeplain[below =1cm of spec]{spec-fetch}{below:$\textit{Speculatively fetched}$}{v}
\runcommit[right of = exec]{commit}{\textit{Committed}}{v}
\draw[->] (dec) edge node[left]{$\prd$} (pred);
\draw[->] (dec) edge node{$\exe$} (spec);
\draw[->] (pred) edge node[right, pos=0.07]{$\pexe$} (spec);
\draw[->] (spec) edge node{$\ret$} (exec);
\draw[->] (spec) edge[bend left] node[above]{$\rbk$} (dec);
\draw[->] (exec) edge node{$\ftc$} (fetch);
\draw[->] (exec) edge node{$\cmt$} (commit);
\draw[->] (spec-fetch) edge node{$\ret$} (fetch);
\draw[->] (spec-fetch) edge node[left,pos=0.05]{$\rbk$} (dec);
\draw[->] (spec) edge node{$\ftc$} (spec-fetch);

\runpred[below =1cm of pred]{l-pred}{}{v}
\node[right =0.1cm of l-pred, name=l-pred-1] {$t \in \guesses$};
\runspec[right of = l-pred-1]{l-spec}{}{v}
\node[right =0.1cm of l-spec, name=l-spec-1]{$\fdefined{\spec}{t} \land t \not \in \decodes$};
\runspecdecodeplain[right of = l-spec-1]{l-spec-fetch}{}{v}
\node[right =0.1cm of l-spec-fetch, name=l-spec-fetch-1] {$\fdefined{\spec}{t} \land t \in \decodes$};
\runexec[below =0.5cm of l-pred]{l-ret}{}{v}
\node[right =0.1cm of l-ret, name=l-ret-1] {$\fundefined{\spec}{t} \land t \not \in \decodes \cup \commits$};
\rundecode[below =0.5cm of l-spec]{l-fetch}{}{v}
\node[right =0.1cm of l-fetch, name=l-fetch-1] {$\fundefined{\spec}{t} \land t \in \decodes$};
\runcommit[below =0.5cm of l-spec-fetch]{l-commit}{}{v}
\node[right =0.1cm of l-commit, name=l-commit-1] {$\fundefined{\spec}{t} \land t \in \commits$};

\end{tikzpicture}
\end{center}
\caption{
  Speculative semantics: Microinstruction lifecycle
}
\label{fig:spec:lifecycle}
\end{figure}
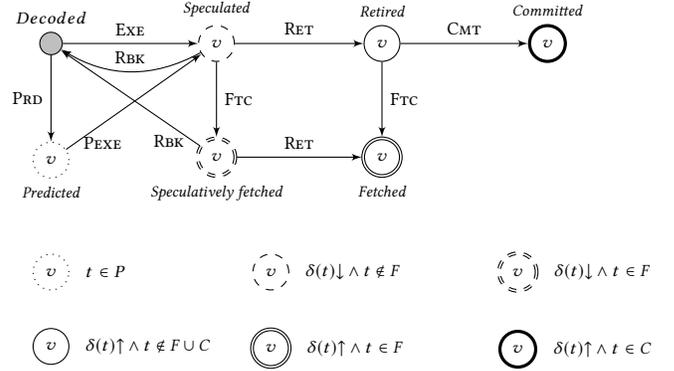

Figure~\ref{fig:spec:lifecycle} depicts the microinstruction lifecycle
under speculative execution. Compared to the OoO lifecycle of
Section~\ref{sec:ooo:cycle}, states \text{Decoded}, \text{Predicted}, \text{Speculated}, 
and \text{Speculatively Fetched}  correspond to state \text{Decoded}, state \text{Retired} corresponds to \text{Executed}, otherwise states
\text{Fetched} and \text{Committed} are the same. As depicted in the legend, transitions between states set different properties of a microinstruction's lifecycle, which we will model in the semantics. 

State \text{Predicted} (dotted circle) models microinstructions that have not yet been executed, but whose result values have been predicted. 
A \text{Decoded} microinstruction can transition to state \text{Predicted} by predicting its result value, thus recording that the value was predicted and causing the state of the microinstruction to be defined. 
A  microinstruction that is ready to be executed (in \text{Decoded}),
possibly relying on predicted values, can be executed and transition
to state \text{Speculated} (dashed circle),
recording its dependencies in the snapshot. Notice that state \text{Speculated} models both speculative and non-speculative execution of a microinstruction. 

From state \text{Speculated}, a microinstruction can: ($a$) roll back to \text{Decoded} (if the predicted values were wrong); ($b$) speculatively fetch the next ISA instruction to be executed, thus 
moving to state \text{Speculatively Fetched}, doubled dashed circle)
and generating newly decoded microinstructions; or ($c$) retire in
state \text{Retired} (single circle) if it no longer depends on speculated values.

Microinstructions in state \text{Speculatively Fetched} can either be
rolled back due to misspeculation, otherwise move to  state
\text{Fetched} (double circle). 
Finally, in state \text{Retired}, as in the OoO case, a PC store
microinstruction can be (non-speculatively) fetched and generate newly
decoded microinstructions, or, if it is a memory store, it can be committed
to the memory subsystem (bold circle).

\subsection{Microinstruction Semantics}

We now present a speculative semantics, denoted by the transition
relation
$(\newstate,\spec,\guesses)\triplestep{}{}(\newstate',\spec',\guesses')$,
that reflects the microinstructions' lifecycle in
Figure~\ref{fig:spec:lifecycle}.
\changed{
  We illustrate the rules of our semantics using the graph in Example~\ref{example:spec:sem} and the interpretation of states (circles) in Figure~\ref{fig:spec:lifecycle}. 
 Additionally, for two microinstruction identifiers $t$ and $t'$ in speculative state $h=(\Inst, \storage, \commits, \decodes, \spec, \guesses)$, we draw an edge from $t$
  to $t'$ labeled with $v$ whenever $\spec(t)(t') = v$.
}

\noindent \textbf{(Predict)}
The semantics allows to predict the value of an internal operation
choosing a value  $v \in\Val$. The rule updates the storage and records the predicted name, while ensuring that the microinstruction has not been executed already.

\begin{newnotation}
\[
\begin{array}{cc}
(\prd)
&
\begin{array}{c}
  \ass{c}{t}{e} \in \Inst  \qquad
  \fundefined{\storage}{t} \qquad \spec' = \spec \cup \{ t
  \mapsto \emptyset\} \qquad
  \\
  \hline
  (\Inst, \storage, \commits, \decodes, \spec, \guesses) \triplestep{}{}
  (\Inst, \subst{\storage}{t}{v}, \commits, \decodes, \spec',
  \guesses \cup \{t\})
\end{array}
\end{array}
\]
\end{newnotation}
We remark that the semantics can predict a value only for an internal
operation ($\ass{c}{t}{e}$) that has not been already executed
($\fundefined{\storage}{t}$).
As we will see, this choice does not hinder 
expressiveness while it avoids the complexity in modeling speculative execution of program counter updates and loads.
\changed{
Concretely, the rule assigns an arbitrary value to the name of the predicted
microinstruction ($\subst{\storage}{t}{v}$)
and records that the result is speculated  ($\spec \cup \{ t
\mapsto \emptyset\}$). Observe that the snapshot $\spec'(t)$ is $\emptyset$ because the
prediction does not depend on the results of other microinstructions.
}

\changed{
Consider state $\specstate_0$ in Example~\ref{example:spec:sem} containing all microinstructions of
 our running program, which have
just been decoded (gray circles).
The CPU can predict that the value of arithmetic operation $t_2$ is $0$. 
Rule $\prd$ updates the storage with 
$t_2 \mapsto 0$ (dotted circle), the snapshot for
$t_2$ with an empty mapping, and adds $t_2$ to the prediction set.}

\newcommand{\upddep}{\mathit{dep\mhyphen save}}
\newcommand{\overbar}[1]{\mkern 1.7mu\overline{\mkern-1.7mu#1\mkern-1.7mu}\mkern 1.7mu}
\newcommand{\mapminus}[2]{#1 \setminus #2}

\noindent \textbf{(Execute)}
The  rules for execution, commit, and fetch reuse the OoO semantics. 
First for the case when the instruction has not been predicted already: 
\begin{newnotation}
\[
\begin{array}{cc}
(\exe)
&
  \begin{array}{c}
    \newstate \doublestep{l}{} \newstate' \quad
    \stepparam(\newstate,\newstate') = (\exe, t)
  \\ \hline
    (\newstate, \spec, \guesses) \triplestep{l}{} (\newstate', \spec \cup \{ t \mapsto \proj{\storage}{\deps(t, \newstate)}\}, \guesses)
\end{array}
\end{array}
\]
\end{newnotation}
\changed{
The rule executes a microinstruction $t$ using the OoO semantics and updates the snapshot $\spec$, recording that the 
execution of $t$ was determined by the value of its dependencies in  $\deps(t, \newstate)$ in storage $s$ of state $\newstate$.
Notice that the premise $\stepparam(\newstate,\newstate') = (\exe, t)$ ensures that
microinstruction $t$ has not been predicted. In fact, 
$\stepparam(\newstate,\newstate') = (\exe, t)$ only if
$\fundefined{\newstate}{t}$, while rule $\prd$ would update the storage with a value for 
 name $t$, hence $t \not \in \guesses$.
}

\changed{Consider now the state $\specstate_2$ resulting from
  the execution of $t_1$ and $t_3$ in Example~\ref{example:spec:sem}.
In $\specstate_2$ the CPU can execute the PC update $t_6$,
updating the storage with 
$t_6 \mapsto 36$.
The rule additionally updates the snapshot for $t_6$ with
the current values of its  dependencies, 
i.e., $\{t_2 \mapsto 0, t_3
      \mapsto 32\}$. Since the executed microinstruction $t_6$ is a store,
      its dependencies are the free names occurring in the microinstruction.
These snapshots are used by rules $\cmt$ and $\rbk$ to identify
mispredictions.
Similarly, the rule enables the execution of the memory store $t_4$ in
$\specstate_3$, which updates the storage with $t_4 \mapsto 1$ and 
 the snapshot for $t_4$ with  the values of its
dependencies $\{t_1 \mapsto 1\}$.}

\changed{The following rule enables the execution of microinstructions whose
  result has been previously predicted: }
\begin{newnotation}
\[
\begin{array}{cc}
(\pexe)
&
          \begin{array}{c}
            \newstate = (\Inst, \storage, \commits, \decodes) \quad t\in P \\
    (\Inst, \mapminus{\storage}{\{t\}}, \commits, \decodes)
    \doublestep{l}{}
    \newstate' \quad
    \stepparam(\newstate,\newstate') = (\exe, t)   
  \\ \hline
    (\newstate, \spec, \guesses) \triplestep{l}{} (\newstate', \spec \cup \{ t \mapsto \proj{\storage}{\deps(t, \newstate)}\}, \guesses \setminus \{t\})
\end{array}
\end{array}
\]
\end{newnotation}
\changed{
  The rule removes the value predicted for $t$ from the
  storage ($\mapminus{\storage}{\{t\}}$) to enable the actual execution
  of $t$ in the OoO semantics. It also removes $t$ from the set of predicted
  names $P$ and updates the snapshot with the new dependencies of $t$.   
}

\changed{In our example, rule $\pexe$ 
computes the actual value of $t_2$ in state $\specstate_4$, which was previously mispredicted as $0$.
The rule corrects the misprediction 
updating the storage with $t_2
\mapsto 1$ and the snapshot for $t_2$ with the values
of its dependencies, i.e., $t_1 \mapsto 1$.
Notice that in case of a misprediction,
the rule does not immediately roll back all other speculated microinstructions that are affected
by the mispredicted values, e.g., $t_6$.}

\noindent \textbf{(Commit)}
To commit a microinstruction it is sufficient to ensure that 
there are no dependencies left ($\fundefined{\spec}{t}$), i.e., the
microinstruction has been retired.
\changed{Since memory commits have observable side effects outside the
  processor pipeline, only retired
  memory stores can be sent to the memory subsystem.
}
\begin{newnotation}
\[
\begin{array}{cc}
(\cmt)
&
  \begin{array}{c}
    \newstate
    \doublestep{l}{}
    \newstate' \quad
    \stepparam(\newstate,\newstate') = (\cmt(a,v), t) \quad
     \fundefined{\spec}{t}
  \\ \hline
    (\newstate, \spec, \guesses) \triplestep{l}{} (\newstate', \spec, \guesses)
\end{array}
\end{array}
\]
\end{newnotation}
\changed{Consider the state $\specstate_4$
and the memory store $t_4$ in our example. Since $t_4$  has not been
retired (i.e., $\spec{(t_4)} = \{t_1 \mapsto 1\}$) it cannot be committed as
$\fdefined{\spec}{t_4}$. By contrast, the commit of $t_4$ is allowed in state $\specstate_8$ where $\fundefined{\spec}{t_4}$.
}

\noindent \textbf{(Fetch)}
Finally, for the case of (speculative or non-speculative) fetching,  the snapshot must be updated to record the dependency of the newly added microinstructions:
\begin{newnotation}
\[
\begin{array}{cc}
(\ftc)
&
  \begin{array}{c}
    \newstate
    \doublestep{l}{}
    \newstate' \qquad
    \stepparam(\newstate,\newstate') = (F(\Inst), t)
  \\ \hline
    (\newstate, \spec, \guesses) \triplestep{l}{} (\newstate', \spec \cup \{t' \mapsto \proj{\storage}{\{t\}} \: | \:  t' \in \Inst\}, \guesses)
\end{array}
\end{array}
\]
\end{newnotation}
\changed{
 Following the OoO semantics, if
  $\stepparam(\newstate,\newstate') = (F(\Inst), t)$ then $t$ is a
  PC update and $\storage(t)$ is the new value of the
  PC. For every newly added microinstruction in $t' \in \Inst$, we extend the snapshop $\spec$
  recording that $t'$ was added as result of updating the PC microinstruction $t$ with the value $\storage(t)$ 
  (formally, we project the storage $\storage$ on $t$, i.e., $\proj{\storage}{\{t\}}$). 
The new snapshop may be used later to roll back the newly added microinstructions in $I$ if the value of the  PC  is misspeculated.
}

\changed{For example, in state $\specstate_5$
the CPU can speculatively fetch the PC update $t_6$,
which sets the program counter to $36$.
Suppose that the newly added microinstructions in $\Inst$ (i.e., the microinstructions resulting from the
translation of the ISA instruction at address $36$) are  $t'_1$ and
$t'_2$. Following the OoO semantics, $\Inst$ is added to existing
microinstructions in $\newstate'$ . The rule 
additionally updates the snapshot for $t'_1$ and $t'_2$
recording the PC store that generated the new microinstructions, i.e., $t_6 \mapsto 36$.
}

\begin{pomset}
\begin{center}
\begin{tikzpicture}[node distance=2.5cm,auto,>=latex']

  \pomsetstyle

  \picloadT{t_1}{\Regs}{z}
  \picexpT[right of=t_1]{t_2}{t_1 = 1}
  \picloadT[right of=t_2]{t_3}{\Pc}{}
  \picstoreTT[below = \pomsetnewline of t_1]{t_4}{\Locs}{16}{t_1}
  \picstore[below = \pomsetnewline of t_2]{t_5}{t_2}{\Pc}{}{a}
  \picstore[below = \pomsetnewline of t_3]{t_6}{\neg t_2}{\Pc}{}{t_3+4}

 \draw[->]
 (t_1) edge (t_2)
 (t_2) edge (t_5)
 (t_2) edge (t_5) (t_3) edge (t_6)
 (t_2) edge (t_6)
 (t_1) edge (t_4) ;

\end{tikzpicture} 
\\[10pt]
\begin{tikzpicture}[node distance=0.8cm,auto,>=latex']

\statestyle                      

\runempty{11}{t_1}
\runempty[right of = 11]{12}{t_2}
\runempty[right of = 12]{13}{t_3}
\runempty[below of = 11]{14}{t_4}
\runempty[right of = 14]{15}{t_5}
\runempty[right of = 15]{16}{t_6}

\node    [below =0.1cm of 15, name=l1] {$\specstate_0$};

\tracerightarrow{13}{\prd, t_2}

\runempty[right =1.1cm of 13]{21}{}
\runpred[right of = 21]{22}{}{0}
\runempty[right of = 22]{23}{}
\runempty[below of = 21]{24}{}
\runempty[right of = 24]{25}{}
\runempty[right of = 25]{26}{}

\node    [below =0.1cm of 25, name=l2] {$\specstate_1$};
\tracerightarrow{23}{*}

\runexec[right  =1.1cm of 23]{31}{}{1}
\runpred[right of = 31]{32}{}{0}
\runexec[right of = 32]{33}{}{32}
\runempty[below of = 31]{34}{}
\runempty[right of = 34]{35}{}
\runempty[right of = 35]{36}{}

\node    [below =0.1cm of 35, name=l3] {$\specstate_2$};

\tracedownarrow{35}{\exe, t_6}

\runexec[below  =1.05cm of 34]{41}{}{1}
\runpred[right of = 41]{42}{}{0}
\runexec[right of = 42]{43}{}{32}
\runempty[below of = 41]{44}{}
\runempty[right of = 44]{45}{}
\runspec[right of = 45]{46}{}{36}

\node    [below =0.1cm of 45, name=l4] {$\specstate_3$};
\draw[->] (46) edge node[right]{$32$} (43);
\draw[->] (46) edge node[right,pos=.5]{$0$} (42);

\traceleftarrow{41}{\exe, t_4}

\runexec[below  =1cm of 24]{51}{}{1}
\runpred[right of = 51]{52}{}{0}
\runexec[right of = 52]{53}{}{32}
\runspec[below of = 51]{54}{}{1}
\runempty[right of = 54]{55}{}
\runspec[right of = 55]{56}{}{36}

\node    [below =0.1cm of 55, name=l5] {$\specstate_4$};
\draw[->] (56) edge node[right]{$32$} (53);
\draw[->] (56) edge node[right,pos=.5]{$0$} (52);
\draw[->] (54) edge node{$1$} (51);

\traceleftarrow{51}{\pexe, t_2}

\runexec[below  =1cm of 14]{61}{}{1}
\runexec[right of = 61]{62}{}{1}
\runexec[right of = 62]{63}{}{32}
\runspec[below of = 61]{64}{}{1}
\runempty[right of = 64]{65}{}
\runspec[right of = 65]{66}{}{36}

\node    [below =0.1cm of 65, name=l6] {$\specstate_5$};
\draw[->] (64) edge node{$1$} (61);
\draw[->] (62) edge node[above]{$1$} (61);
\draw[->] (66) edge node[right]{$32$} (63);
\draw[->] (66) edge node[right]{$0$} (62);

\tracedownarrow{65}{\ftc(\Inst), t_6};

\runexec[below  =1cm of 64]{71}{}{1}
\runexec[right of = 71]{72}{}{1}
\runexec[right of = 72]{73}{}{32}
\runspec[below of = 71]{74}{}{1}
\runempty[right of = 74]{75}{}
\runspecdecodeplain[right of = 75]{76}{}{36}
\runemptyplain[below of = 75]{77}{left:$t'_1$}
\runemptyplain[right of = 77]{78}{left:$t'_2$}

\node    [below =0.1cm of 77, name=l7] {$\specstate_6$};

\draw[->] (74) edge node{$1$} (71);
\draw[->] (72) edge node[above]{$1$} (71);
\draw[->] (76) edge node[right]{$32$} (73);
\draw[->] (76) edge node[right]{$0$} (72);
\draw[->] (77) edge node[right]{$36$} (76);
\draw[->] (78) edge node[right]{$36$} (76);

\tracerightarrow{73}{\ret, t_4}

\runexec[right  =1cm of 73]{81}{}{1}
\runexec[right of = 81]{82}{}{1}
\runexec[right of = 82]{83}{}{32}
\runexec[below of = 81]{84}{}{1}
\runempty[right of = 84]{85}{}
\runspecdecodeplain[right of = 85]{86}{}{36}
\runemptyplain[below of = 85]{87}{left:$t'_1$}
\runemptyplain[right of = 87]{88}{left:$t'_2$}

\draw[->] (82) edge node[above]{$1$} (81);
\draw[->] (86) edge node[right]{$32$} (83);
\draw[->] (86) edge node[right]{$0$} (82);
\draw[->] (87) edge node[right]{$36$} (86);
\draw[->] (88) edge node[right]{$36$} (86);

\node    [below =0.1cm of 87, name=l8] {$\specstate_7$};

\tracerightarrow{83}{\rbk, t_6}

\runexec[right  =1cm of 83]{91}{}{1}
\runexec[right of = 91]{92}{}{1}
\runexec[right of = 92]{93}{}{32}
\runexec[below of = 91]{94}{}{1}
\runempty[right of = 94]{95}{}
\runempty[right of = 95]{96}{}

\draw[->] (92) edge node[above]{$1$} (91);

\node    [below =0.1cm of 95, name=l9] {$\specstate_8$};

\tracedownarrow{95}{\cmt, t_4}

\runexec[below  =1 of 94]{101}{}{1}
\runexec[right of = 101]{102}{}{1}
\runexec[right of = 102]{103}{}{32}
\runcommit[below of = 101]{104}{}{1}
\runempty[right of = 104]{105}{}
\runempty[right of = 105]{106}{}

\draw[->] (102) edge node[above]{$1$} (101);

\node    [below =0.1cm of 105, name=l10] {$\specstate_9$};

\end{tikzpicture}
\end{center}
\caption{
  Execution trace of speculative semantics.
}
\label{example:spec:sem}
\end{pomset}

\noindent \textbf{(Retire)}
The following transition rule allows to retire a microinstruction in case of correct speculation:
\begin{newnotation}
\[
\begin{array}{cc}
(\ret)
&
\begin{array}{c}
  \fdefined{\storage}{t} \qquad
  \dom{\spec(t)} \cap \dom{\spec} = \emptyset \\
  (\Inst, \storage, \commits, \decodes) \sim_t (\Inst, \spec(t), \commits, \decodes) \qquad
  t \not\in \guesses
  \\
  \hline
  (\Inst, \storage, \commits, \decodes, \spec, \guesses) \triplestep{}{}
  (\Inst, \storage, \commits, \decodes, \mapminus{\spec}{\{t\}}, \guesses)
\end{array}
\end{array}
\]
\end{newnotation}
The map $\spec(t)$ contains the snapshot of $t$'s
dependencies at time of $t$'s execution.
A microinstruction  can be retired only if all its dependencies have been
retired ($\dom{\spec(t)} \cap \dom{\spec} = \emptyset$),
the microinstruction has been executed (i.e. its value has not been
just predicted $\fdefined{\storage}{t} \wedge t \not\in \guesses$), 
and the snapshot of $t$'s dependencies is $\sim_t$ equivalent with the
current state, hence the semantics of $t$ has been correctly
speculated (see Lemma~\ref{lem:deps-eq}).
Retiring a microinstruction results in removing the state of its
dependencies from $\spec$, as captured by $\mapminus{\spec}{\{t\}}$.

\changed{
  For instance, in state $\specstate_6$ the
  PC store $t_6$ cannot be retired for two reasons: one of its
  dependencies has not been retired (i.e., $\spec(t_6) = \{t_2 \mapsto
  0, t_3 \mapsto 32\}$ and $\spec{(t_2)} = \{t_1 \mapsto 1\}$, hence $\dom{\spec(t_6)} \cap \dom{\spec} = \{t_2\}$),
  and the snapshot for $t_6$ differs with respect to the
  storage (i.e., $\spec{(t_6)}(t_2) \neq \storage(t_2)$).
  Instead, the microinstruction $t_4$ can be retired because its
  dependencies (i.e., $t_1$) have been retired (i.e.,
  $\fundefined{\spec}{t_1}$) and the snapshot for $t_4$ (i.e., $t_1 \mapsto 1$) exactly matches the
  values in the storage. Notice that retiring $t_4$ would simply remove the
  mapping for $t_4$ from $\spec$.
}

Notice that in case of a load, $ (\Inst, \storage, \commits, \decodes) \sim_t (\Inst, \spec(t), \commits, \decodes)$ may hold even if some
dependencies of $t$ differ in $\storage$ and  $\spec(t)$. In fact, a load may have been executed as a result of misspeculating
the address of a previous store. In this case, $\sim_t$ implies that the misspeculation has not affected the
calculation of $\stract$ of the load (i.e., it does not cause a store bypass), hence there is no reason to 
re-execute the load.
This mechanism is demonstrated in examples later in this section.

\noindent \textbf{(Rollback)}
A microinstruction $t$ can be rolled back when it is found to transitively reference a value that was wrongly speculated.
This is determined by comparing $t$'s dependencies at execution time ($\spec(t)$) with the current storage assignment ($\storage$). 
In case of a discrepancy, if $t$ is not a program counter store, the
assignment to $t$ can simply be undone, leaving speculated
microinstructions $t'$ that reference $t$ to be rolled back later, if
necessary.
\newcommand{\shortsetminus}{\mbox{$\setminus$}}
\begin{newnotation}
\[
\begin{array}{cc}
(\rbk)
&
\begin{array}{c}
  t \not\in \guesses \qquad
  t \not \in \decodes \qquad
   (\Inst, \storage, \commits, \decodes) \not \sim_t (\Inst, \delta(t), \commits, \decodes)
  \\
  \hline
  (\Inst, \storage, \commits, \decodes, \spec, \guesses) \triplestep{}{}
  (\Inst,  \mapminus{\storage}{\{t\}},
  \commits, \decodes,\mapminus{\spec}{\{t\}},\guesses)
\end{array}
\end{array}
\]
\end{newnotation}
However, if $t$ \emph{is} a program counter store, the
speculative evaluation using rule $\ftc$ will have caused a new microinstruction to be speculatively fetched. This fetch needs to be undone.
To that end let $t'\prec t$ ($t'$ refers to $t$) if
$t\in\dom{\spec(t')}$, let $\prec^+$ be the  transitive closure of $\prec$.
As expected $\prec^+$ is antisymmetric and its the reflexive closure is a
partial order.
Define then the set $\Delta^+$ as $\{t' \mid t'
\prec^+ t\}$: i.e., 
$\Delta^+$ is the set of names that reference $t$, not including $t$
itself.
Finally, let $\Delta^* = \Delta^+ \cup \{t\}$.
\begin{newnotation}
\[
\begin{array}{cc}
(\rbk)
&
\begin{array}{c}
  t \not\in \guesses \qquad
  t \in \decodes \qquad
   (\Inst, \storage, \commits, \decodes) \not \sim_t (\Inst, \delta(t), \commits, \decodes)
  \\
  \hline
  (\Inst, \storage, \commits, \decodes, \spec, \guesses) \triplestep{}{}
  (\Inst \setminus \Delta^+,  \mapminus{\storage}{\Delta^*},
  \commits, \decodes \setminus
  \Delta^*,\mapminus{\spec}{\Delta^*},\guesses \setminus \Delta^*)
\end{array}
\end{array}
\]
\end{newnotation}
\changed{
For example, in state $\specstate_7$ the program counter update $t_6$ can be
rolled back because $\storage(t_2) = 1 \neq 0 = \spec(t_6)(t_2)$. The
transition moves the microinstruction $t_6$ back to the decoded state
(i.e., the storage and snapshot $\specstate_8$ are undefined for
$t_6$) and removes every microinstruction that have been decoded by $t_6$
(i.e., $t'_1$ and $t'_2$).
}

Notice that  rollbacks can be performed out of order and that 
loads can be retired even in case of mispredictions if their
dependencies have been enforced. This permits to model
advanced recovery methods used by modern processors, including
concurrent and partial recovery in case of multiple mispredictions.

\noindent \textbf{Speculation of load/store dependencies}
\changed{Since the predicted values of internal operations (cf. rule $\prd$) can affect conditions and targets
of program counter stores, the speculative semantics supports speculation of control flow, as well as 
speculative execution of cross-dependencies resulting from prediction
of load/store's addresses. We illustrate these features with an example (Figure~\ref{fig:spec:load-store:dep}),
which depicts one possible execution of the program in Example~\ref{example:active-store}. 
}

\changed{
Consider the state $\specstate_0$ after
the CPU has executed and retired microinstructions $t_{11}, t_{12}, t_{21}, t_{22}$, and
$t_{41}$, thus resolving the first two stores and the load's address. 
}
In state $\specstate_0$ the CPU can predict the address (i.e., the value of $t_{31}$)
of the third store as $0$ and modify the state as in $\specstate_1$ (rule $\prd$).

This prediction enables speculative execution of the load $t_{42}$ in state $\specstate_1$: the active store's bounded names  $bn(\stract(\newstate_1, t_{42}))$ 
consist of the singleton set $\{t_{12}\}$, since $\storage_1(t_{21})=
\storage_1(t_{31}) = 0$, while $\storage_1(t_{41}) = 1$. 
Hence, we can apply rule $\exc$ to execute  $t_{42}$, thus updating the
storage  with $t_{42} \mapsto 1$, and recording the snapshot 
$\{t_{11} \mapsto 1, t_{21} \mapsto 0, t_{31} \mapsto 0, t_{41}
\mapsto 1, t_{12} \mapsto 1\}$ for $t_{42}$. Concretely, $t_{42}$'s dependencies in state $\specstate_1$ \todo{Pls check, it was $\newstate_3$} consists of 
the local dependencies (i.e.,  the load's address $t_{41}$), and the cross dependencies containing  $t_{12}$ (i.e., active store it loads the value from), as well as 
the potential sources of $t_{42}$, that is, the  addresses of all stores between the active store $t_{12}$ and the load $t_{42}$, namely $t_{11}, t_{21}$ and $t_{31}$.

\begin{figure}
\begin{center}
\begin{tikzpicture}[node distance=1cm,auto,>=latex']

\statestyle                      

\runexecplain                   {1-11}{left:$t_{11}$}{1}
\runexecplain [right of = 1-11] {1-12}{left:$t_{12}$}{1}
\runexecplain [below = 0.1cm of 1-11] {1-21}{left:$t_{21}$}{0}
\runexecplain [right of = 1-21] {1-22}{left:$t_{22}$}{2}
\runemptyplain[below = 0.1cm of 1-21] {1-31}{left:$t_{31}$}
\runemptyplain[right of = 1-31] {1-32}{left:$t_{32}$}
\runexecplain [below = 0.1cm of 1-31] {1-41}{left:$t_{41}$}{1}
\runemptyplain[right of = 1-41] {1-42}{left:$t_{42}$}

\node    [below right =0.1cm and 0.15cm of 1-41, name=l10] {$\specstate_0$};

\tracerightarrow{1-12}{\prd, t_{31}}

\runexec [right =1cm of 1-12] {2-11}{}{1}
\runexec [right of = 2-11]    {2-12}{}{1}
\runexec [below = 0.1cm of 2-11]    {2-21}{}{0}
\runexec [right of = 2-21]    {2-22}{}{2}
\runpred [below = 0.1cm of 2-21]    {2-31}{}{0}
\runempty[right of = 2-31]    {2-32}{}
\runexec [below = 0.1cm of 2-31]    {2-41}{}{1}
\runempty[right of = 2-41]    {2-42}{}

\node    [below right =0.1cm and 0.15cm of 2-41, name=l2] {$\specstate_1$};

  \node[below right =0.2cm and 0.2cm of 2-12, name=tr2-a] {};
  \node[below right =0.2cm and 0.6cm of tr2-a, name=tr2-b] {};
  \draw[->] (tr2-a) edge node[above]{$\exe, t_{42}$} (tr2-b);

\runexec [right =1.2cm of 2-32] {3-11}{}{1}
\runexec [right of = 3-11]    {3-12}{}{1}
\runexec [below = 0.1cm of 3-11]    {3-21}{}{0}
\runexec [right of = 3-21]    {3-22}{}{2}
\runpred [below = 0.1cm of 3-21]    {3-31}{}{0}
\runempty[right of = 3-31]    {3-32}{}
\runexec [below = 0.1cm of 3-31]    {3-41}{}{1}
\runexec [right of = 3-41]    {3-42}{}{1}

\draw[->] (3-42) edge node[above,pos=0.7]{$1$} (3-11);
\draw[->] (3-42) edge node[above,pos=0.7]{$0$} (3-21);
\draw[->] (3-42) edge node[above,pos=0.7]{$0$} (3-31);
\draw[->] (3-42) edge node[above,pos=0.7]{$1$} (3-41);
\draw[->] (3-42) edge[bend right] node[right]{$1$} (3-12);

\node    [below right =0.1cm and 0.15cm of 3-41, name=l3] {$\specstate_2$};

\node[below left =0.2cm and 0.2cm of 3-41, name=tr3-a] {};
\node[below left =0.2cm and 0.4cm of tr3-a, name=tr3-b] {};
\draw[->] (tr3-a) edge node[above, pos=0]{$\pexe, t_{31}$} (tr3-b);

\runexec [below =0.5cm of 2-41]                         {4-11}{}{1}
\runexec [right of = 4-11]                            {4-12}{}{1}
\runexec [below = 0.1cm of 4-11]                            {4-21}{}{0}
\runexec [right of = 4-21]                            {4-22}{}{2}
\runexec [below = 0.1cm of 4-21]                            {4-31}{}{1}
\runempty[right of = 4-31]                            {4-32}{}
\runexec [below = 0.1cm of 4-31]                            {4-41}{}{1}
\runexec [right of = 4-41]                            {4-42}{}{1}
\draw[->] (4-42) edge node[above,pos=0.7]{$1$}    (4-11);
\draw[->] (4-42) edge node[above,pos=0.7]{$0$}    (4-21);
\draw[->] (4-42) edge node[above,pos=0.7]{$0$}    (4-31);
\draw[->] (4-42) edge node[above,pos=0.7]{$1$}    (4-41);
\draw[->] (4-42) edge[bend right] node[right]{$1$}(4-12);

\node    [below right =0.1cm and 0.15cm of 4-41, name=l3] {$\specstate_3$};

\traceleftarrow{4-21}{\rbk, t_{42}}

\runexec [left =2cm of 4-11]                         {5-11}{}{1}
\runexec [right of = 5-11]                            {5-12}{}{1}
\runexec [below = 0.1cm of 5-11]                            {5-21}{}{0}
\runexec [right of = 5-21]                            {5-22}{}{2}
\runexec [below = 0.1cm of 5-21]                            {5-31}{}{1}
\runempty[right of = 5-31]                            {5-32}{}
\runexec [below = 0.1cm of 5-31]                            {5-41}{}{1}
\runempty[right of = 5-41]                            {5-42}{}

\node    [below right =0.1cm and 0.15cm of 5-41, name=l4] {$\specstate_4$};

\node    [below left =0.35cm and 0.1 cm of 5-41, name=line1] {};
\node    [right =4cm of line1, name=line2] {};
\draw[-] (line1) edge  (line2);

\runexec [below =0.5cm of 4-41]                         {4'-11}{}{1}
\runexec [right of = 4'-11]                            {4'-12}{}{1}
\runexec [below = 0.1cm of 4'-11]                            {4'-21}{}{0}
\runexec [right of = 4'-21]                            {4'-22}{}{2}
\runexec [below = 0.1cm of 4'-21]                            {4'-31}{}{5}
\runempty[right of = 4'-31]                            {4'-32}{}
\runexec [below = 0.1cm of 4'-31]                            {4'-41}{}{1}
\runexec [right of = 4'-41]                            {4'-42}{}{1}
\draw[->] (4'-42) edge node[above,pos=0.7]{$1$} (4'-11);
\draw[->] (4'-42) edge node[above,pos=0.7]{$0$} (4'-21);
\draw[->] (4'-42) edge node[above,pos=0.7]{$0$} (4'-31);
\draw[->] (4'-42) edge node[above,pos=0.7]{$1$} (4'-41);
\draw[->] (4'-42) edge[bend right] node[right]{$1$} (4'-12);

\node    [below right =0.1cm and 0.15cm of 4'-41, name=l'3] {$\specstate'_3$};

\traceleftarrow{4'-21}{\ret, t_{42}}

\runexec [left =2cm of 4'-11]                         {5'-11}{}{1}
\runexec [right of = 5'-11]                            {5'-12}{}{1}
\runexec [below = 0.1cm of 5'-11]                            {5'-21}{}{0}
\runexec [right of = 5'-21]                            {5'-22}{}{2}
\runexec [below = 0.1cm of 5'-21]                            {5'-31}{}{1}
\runempty[right of = 5'-31]                            {5'-32}{}
\runexec [below = 0.1cm of 5'-31]                            {5'-41}{}{1}
\runexec [right of = 5'-41]                            {5'-42}{}{1}

\node    [below right =0.1cm and 0.15cm of 5'-41, name=l4'] {$\specstate'_4$};
\end{tikzpicture}
\caption{Speculation of load/store dependencies}
\label{fig:spec:load-store:dep}

\end{center}
\vspace*{-.5cm}
\end{figure}
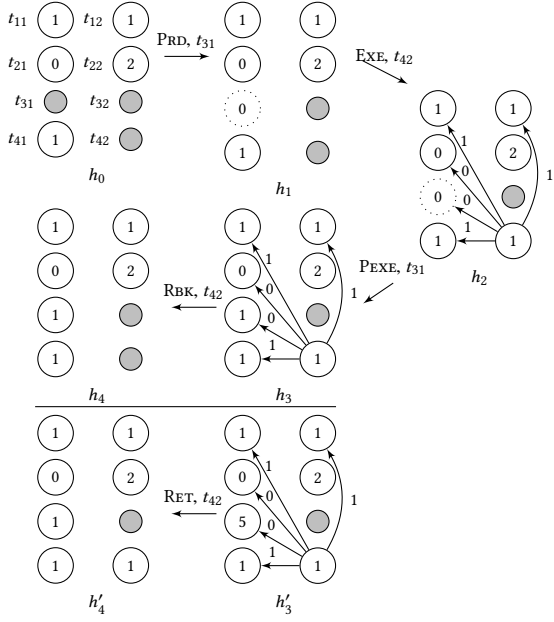

At this point, load $t_{42}$ cannot be retired by  rule $\ret$ in state $\specstate_2$ since its dependencies, e.g., $t_{31}$, are yet to be retired.
However, we can execute  $t_{31}$ by applying rule $\pexe$. The execution updates the state by removing $t_{31}$ from the prediction set and storing its correct value, 
as well as extending the snapshot with $t_{31} \mapsto \emptyset$, as depicted in state $\specstate_3$.

The execution of $t_{31}$ enables the premises of rule $\rbk$ to capture that the dependency misprediction led to 
misspeculation of the address of the load $t_{42}$.
Specifically, the set $asn$ at the time of $t_{42}$'s
execution
$bn(\stract((\Inst_3, \delta_3(t_{42}), \commits_3, \decodes_3), t_{41})) = \{t_{12}\}$ differs from 
the active store set $bn(\stract(\newstate_3, t_{41}))) = \{t_{32}\}$
in the current state. Therefore, we roll back the execution removing
the mappings for $t_{42}$ from the storage and the snapshot as in $\specstate_4$.

Finally, we remark that the speculative execution of loads is rolled back \emph{only if} a misprediction causes 
a violation of load/store dependencies. For instance, if the value of $t_{31}$ was $5$ instead of $1$, as depicted in $\specstate'_3$,
the misprediction of $t_{31}$'s value as $0$ in $\specstate_1$ does
not enable a rollback of the load. This is because the  actual value
of $t_{31}$ does not change the \todo{check $\delta$}
set of active stores. In fact, the set of active stores at the time of $t_{42}$'s execution $bn(\stract((\Inst'_3, \delta'_3(t_{42}), \commits'_3, \decodes'_3), t_{41})) = \{t_{12}\}$ is the same as 
the active store's set $bn(\stract(\newstate'_{3}, t_{41}))) = \{t_{12}\}$ in the current state.

\input{app-in-order}

\newcommand{\prediction}{\mathit{pred}}
\section{Attacks and Countermeasures
}\label{sec:newvul}

InSpectre can be used to model and analyze (combinations of) microarchitectural features  underpinning Spectre
attacks~\cite{DBLP:conf/sp/KocherHFGGHHLM019,maisuradze2018ret2spec,DBLP:conf/uss/CanellaB0LBOPEG19}, and,
importantly, to discover new vulnerabilities and to reason about the
security of proposed countermeasures. Observe that these results hold for our generic microarchitectural model,
while specific CPUs  would require instantiating InSpectre to model their microarchitectural features.
We remark that real-world feasibility of our new vulnerabilities 
falls outside the scope of this work.

Specifically, we use the following recipe: We model a specific prediction strategy in InSpectre and try to prove
conditional noninterference for arbitrary programs. Failure to complete the security proof results in new 
classes of counterexamples as we report below. 


Concretely, prediction strategies and countermeasures are
modeled by constraining the nondeterminism in the microinstruction scheduler 
and in the prediction semantics (see rule $\prd$). 
The prediction function $\prediction_{p} : \Sigma \rightarrow \ANames \rightharpoonup
2^{\Val}$  captures a prediction strategy $p$ by computing the set of predicted values for a name $t \in \ANames$  and a  state 
$\newstate \in \Sigma$. We assume the
transition relation satisfies the following property:
If $(\newstate, \spec, \guesses) \triplestep{l}{}
  (\newstate', \spec',
  \guesses \cup \{t\})$ then
  $ t \in \dom{\prediction_p(\newstate)}$ and
  $\newstate'(t) \in \prediction_p(\newstate)(t)$. 
  This property ensures that the transition relation chooses predicted values from function $\prediction_p$.

Following the security model in Section~\ref{sec:secmod}, we check conditional noninterference by: ($a$) using the in-order transition relation
$\singlestep{}{}$ as reference model and speculative (OoO)
transition relation $\triplestep{}{}$ ($\doublestep{}{}$) as target
model; ($b$) providing the security policy
$\sim$ for memory and registers.
To invalidate conditional
noninterference it is sufficient to find two $\sim$-indistinguishable
states that yield the same observations in the reference model and
different observations  
in the target model.
%
We use the classification by Canella et al.~\cite{DBLP:conf/uss/CanellaB0LBOPEG19} to refer to existing attacks.

\subsection{Spectre-PHT}\label{sec:spectre:v1}
Spectre-PHT~\cite{DBLP:conf/sp/KocherHFGGHHLM019} exploits the prediction
mechanism for the outcome of conditional branches.
Modern CPUs use \emph{Pattern History Tables} (PHT) to record 
patterns of past executions of conditional branches, i.e.,
whether the \emph{true} or the \emph{false} branch was executed, and  
then use it to predict the outcome of that branch. 
By poisoning the PHT to execute one direction (say the
\emph{true} branch), an attacker can fool the
prediction mechanism to 
execute the \emph{true} branch, even when the actual outcome of the branch is ultimately \emph{false}. 
The following program (and the corresponding MIL) illustrates
information leaks via Spectre-PHT: \\
\begin{tabular}{ll}
\begin{lstlisting}[mathescape=true]
$a_1: r_1 = A_1.size;$
\end{lstlisting}
  &
\adjustbox{valign=t}{
\begin{tikzpicture}[node distance=1.7cm,auto,>=latex']
\pomsetstyle
\picloadT{t_{11}}{\Locs}{A_1.size}
\picstoreT[right of=t_{11}]{t_{12}}{\Regs}{r_1}{t_{11}}
\draw[->]  (t_{11}) edge (t_{12}) ;
\picstoreT[right of=t_{12}]{t_{13}}{\Pc}{}{a_2}
\end{tikzpicture}
    }
    \vspace*{.1cm}
  \\
  \hline
\begin{lstlisting}[mathescape=true]
$a_2: \mathit{if}\; (r_0 < r_1)$
\end{lstlisting}
  &
    \adjustbox{valign=t}{
    \begin{tikzpicture}[node distance=1.5cm,auto,>=latex']
\pomsetstyle
\picloadT{t_{21}}{\Regs}{r_0}
\picexpT[right of = t_{21}]{t_{23}}{t_{21} < t_{22}}
\picloadT[right of = t_{23}]{t_{22}}{\Regs}{r_1}
\picstore[below = \pomsetnewline of t_{21}]{t_{24}}{t_{23}}{\Pc}{}{a_3}
\picstore[below = \pomsetnewline of t_{22}]{t_{25}}{\neg t_{23}}{\Pc}{}{a_4}
\draw[->]  (t_{21}) edge (t_{23}) ;
\draw[->]  (t_{22}) edge (t_{23}) ;
\draw[->]  (t_{23}) edge (t_{24}) ;
\draw[->]  (t_{23}) edge (t_{25}) ;
\end{tikzpicture}
    }
        \vspace*{.1cm}
  \\
  \hline
\begin{lstlisting}[mathescape=true]
$a_3:  \quad y = A_2[A_1[r_0]];$
\end{lstlisting}
  &
    \adjustbox{valign=t}{
    \begin{tikzpicture}[node distance=1.8cm,auto,>=latex']
\pomsetstyle
\picloadT{t_{31}}{\Regs}{r_0}
\picloadT[right of = t_{31}]{t_{32}}{\Locs}{(A_1 + t_{31})}
\picloadTT[below = \pomsetnewline of t_{32}]{t_{33}}{\Locs}{(A_2 + t_{32})}
\picstoreT[right of = t_{33}]{t_{34}}{\Regs}{r_0}{t_{33}}
\picstoreT[right of = t_{32}]{t_{35}}{\Pc}{}{a_4}
\draw[->]  (t_{31}) edge (t_{32}) ;
\draw[->]  (t_{32}) edge (t_{33}) ;
\draw[->]  (t_{33}) edge (t_{34}) ;
\end{tikzpicture}
    }
\end{tabular}\\
Suppose the security policy labels as public the data in arrays $A_1$
and $A_2$, and in register $r_0$, and that the attacker controls the
value of $r_0$. This program is secure at the ISA level as it ensures that $r_0$ always lies within the bounds of $A_1$.
However, an attacker can fool the prediction mechanism by first supplying values of $r_0$ that execute the \emph{true} branch, and then  a value that 
exceeds the size of $A_1$. This causes the CPU to perform an
out-of-bounds memory access of sensitive data, which is later 
used as index for a second memory access of $A_2$, thus leaving a
trace into the cache.

Branch prediction  predicts values for MIL instructions that block the
evaluation of the guard of a PC store whose
target address has been already resolved. For  $\newstate = (\Inst, \storage, \commits, \decodes, \spec, \guesses)$, we model it as:
\[
 \prediction_{br}(\newstate) = \left \{ 
    t' \mapsto v \mid
    \ass{c}{t}{\store{\Pc}{}{t_a}} \in \Inst
    \land t' \in \fn(c) \land
    \fdefined{\storage}{t_a}
  \right  \}
\]
Let $\newstate_0$ be the state where only the
instruction in $a_1$ has been translated.
Then $\prediction_{br}(\newstate_0)$ is empty, since $\newstate_0$
contains a single unconditional PC update (the guard of
$t_{13}$ has no free names).
The CPU may  apply rules $\exc$, $\ret$, and $\ftc$  
on $t_{13}$ without waiting
the result of $t_{11}$. This leads to a new state $\newstate_1$ which is
obtained by updating the storage
with $\storage_1 = \{t_{13} \mapsto a_2\}$,  extending the  microinstructions' list with the translation
of $a_2$, and  the snapshot with  $\spec_1 = \{t_{2i} \mapsto t_{13} \mapsto a_2 \mbox{ for } 1 \leq i \leq 5\}$, while producing the observation $\il{a_2}$.
In this state $\prediction_{br}(\newstate_1) = \{t_{23} \mapsto 0,
t_{23} \mapsto 1\}$ since the conditions of the two PC
stores (i.e., $t_{24}$ and $t_{25}$) depend on $t_{23}$ which is yet
to be resolved.
The CPU can now apply rule $\prd$ using the prediction $t_{23} \mapsto 1$, thus guessing that 
the condition is true. The new state $\newstate_2$
contains $\storage_2 = \storage_1 \cup \{t_{23} \mapsto 1\}$, $\spec_2 =  \spec_1$, and
$\guesses_2 = \{t_{23}\}$.  

The CPU can follow the speculated
branch by applying rules $\exe$  and  $\ftc$ on  $t_{24}$, which results in  state $\newstate_3$ with
$\storage_3 = \storage_2 \cup \{t_{24} \mapsto a_3\}$,
$\spec_3 = \spec_2 \cup \{t_{24} \mapsto \{t_{23} \mapsto 1\}, t_{3i} \mapsto
\{t_{24} \mapsto a_3\} \mbox{ for } 1 \leq i \leq 5
\}$, and $\decodes_3=\{t_{13},
t_{24}\}$. Additionally, it produces the observation $\il{a_3}$.

Applying rule $\exe$ on  $t_{31}$ and $t_{32}$ results in a buffer overread and produces state $\newstate_4$ 
with $\storage_4 = \storage_3 \cup \{t_{31} \mapsto r_0, t_{32} \mapsto A_1[r_0]\}$, and observation $\dl{r_0}$.
Similarly, rule $\exe$ on $t_{33}$ produces the observation
$\dl{A_2 + A_1[r_0]}$. 

Clearly, if $r_0 \geq A_1.size$,  the observation reveals memory content
outside $A_1$, allowing an attacker to learn sensitive data. Observe that this is rejected by the security condition, since
such observation is not possible in the sequential semantics. 

\subsubsection{Countermeasure: Serializing Instructions}\label{sec:hcm}
Serializing instructions can be modeled by constraining the scheduling of microinstructions.
For example, we can model the Intel's \emph{lfence} instruction via a function 
$\mathit{lfence}(\Inst)$
that extracts all microinstructions resulting from the translation
of lfence.  

Concretely, for $\newstate = (\Inst, \storage, \commits, \decodes, \spec,
\guesses)$, $t \in \mathit{lfence}(\Inst)$
and $\newstate \triplestep{}{}  \newstate'$, it holds that:
\begin{inparaenum}[(i)]
\item if $\fundefined{\newstate}{t}$ and $\fdefined{\newstate'}{t}$
  then
 for each $\ass{c}{t'}{\load{\Locs}{t_a}} \in \newstate$ such that
 $t' < t$ 
 $fn(c) \subseteq \dom{\storage}
 \setminus \dom{\delta}$ 
 and $c \Rightarrow (\fdefined{\newstate}{t'} \wedge
 \fundefined{\spec}{t'})$; 
\item if  $\fundefined{\newstate}{t'}$, $\fdefined{\newstate'}{t'}$, 
  $t' > t$,
  and $\ass{c}{t'}{\load{\Locs}{t_a}} \in \newstate$, or
  $\ass{c}{t'}{\store{\Locs}{t_a}{t_v}} \in \newstate$, or
  $\ass{c}{t'}{\store{\Regs}{t_a}{t_v}} \in \newstate$,
  then $\fdefined{\newstate}{t} \:\wedge\:
  \fundefined{\spec}{t}$. 
\end{inparaenum}

Intuitively, the conditions restrict the scheduling of microinstructions to ensure that: 
(i) whenever a fence is executed, all previous loads have been retired, and (ii) subsequent memory 
operations or register stores can be executed only if the fence has been retired.

In order to reduce the
performance overhead, several works~(e.g. \cite{pardoe2018spectre}) use static
analysis to identify necessary serialization points in a program.
In the previous example, it is sufficient to place lfence after $t_{32}$ and
before $t_{33}$.
This does not prevent the
initial buffer overread of $t_{32}$, however, it suspends $t_{33}$ until
$t_{32}$ is retired. In case of misprediction, $t_{32}$ and $t_{33}$
will be rolled back, preventing the observation $\dl{A_2 + A_1[r_0]}$ which causes the information leak.
\subsubsection{Countermeasure: Implicit Serialization}\label{sec:spectre:cmov}
An alternative countermeasure to prevent Spectre-PHT is to use instructions
that introduce implicit serialization~\cite{fogh18,miller18}. For instance, adding 
the following gadget  between instructions $a_2$ and $a_3$ in the previous example
prevents Spectre-PHT on existing Intel CPUs:\\
\begin{tabular}{ll}
  \begin{lstlisting}[mathescape=true]
// cmp
$a'_3:  f = (r_0 \ge r_1)$
\end{lstlisting}
  &
    \adjustbox{valign=t}{
    \begin{tikzpicture}[node distance=2cm,auto,>=latex']
\pomsetstyle
\picloadT{t'_{31}}{\Regs}{r_0}
\picexpT[right of = t'_{31}]{t'_{33}}{t'_{31} \ge t'_{32}}
\picloadT[right of = t'_{33}]{t'_{32}}{\Regs}{r_1}
\picstoreTT[below = \pomsetnewline of t'_{33}]{t'_{34}}{\Regs}{f}{t'_{33}}
\picstoreT[below = \pomsetnewline of t'_{32}]{t'_{35}}{\Pc}{}{a''_3}
\draw[->]  (t'_{31}) edge (t'_{33}) ;
\draw[->]  (t'_{32}) edge (t'_{33}) ;
\draw[->]  (t'_{33}) edge (t'_{34}) ;
\end{tikzpicture}
    }
            \vspace*{.1cm}
  \\
  \hline
\begin{lstlisting}[mathescape=true]
$a''_3:  cmov f, r_0, 0$
\end{lstlisting}
  &
\adjustbox{valign=t}{
    \begin{tikzpicture}[node distance=2cm,auto,>=latex']
\pomsetstyle
\picloadT{t''_{31}}{\Regs}{f}
\picstore[right of=t''_{31}]{t''_{32}}{t''_{31}=1}{\Regs}{r_0}{0}
\picstoreT[right =0.45cm of t''_{32}]{t''_{34}}{\Pc}{}{a'''_3}
\draw[->]  (t''_{31}) edge (t''_{32}) ;
\end{tikzpicture}
  }
\vspace*{.2cm}
\end{tabular}\\
Intuitively, this gadget forces mispredictions to always access 
$A_1[0]$,
Consider the extension of the previous example with the gadget and suppose $\prediction_{br}$ mispredicts  $t_{23} \mapsto 1$. The instruction in $a''_3$
introduces a data dependency between $t_{11}$ and $t_{32}$ since
 $\stract$ of $t_{31}$ includes $t''_{32}$ until $t''_{31}$
has been executed;  $\stract$ of $t''_{31}$ includes $t'_{34}$;
and $\stract$ of $t'_{32}$ includes $t_{12}$. These names
(and intermediate intra-instruction dependencies) are in the free
names of some condition of a PC store, hence they cannot be predicted by
$\prediction_{br}$ and their dependencies are enforced by the
semantics. In particular, when $t_{23}$ is mispredicted as $1$, 
$t''_{32}$ is executed after that $t_{11}$ has obtained the value from
the memory. This ensures that $t''_{32}$ sets $r_0$ to $0$ every time
a buffer overread occurs. Therefore misspeculations generate the
observations $\dl{A_1 + 0}$ and $\dl{A_2 + A_1[0]}$, which
do not violate the security condition (since $A_1$ is
labeled as public).

\subsubsection{New Vulnerability: Spectre-PHT ICache}\label{sec:new-attack:fetch}
When the first Spectre attack was published, some microarchitectures
(e.g., Cortex A53) were claimed immune to the attack because of ``allowing speculative fetching but not speculative
execution'' \cite{cortexA53}.
%
The informal argument was that mispredictions cannot cause buffer overreads or leave any footprint on the cache in absence of
speculative loads.
To check this claim, we  constrain the
semantics to  only allow speculation of PC values. 
Specifically, we require for any transition $(\newstate, \spec, \guesses) \triplestep{}{} (\newstate', \spec', \guesses')$ that executes a microinstruction ($\stepparam(\newstate,\newstate') = (\exe,
  t)$) which is either a load
  ($\ass{c}{t}{\load{\tau}{t_a}} \in \newstate$) or a store
  ($\ass{c}{t}{\store{\tau}{t_a}{t_v}} \in \newstate$) of a resource other than the program counter 
  ($\tau \neq \Pc$) to have an empty snapshot on past microinstructions ($\dom{\spec} \cap \{t' \mid t' < t\} = \emptyset$).

The analysis of conditional noninterference for this model led to
the identification of a class of counterexamples, which we call Spectre-PHT ICache, where
branch prediction causes leakage of sensitive data via an ICache disclosure
gadget.

Consider a  program that jumps to the address pointed to by
\stylecode{sec} if a user has \stylecode{admin}
privileges, otherwise it continues to address $a_3$.
\\
\begin{tabular}{ll}
\begin{lstlisting}[mathescape=true]
$a_1: r_1=*sec$
\end{lstlisting}
  &
\adjustbox{valign=t}{
\begin{tikzpicture}[node distance=2cm,auto,>=latex']
\pomsetstyle
\picloadT{t_{11}}{\Locs}{sec}
\picstoreT[right of=t_{11}]{t_{12}}{\Regs}{r_1}{t_{11}}
\draw[->]  (t_{11}) edge (t_{12}) ;

\picstoreT[right of=t_{12}]{t_{13}}{\Pc}{}{a_2}
\end{tikzpicture}
}
        \vspace*{.1cm}
  \\
  \hline
\begin{lstlisting}[mathescape=true]
$a_2: if (*admin)$
       $(*r_1)()$
\end{lstlisting}
  &
    \adjustbox{valign=t}{
\begin{tikzpicture}[node distance=1.7cm,auto,>=latex']
\pomsetstyle
\picloadT{t_{21}}{\Locs}{admin}
\picexpT[right of = t_{21}]{t_{23}}{t_{21}  \neq 1}
\picloadT[right of = t_{23}]{t_{22}}{\Regs}{r_1}
\picstore[below = \pomsetnewline of t_{21}]{t_{24}}{t_{23}}{\Pc}{}{(a_2+4)}
\picstore[below = \pomsetnewline of t_{22}]{t_{25}}{\neg t_{23}}{\Pc}{}{t_{22}}
\draw[->]  (t_{21}) edge (t_{23}) ;
\draw[->]  (t_{22}) edge (t_{25}) ;
\draw[->]  (t_{23}) edge (t_{24}) ;
\draw[->]  (t_{23}) edge (t_{25}) ;
\end{tikzpicture}
    }
    \vspace*{.2cm}
\end{tabular}

 In the sequential model, an attacker that
only observes the instruction cache can see
the sequence of observations $\il{a_1}::\il{a_2}::\il{a_2}$ if
$*admin \neq 1$, otherwise the sequence $\il{a_1}::\il{a_2}::\il{sec}$. 

A CPU that supports only speculative fetching may first complete all microinstructions 
in $a_1$, and then predict the result of $t_{23}$
to enable the execution of $t_{25}$. 
As a result the PC speculatively fetches the instruction at location $sec$ although $*admin \neq
1$. The transition yields the observation sequence $\il{a_1}::\il{a_2}::\il{sec}$ which was not possible in the sequential model, thus violating the
security condition and leaking the value of $sec$ via the
instruction cache.

Intel's lfence does not stop all microarchitectural
operations, like instruction fetching. For this reason lfence may be
ineffective against leakage via ICache. In fact, InSpectre  reveals that placing
a lfence between $t_{21}$ and $t_{23}$ does not prevent the leakage:
$t_{22}, t_{23}, t_{25}$ are neither memory operations nor register
stores, hence they can be speculated before the execution of the
lfence.

\input{app-other-counter}

\subsection{Spectre-STL}
\label{sec:new-attack:v4}
Spectre-STL~\cite{spectrev4}
(Store-To-Load) exploits the CPUs mechanism to predict load-to-store data dependencies. A load  cannot be
executed before executing all the past (in program order) stores that
affect the same memory address. However, if the
address of a past store has not been resolved, the CPU
may execute the load in speculation without waiting for the store,
predicting that the target address of the store is different from the
load's address. Mispredictions cause store bypasses 
leading to information leaks and access to stale data.
This behavior can be modeled as $\prediction_{STL}(\newstate, \spec,
  \guesses) = $
    \[
  \left \{
    t_a \mapsto a \mid
    \begin{array}{l}
      \ass{c'}{t'}{\load{\Locs}{t'_a}} \in \newstate \land  \newstate(t'_a) \neq  a\\
      \ass{c}{t}{\store{\Locs}{t_a}{t_v}} \in \stract(\newstate, t') \land \\
      \fundefined{\newstate}{t_a}
    \end{array}
      \right  \}
    \]
A prediction occurs whenever a memory store ($t$) is waiting an unresolved address
($\fundefined{\newstate}{t_a}$), while the address ($\storage(t'_a)$) of
a subsequent load ($t'$) has been resolved, and the load  may depend on
the store ($t \in bn(\stract(\newstate, t'))$).
Prediction guesses that the store's address ($t_a$) differs with the load's
address.

\subsubsection{Hardware countermeasures to Store Bypass}
The specification of proposed  hardware countermeasures  oftentimes comes with no precise semantics and is ambiguous. 
ARM introduced  the Speculative Store
Bypass Safe (SSBS) configuration to prevent store bypass
vulnerabilities. The specification of SSBS~\cite{ssbs} is:
\emph{
Hardware is not permitted to load \dots speculatively, in a
manner that could \dots give rise to a \dots side
channel, using an address derived from a register value that has been
loaded from memory \dots (L) that speculatively
reads an entry from earlier in the coherence order from that location
being loaded from than the entry generated by the latest store (S) to
that location using the same virtual address as L.
}

InSpectre provides a ground to formalize the behavior of  these hardware
mechanisms.
We formalize SSBS as follows.
Let $\newstate = (\Inst, \storage, \commits, \decodes, \spec, \guesses)$  and $\ass{c}{t}{\load{\restype}{t_a}} \in
\newstate$. 
If $\newstate \triplestep{l}{}  \newstate'$, $\fundefined{\newstate}{t}$, and 
$\fdefined{\newstate'}{t}$,
then
for every $t' \in \sources(t,\newstate)$, if $\newstate(t') \neq \newstate(t_a)$
then $t' \not \in \guesses$.

The reason why SSBS prevents Spectre-STL is simple. The rule
forbids the execution of a load $t$ if any address used to 
identify the last store affecting $t_a$  has been
predicted to differ from $t_a$.

\subsubsection{New Vulnerability: Spectre-STL-D}
Our model reveals that if a microarchitecture
mispredicts \emph{the existence of a Store-To-Load Dependency} (hence Spectre-STL-D), e.g., in
order to forward temporary store results, a similar vulnerability may
be possible. To model this behavior it is enough to substitute
$\newstate(t'_a) \neq  a$ with $\newstate(t'_a) =  a$ in
$\prediction_{STL}$.
We consider this a new form of Spectre because the implementation of
this microarchitectural feature can be substantially different from
the one required for Spectre-STL (e.g., Feiste et al.~\cite{feiste}
patented a mechanism to implement this feature) and because the
vulnerable programs are different. 


This feature may cause Spectre-STL-D if a misspeculated dependency is used to perform subsequent
memory accesses.  Consider the following program:\\
\begin{tabular}{ll}
\begin{lstlisting}[mathescape=true]
$a_1:$*(*$b_1$):=$sec$
\end{lstlisting}
  &
    \adjustbox{valign=t}{
\begin{tikzpicture}[node distance=1.6cm,auto,>=latex']
\pomsetstyle
\picloadT{t_{11}}{\Locs}{b_1}
\picexpT[right =0.4cm of t_{11}]{t_{12}}{t_{11}}
\picstoreT[right =0.4cm of t_{12}]{t_{13}}{\Locs}{t_{12}}{t_{sec}}
\picstoreT[right of=t_{13}]{t_{14}}{\Pc}{}{a_2}
\draw[->]  (t_{11}) edge (t_{12}) ;
\draw[->]  (t_{12}) edge (t_{13}) ;
\end{tikzpicture}}
        \vspace*{.1cm}
  \\
  \hline
\begin{lstlisting}[mathescape=true]
$a_2: r_1$:=*(*$b_2$)
\end{lstlisting}
  &
    \adjustbox{valign=t}{
    \begin{tikzpicture}[node distance=2cm,auto,>=latex']
\pomsetstyle
\picloadT{t_{21}}{\Locs}{b_2}
\picloadT[right of= t_{21}]{t_{22}}{\Locs}{t_{21}}
\picstoreT[right of= t_{22}]{t_{23}}{\Regs}{r_1}{t_{22}}
\draw[->]  (t_{21}) edge (t_{22}) ;
\draw[->]  (t_{22}) edge (t_{23}) ;
\end{tikzpicture}
    }
    \vspace*{.2cm}
\end{tabular}
\\
If the CPU executes and fetches $t_{14}$, predicts that $t_{12} = b_2$ (i.e., it mispredicts the
alias \stylecode{*$b_1$==*$b_2$}), executes $t_{13}$, forwards the
result of $t_{13}$ to $t_{21}$, and executes $t_{22}$ 
before the load $t_{11}$ is retired, then the address accessed by
$t_{22}$ depends on $t_{sec}$.
This can produce the secret-dependent sequence of observations
$\il{a_1}::\il{a_2}::\dl{sec}$,
while the sequential semantics always produces the
secret-independent sequence of observations $\il{a_1}::\dl{b_1}::\ds{*b_1}::\il{a_2}::\dl{b_2}::\dl{*b_2}$.
Notice that SSBS may not be effective against Spectre-STL-D.

\subsection{New Vulnerability: Spectre-OoO}\label{sec:newvul:ooo}
\label{sec:newattack:ooo}
A popular countermeasure to prevent sensitive data from 
affecting the execution time and caches
is ``constant time programming'', also known as ``data oblivious
algorithms''. This mechanism ensures that branch conditions 
and  memory addresses are independent of sensitive data.
The following definition formalizes ``ISA constant time''
while abstracting from the specific ISA:
\newcommand{\bisim}{\approx}
\begin{definition}
  \label{def:ct}
  A program is ``ISA constant time'' if for every pair of states $\newstate_1 \sim  \newstate_2$ and every pair of in-order  executions of length
  $n$, $\newstate_1 \singlestep{}{}^n \newstate'_1$ and $\newstate_2
  \singlestep{}{}^n \newstate'_2$, it is the case that $\newstate'_1 \bisim_{ISA} \newstate'_2$, where
  $(\Inst, \storage, \commits, \decodes) \bisim_{ISA} (\Inst', \storage', \commits',
  \decodes')$ iff
  \begin{enumerate}
  \item $\Inst = \Inst'$, $\commits=\commits'$, $\decodes=\decodes'$:
    the sets of microinstructions,  commits and  decodes  are equal.
  \item If $\ass{c}{t}{\load{\Locs}{t_a}} \in \Inst$ or
    $\ass{c}{t}{\store{\Locs}{t_a}{t_v}} \in \Inst$
    then 
    $\den{c}\newstate = \den{c}\newstate'$,
    $\fdefined{\storage}{t} = \fdefined{\storage'}{t}$,
    (whenever defined, guards evaluate the same, and memory operations execute in lockstep)
    and $\den{c}\newstate \Rightarrow (\newstate(t_a) =
    \newstate'(t_a))$
    (the same values are used to address memory)
  \item If $\ass{c}{t}{\store{\Pc}{t_v}} \in \Inst$
    then 
    $\den{c}\newstate = \den{c}\newstate'$,
    $\fdefined{\storage}{t} = \fdefined{\storage'}{t}$,
    and $\den{c}\newstate \Rightarrow (\newstate(t_v) =
    \newstate'(t_v))$
    (the same values are used to update the PC)
  \end{enumerate}
\end{definition}
The following program (and its MIL translation) exemplifies this policy. It loads register $r_1$ from address $b_1$, copies
the value of $r_1$ in $r_2$ if the flag
$z$ is set, and saves the result into $b_2$.\\
\begin{tabular}{ll}
\begin{lstlisting}[mathescape=true]
$a_1: r_1 = *b_1;$
\end{lstlisting}
  &
    \adjustbox{valign=t}{
 \begin{tikzpicture}[node distance=2cm,auto,>=latex']
\pomsetstyle
\picloadT{t_{11}}{\Locs}{b_1}
\picstoreT[right of= t_{11}]{t_{12}}{\Regs}{r_1}{t_{11}}
\picstoreT[right of=t_{12}]{t_{13}}{\Pc}{}{a_2}
\draw[->]  (t_{11}) edge (t_{12}) ;
\end{tikzpicture}
    }
            \vspace*{.1cm}
  \\
  \hline
  \begin{lstlisting}[mathescape=true]
$a_2: cmov \: z, \: r_2, \:r_1;$ 
\end{lstlisting}
  &
    \adjustbox{valign=t}{
\begin{tikzpicture}[node distance=2cm,auto,>=latex']
\pomsetstyle
\picloadT{t_{21}}{\Regs}{z}
\picload[right of=t_{21}]{t_{22}}{t_{21}=1}{\Regs}{r_1}
\picstoreT[right =0.6cm of t_{22}]{t_{24}}{\Pc}{}{a_3}
\picstore[below = \pomsetnewline of t_{22}]{t_{23}}{t_{21}=1}{\Regs}{r_2}{t_{22}}
\draw[->]  (t_{21}) edge (t_{22}) ;
\draw[->]  (t_{21}) edge (t_{23}) ;
\draw[->]  (t_{22}) edge (t_{23}) ;
\end{tikzpicture}
    }
            \vspace*{.1cm}
  \\ \hline
  \begin{lstlisting}[mathescape=true]
$a_3: *b_2 = r_2;$
\end{lstlisting}
  &
    \adjustbox{valign=t}{
\begin{tikzpicture}[node distance=2cm,auto,>=latex']
\pomsetstyle
\picloadT{t_{31}}{\Regs}{r_2}
\picstoreT[right of=t_{31}]{t_{32}}{\Locs}{b_2}{t_{31}}
\draw[->]  (t_{31}) edge (t_{32}) ;
\end{tikzpicture}
    }
    \vspace*{.2cm}
\end{tabular}
\\
Suppose that flag $z$ contains sensitive information and the attacker
observes only the data cache.
The ``conditional move'' instruction in $a_2$  executes in constant
time~\cite{intelct} and is used to re-write branches that may leak
information via the execution time or the instruction cache.
This allows the program to always access address $b_1$ and
$b_2$ unconditionally and execute always the same ISA instructions:
In the sequential model the program always produces the sequence of
observations $\dl{b_1}::\ds{b_2}$.

Programs that are ISA constant time could be insecure in presence of
speculation, as demonstrated by Spectre-PHT~\cite{DBLP:conf/sp/KocherHFGGHHLM019}.
Perhaps surprisingly, it turns out that ISA constant time is not secure even for the
OoO model, in absence of speculation.
In fact, our analysis of conditional noninterference for ISA constant
time programs in the OoO model led to the identification of a
class of vulnerable programs, where secrets influence the existence of
data dependency between registers.
The above program exemplifies this problem:  the data dependency between $t_{11}$ and $t_{32}$ exists only
if $z$ is set. Concretely, consider two states $\newstate_0$ and $\newstate_1$ in which $z=0$ and $z=1$, respectively. Then, $\stract(\newstate_1,t_{31}) = \{t_{23}\}$ and 
$\stract(\newstate_1, t_{23}) = \{t_{12}\}$, while
$\stract(\newstate_0, t_{31})$ is the microinstruction representing
the initial value of $r_2$.
Therefore, state $\newstate_0$ may produce the observation sequence $\ds{b_2}::\dl{b_1}$ 
only if the flag $z=0$, thus leaking its value through the data cache.

\subsubsection{MIL Constant Time}
Spectre-OoO
motivates the need for a new microarchitecture-aware definition of constant time. 

\begin{definition}
  \label{def:ct}
  A program is ``MIL constant time'' if for every pair of states $\newstate_1 \sim  \newstate_2$ and every pair of in-order  executions of length
  $n$, $\newstate_1 \singlestep{}{}^n \newstate'_1$ and $\newstate_2
  \singlestep{}{}^n \newstate'_2$, it is the case that $\newstate'_1 \bisim_{MIL} \newstate'_2$, where
  $(\Inst, \storage, \commits, \decodes) \bisim_{MIL} (\Inst', \storage', \commits',
  \decodes')$ iff
  \begin{enumerate}
  \item $(\Inst, \storage, \commits, \decodes) \bisim_{ISA} (\Inst', \storage', \commits',
  \decodes')$
  \item If $\ass{c}{t}{\load{\Regs}{t_a}} \in \Inst$ or
    $\ass{c}{t}{\store{\Regs}{t_a}{t_v}} \in \Inst$
    then 
    $\den{c}\newstate = \den{c}\newstate'$,
    $\fdefined{\storage}{t} = \fdefined{\storage'}{t}$,
    and $\den{c}\newstate \Rightarrow (\newstate(t_a) =
    \newstate'(t_a))$
  \end{enumerate}
\end{definition}
Notice that in addition to standard requirements of constant time, MIL constant time
requires that starting from two
$\sim$-indistinguishable states the program makes the same 
accesses to registers.
%
MIL constant time is sufficient to ensure security in the
OoO model:
\begin{theorem}\label{thm:ct}
  If a program $P$ is MIL constant time then
  $P$ is conditionally noninterferent in the OoO model.
\end{theorem}
%
The theorem 
enables the enforcement of  conditional noninterference for
the OoO model by verifying MIL constant time in the sequential model.
This strategy has the advantage of performing the verification in the sequential model, which is deterministic, thus making it easier to
reuse existing tools for binary code analyses~\cite{DBLP:conf/ccs/BalliuDG14}.

Finally, we remark that MIL constant time is microarchitecture aware. This means that the same ISA program may or may not
satisfy MIL constant time when translated to a given
microarchitecture. In fact, the MIL translation of conditional move
above is not MIL constant time because of the dependency between  the sensitive value in $t_{21}$ and conditional store in $t_{23}$.
However, if a microarchitecture translates the same conditional move
as below, the translation is clearly MIL constant time.
\begin{center}
\begin{tikzpicture}[node distance=3cm,auto,>=latex']
\pomsetstyle
\picloadT{t_{1}}{\Regs}{z}
\picloadT[right of=t_{1}]{t_{2}}{\Regs}{r_2}
\picloadT[right of=t_{2}]{t_{3}}{\Regs}{r_1}
\picstoreTT[below = \pomsetnewline of t_{2}]{t_{4}}{\Regs}{r_2}{((\neg t_1 * t_{2}) + (t_{1} * t_3))}
\draw[->]  (t_{1}) edge (t_{4}) ;
\draw[->]  (t_{2}) edge (t_{4}) ;
\draw[->]  (t_{3}) edge (t_{4}) ;
\end{tikzpicture}
\end{center}

\section{Related Work}
\textbf{Speculative semantics and foundations}
Several works have recently addressed the formal foundations of specific
forms of speculation to capture Spectre-like vulnerabilities.
Cheang et al.~\cite{secspec}, Guarnieri et al.~\cite{guarnieri2018spectector}, and Mcilroy et al.~\cite{mcilroy2019spectre}
propose semantics that support branch prediction, thus modeling only Spectre v1. Neither work supports
speculation of target address, speculation of dependencies, or
OoO execution. 
Disselkoen at al.~\cite{Disselkoen2019TheCT} propose a pomset-based semantics that supports OoO
execution and branch prediction. Their model targets a higher
abstraction level modeling memory references using logical
program variables. Hence, the model cannot support dynamic dependency resolution, dependency prediction, and speculation of
target addresses.

Like us, Cauligi et al.~\cite{cauligi2019towards} propose a model that captures existing variants of Spectre and independently discover
a vulnerability similar to our Spectre-STL-D. Remarkably, they demonstrate the feasibility of the attack on Intel Broadwell and 
Skylake processors. A key difference between the two models is that Cauligi et al. impose sequential order to instruction retire and memory stores.
While simplifying the proof of memory consistency and verification, it does not reflect the inner workings of modern
CPUs, which reorder memory stores and implement a relaxed consistency model. 
These features are required to capture Spectre-OoO in
Section~\ref{sec:newvul:ooo}.
Moreover, our model provides a clean separation between the general speculative semantics and 
microarchitecture-specific features, where the latter is obtained by reducing the
nondeterminism of the former. This enables a modular analysis of (combinations of) predictive strategies, as in Spectre-PHT ICache in
Section~\ref{sec:new-attack:fetch}.

 
%
%

\textbf{Cache side channels}
In line with prior works~\cite{secspec,guarnieri2018spectector,cauligi2019towards}, our attacker model
abstracts away the mechanism used by an attacker to profile the sequence of a victim's
memory accesses, providing a general
account of trace-driven attacks~\cite{page2002theoretical}.
Complementary works~\cite{gruss2016flush,disselkoen2017prime,liu2015last,yarom2014flush}
show that cache profiling is becoming increasingly steady and precise. 
Performance jitters caused by cache usage have been widely exploited to leak sensitive data~\cite{gruss2015cache,kocher1996timing,maurice2017hello,osvik2006cache,neve2006advances,aciiccmez2006trace,yarom2017cachebleed},
e.g., in cryptography software. Miller~\cite{miller18}, and Fogh and Ertl~\cite{fogh18} propose a taxonomy for mitigating speculative execution vulnerabilities.
We refer to a recent survey by Canella et al.~\cite{ge2018survey}  on cache-based countermeasures.

\textbf{Spectre vs Meltdown}
Recent attacks that use microarchitectural effects of speculative
execution have been generally distinguished as Spectre and
Meltdown attacks~\cite{DBLP:conf/uss/CanellaB0LBOPEG19}.
We focus on the former~\cite{kiriansky2018speculative,
  DBLP:conf/sp/KocherHFGGHHLM019, koruyeh2018spectre,
  maisuradze2018ret2spec, horn2018speculative, horn2018reading,
  chen2018sgxpectre, evtyushkinalmbox,speechminer20}, which
exploits speculation to cause a victim program to transiently
access sensitive  memory locations that the attacker is not authorized to
read. 
Meltdown attacks~\cite{lipp2018meltdown} transiently bypass
the hardware security mechanisms that enforce memory isolation.
Importantly, Meltdown attacks can be easily countered in
hardware, while Spectre attacks require hardware-software co-design, which motivates our model.
We remark that the vulnerability in Section~\ref{sec:new-attack:v4} is different from
the recent Microarchitectural Data Sampling attacks~\cite{Schwarz2019ZombieLoad,ridl,canella2019fallout}, since it only requires the CPU to predict memory aliases
with no need of violating memory protection mechanisms.
Microarchitectures supporting this feature have been
proposed, e.g., in Feiste et al.~\cite{feiste}.

\textbf{Tool support}
Several prototypes have been developed to reproduce and detect known Spectre-PHT attacks~\cite{secspec,guarnieri2018spectector,wang2019kleespectre,wangooo}.
Checkmate~\cite{DBLP:conf/micro/TrippelLM18}
synthesizes proof-of-concept attacks by using models of
speculative and OoO pipelines.
Tool support for vulnerabilities beyond Spectre-PHT requires dealing
with a large number of possible predictions and instruction interleavings.
In fact, current tools mainly focus on Spectre-PHT ignoring OoO execution. 

\textbf{Functional Pipeline Correctness}
A number of authors, cf.  \cite{BurchD94,sawada2002verification,AagardCDJ01,ManoliosS05,jhala2001microarchitecture}, have studied the orthogonal problem of functional correctness in the context of concrete pipeline architectures 
involving features such as OoO and speculation, usually using a complex refinement argument based on Burch-Dill style flushing \cite{BurchD94} in order to align OoO executions with their 
sequential counterparts. Our correlate is the serialization proofs for OoO and speculation sketched in appendices \ref{thm:ooo:co} and \ref{appendix:proof:spec:cho}. 
It is of interest to mechanize these proofs and to examine if a generic account of serialization using, e.g., InSpectre can help also in the functional verification of 
concrete pipelines.

\textbf{Hardware countermeasures} While CPU vendors and researchers propose 
countermeasures, it is hard to validate their effectiveness  without a model. InSpectre can 
help modeling  and reasoning about their security guarantees, as in Section~\ref{sec:hcm}.  
Similarly, InSpectre can model the hardware configurations and fences designed by Intel~\cite{intelfences}  
to stall (part of) an instruction stream in case of speculation. 
%
Several works~\cite{Kiriansky18,taram2019context,zagieboylo2019using,weisse2019nda,woodruff2014cheri} propose security-aware hardware that prevent
Spectre-like attacks.
InSpectre can help formalizing  these hardware
features and analyzing their security.

\section{Concluding Remarks}

\changed{This paper presented InSpectre, the first comprehensive model capable of capturing out-of-order 
execution along with the different forms of speculation that could be implemented in a high-performance pipeline. We used InSpectre to model existing vulnerabilities, to discover three new potential vulnerabilities, and to
reason about the security of existing countermeasures proposed in the literature. There are a number of interesting directions left open in this work.
}

\changed{
\textbf{Foundations of microarchitecture security} We argue that InSpectre pushes the boundary on
foundations of microarchitecture security with respect to the current state-of-the-art substantially. 
Existing models~\cite{cauligi2019towards,secspec,Disselkoen2019TheCT,guarnieri2018spectector,mcilroy2019spectre} 
miss features  like dynamic inter-instruction dependency (except~\cite{cauligi2019towards}]), instruction non-atomicity, 
OoO memory commits, and partial misprediction of rollbacks. These features were essential 
to discover the vulnerabilities, as well as to reason about countermeasures like retpoline or memory fences for data dependency.
For instance, InSpectre would not have captured our Spectre-OoO vulnerability if the memory stores and instruction retire are performed in the \emph{sequential} order. 
Similarly, \emph{static} computation of active  stores would not have exposed Store-To-Load variants of Spectre. 
Moreover, forcing the rollback of all subsequent microinstructions as soon as 
a value is mispredicted  prevents modeling advanced recovery methods used by modern processors, 
including concurrent and partial recovery in case of multiple mispredictions.
}

\changed{A novel feature of our approach is to decompose instructions into smaller microinstruction-like units.
We argue that the modeling of pipelines using ISA level instructions as atomic units is in the long run the wrong approach, not reflecting well the behaviour at the hardware level, and unable to provide the foundation for real pipeline information flow. Non-atomicity is needed to handle, for instance, 
intra-instruction dependencies and interactions between I/D-caches. Therefore, decomposing instructions into smaller 
microinstructions, as we do, appears essential.
}


\changed{
InSpectre lacks explicit support of Meltdown-like vulnerabilities, multicore and hyperthreading, fences, TLBs, cache eviction policies, 
and mechanisms used to update branch predictor tables. Our model can already capture many of these features. In the paper we give Intel's \emph{lfence} as an example.
We focus here on core aspects of out-of-order and speculation, but there is nothing inherent in the framework that 
prevents modeling the above additional features. Also, by providing a general model we cannot currently argue if a 
concrete architecture is secure. For that we need to specialize the model to a given architecture, by adding detail and eliminating nondeterminism.
}

\changed{
\textbf{Tooling}
Tooling is needed to explore more systematically the utility of
the model for exploit search and countermeasure proof, and the
framework needs to be instantiated to different concrete pipeline
architectures and be experimentally validated.}

\changed{One can envisage MIL-based analysis tools like Spectector~\cite{guarnieri2018spectector}, Pitchfork~\cite{cauligi2019towards}, and oo7~\cite{wangooo}.
However, the large nondeterminism introduced by out-of-order and speculation will make such an approach inefficient. 
We are currently taking a different route by modeling concrete microarchitectures within a theorem prover. 
This allows verifying conditional noninterference if the microarchitecture is inherently secure. A failing security 
proof gives a basis for proving countermeasure soundness as in Section~\ref{sec:newvul:ooo}, and the identification 
of sufficient conditions that can be verified in the (more tractable) sequential model. 
}

%% file: app-in-order.tex
\section{In-order Semantics}\label{sec:inorder}
We define the in-order (i.e.,  sequential) semantics by restricting the
scheduling of the OoO semantics and enforcing the execution of microinstructions in program order.

A microinstruction $\inst = \ass{c}{t}{o}$ 
is \emph{completed} in state $\newstate$ (written $\completed{\newstate,
 \inst}$) if one of the following conditions hold:
\begin{itemize}
 \item The instruction's guard evaluates to false in $\newstate$, i.e.
   $\neg \den{c}(\newstate)$.
 \item The instruction has been executed and is not a memory or a 
   program counter store, i.e., 
   $o \neq  \store{\Locs}{t_a}{t_v}
    \land
    o \neq \store{\Pc}{}{t_v} \land \fdefined{\newstate}{t}
    $.
    \item The instruction is a committed memory store or a fetched and
      decoded program counter store, i.e., $t \in \commits \cup \decodes$
\end{itemize}
The in-order transition rule allows an evaluation step to proceed only
if program-order preceding microinstructions have been completed.


\begin{newnotation}
\[
\begin{array}{c}
  \newstate \doublestep{l}{}  \newstate' \qquad \stepparam(\newstate, \newstate') =
  (\alpha, t) \\
  \forall \inst \in \newstate \mbox{ if } bn(\inst) < t  \mbox{ then }
  \completed{\newstate, \inst}
  \\ \hline
  \newstate \singlestep{l}{} \newstate'
\end{array}
\]
\end{newnotation}
%
%
It is easy to show that the sequential model is deterministic. In
fact, the OoO model allows each transition to modify  
one single name $t$, while the precondition of the in-order rule forces all
previous instructions to be completed, therefore only one transition at a time
is enabled.
\begin{definition}
  Let $\newstate_1:: \dots :: \newstate_n$ be the sequence of states of execution $\pi$, then
  $\mathit{commits}(\pi, a)$ is the list of memory commits at address $a$
  in $\pi$, and is empty if $n < 2$;
  $v :: \mathit{commits}(\newstate_2 :: \dots :: \newstate_n, a)$ if
  $\stepparam(\newstate_1, \newstate_2)=(\cmt(a,v), t)$; and
  $\mathit{commits}(\newstate_2 :: \dots :: \newstate_n, a)$ otherwise.
\end{definition}
We say that two models are memory consistent if 
writes to the same memory location are seen in the same
order.
\begin{definition}
The transition systems $\rightarrow_1$ and $\rightarrow_2$ are
\emph{memory consistent} if
for any program and initial state $\newstate_0$,
for all executions $\pi = \newstate_0 \rightarrow^\ast_1 \newstate$, there exists
  $\pi' = \newstate_0 \rightarrow^\ast_2 \newstate'$
  such that for all $a \in \Locs$ $\mathit{commits}(\pi, a)$ is a prefix of $\mathit{commits}(\pi', a)$.
\end{definition}
Intuitively, two models that are memory consistent yield the same
sequence of memory updates for each memory address. This
ensures that the final result of a program is the same in both models.
Notice that since we do not assume any fairness property for the transition systems then an execution $\pi$ of $\rightarrow_1$ may indefinitely postpone the commits for a given address. For this reason we only require
to find an execution such that
$\mathit{commits}(\pi, a)$ is a prefix of $\mathit{commits}(\pi', a)$.
We obtain memory consistency of both the OoO and the speculative semantics against the in-order semantics.

\begin{theorem}\label{thm:ooo:co}
 $\doublestep{}{}$ and $\singlestep{}{}$ are memory consistent.\hfill$\Box$
\end{theorem}
\begin{proof}
See Appendix~\ref{prf:oooco}.
\end{proof}

\begin{theorem}
\label{thm:spec:co}
$\triplestep{}{}$ and $\singlestep{}{}$ are memory consistent.\hfill$\Box$
\end{theorem}
\begin{proof}
See Appendix~\ref{prf:speco}.
\end{proof}



%% file: app-other-counter.tex
\subsection{Spectre-BTB and Spectre-RSB}\label{appx:cm}
Two variants of Spectre attacks~\cite{DBLP:conf/uss/CanellaB0LBOPEG19}  exploit a CPU's prediction mechanism for
jump targets to  leak sensitive data. 
In particular, Spectre-BTB~\cite{DBLP:conf/sp/KocherHFGGHHLM019} (Branch Target Buffer) poisons the prediction
of indirect jump targets.
To model this prediction strategy we assume a function $ijmps(\Inst)$
that extracts all PC stores resulting from the translation
of indirect jumps. This can be
accomplished by making the translation of these
instructions syntactically distinguishable from other control flow
updates.
As a result, prediction is possible for all indirect jumps whose
address is yet to be resolved: Namely, 
 $\prediction_{BTB}(\Inst, \storage, \commits, \decodes, \spec,
 \guesses)=$
\[
  \left \{
    t_a \mapsto v \mid
    \ass{c}{t}{\store{\Pc}{}{t_a}} \in ijmps(\Inst) \land
    \fundefined{\storage}{t_a}
  \right  \}
\]
We do not restrict the possible predicted values $v$, since
an accurate model of jump prediction requires knowing the strategy used by the CPU to update
the BTB buffer.

Spectre-RSB~\cite{sv2,maisuradze2018ret2spec} poisons the Return Stack
Buffer (RSB), which is used to temporally store the 
$N$ most recent return addresses: \stylecode{call} instructions push the return address on the RSB,
while \stylecode{ret} instructions pop from the RSB to predict the return target.
A misprediction can happen if: ($i$)  a return address on the
stack has been explicitly overwritten, e.g., when a program
handles a software exception using \stylecode{longjmp} instructions,
or, ($ii$) returning from a call stack deeper than $N$, the RSB is
empty and the CPU uses the same prediction as for the other indirect jumps.
We model
\stylecode{call} and \stylecode{ret} instructions  via program counter stores. A \stylecode{call} to address $b_1$ from
address $a_1$ can be modeled as\\
\begin{tikzpicture}[node distance=3cm,auto,>=latex']
\pomsetstyle
\node[name=a_1]{$a_1$};
\picloadT[right =.3cm of a_1]{t_{11}}{\Regs}{sp}
\picexpT[right of = t_{11}]{t_{12}}{t_{11} - 4}
\picstoreT[right of = t_{12}]{t_{13}}{\Regs}{sp}{t_{12}}
\picexpT[below = 0.4cm of t_{11}]{t_{14}}{a_1+4}
\picstoreT[right of = t_{14}]{t_{15}}{\Locs}{t_{11}}{t_{14}}
\picstoreT[right of = t_{15}]{t_{16}}{\Pc}{}{b_1}
\draw[->]  (t_{11}) edge (t_{12}) ;
\draw[->]  (t_{12}) edge (t_{13}) ;
\draw[->]  (t_{11}) edge (t_{15}) ;
\draw[->]  (t_{14}) edge (t_{15}) ;
\end{tikzpicture}\\
The \stylecode{call} instruction saves (e.g. $t_{15}$) the return address (e.g. $a_1+4$)
into the stack, decreases the stack pointer (e.g. $t_{13}$), and jumps to
address $b_1$ (e.g. $t_{16}$).

A \stylecode{ret} instruction  from
address $a_2$ can be modeled as\\
\begin{tikzpicture}[node distance=3cm,auto,>=latex']
\pomsetstyle
\node[name=a_2]{$a_2$};
\picloadT[right =.3cm of a_1]{t_{21}}{\Regs}{sp}
\picexpT[right of = t_{21}]{t_{22}}{t_{21} + 4}
\picstoreT[right of = t_{22}]{t_{23}}{\Regs}{sp}{t_{22}}
\picloadT[below = 0.4cm of t_{21}]{t_{24}}{\Locs}{t_{22}}
\picexpT[right of = t_{24}]{t_{25}}{t_{24}}
\picstoreT[right of = t_{25}]{t_{26}}{\Pc}{}{t_{25}}
\draw[->]  (t_{21}) edge (t_{22}) ;
\draw[->]  (t_{22}) edge (t_{23}) ;
\draw[->]  (t_{22}) edge (t_{24}) ;
\draw[->]  (t_{24}) edge (t_{25}) ;
\draw[->]  (t_{25}) edge (t_{26}) ;
\end{tikzpicture}\\
The instruction loads  the return address
from the stack ( $t_{24}$), increases the stack pointer ($t_{23}$), and returns
($t_{26}$).

\mathchardef\mhyphen="2D
\newcommand{\retaddress}{\mathit{ret\mhyphen ra}}
We assume functions $calls(\Inst)$ and
$rets(\Inst)$ to extract the
PC stores that belong to
a \stylecode{call} and
\stylecode{ret} respectively. Moreover, if $t \in bn(calls(\Inst))$,
we use $\retaddress(\Inst, t)$ to retrieve name of the microinstruction that saves the return address (e.g $t_{15}$) of
the corresponding \stylecode{call}.
We  model return address prediction as   
\begin{tabbing}
\quad \= \qquad \quad \=\kill
$\prediction_{RSB}(\Inst, \storage, \commits, \decodes, \spec, \guesses)=$ \\
 \> $\{ t_a \mapsto v \mid \ass{c}{t}{\store{\Pc}{}{t_a}} \in rets(\Inst)
      \land \fundefined{\storage}{t_a}\: \land $ \\
 \> \> $\exists t' \in bn(calls(\Inst)) . \: t' < t \land \:
      \storage(\retaddress(\Inst, t')) = v\
      \land$ \\
 \> \> $\textit{RSB-depth}(\Inst, t',t) \subseteq \{1 \dots N\}\}$
\end{tabbing}
%
Prediction is possible only for \stylecode{ret} microinstructions $t$ that
have a prior matching \stylecode{call} $t'$, provided that the size of
intermediary stack depth is between $1$ and  $N$.
We define the latter as the set
$
\textit{RSB-depth}(\Inst, t', t) = \{
\#(
  bn(calls(\Inst)) \cap \{t' \dots t''\}
  ) - \#(
  bn(rets(\Inst)) \cap \{t' \dots t''\}
  ) \mid t' \leq t'' < t
\}
$, where $\{t' \dots t''\}$ is an arbitrary continuous sequence of
 names starting from $t'$ and ending before $t''$, and
$\#(
  bn(calls(\Inst)) \cap \{t' \dots t''\}
  )$ and 
  $\#(
  bn(rets(\Inst)) \cap \{t' \dots t''\}
  )$ count the number of \stylecode{call}s and \stylecode{ret}s in the
  sequence respectively.
The prediction consists in
assuming the target address (e.g. $t_a$) of the \stylecode{ret} to be equal to the return
address (e.g. $v$) that has been pushed into the stack by the matching \stylecode{call}.

In some microarchitectures (e.g.,~\cite{skylane}), 
 RSB prediction
falls back to the BTB in case of underflow, i.e., when returing from a
function with nested stack deeper than $N$. This case can be
modeled by considering the prediction strategy $\prediction_{RSB/BTB}$  defined as
\begin{tabbing}
\quad \= \qquad \=\kill
  $\prediction_{RSB} \bigcup \{ t_a \mapsto v \mid \ass{c}{t}{\store{\Pc}{}{t_a}} \in rets(\Inst)
    \land \fundefined{\storage}{t_a}\: \land $\\ 
  \> $\exists t' \in bn(calls(\Inst)) .\: t' < t \land \:
  \textit{RSB-depth}(\Inst, t',t) = \{1 \dots N'\} \land \: N' > N\}$
\end{tabbing}

The following  example shows how jump target prediction
may violate the security condition.
Consider the program \stylecode{*p:=\&f; (*p)()} that saves the address 
of a function (i.e., $\&f$) in a
function pointer at constant address $p$ and immediately invokes the function.
Assuming that these instructions are stored at addresses $a_1$ and $a_2$, their MIL
translation is:


\begin{pomset}
\begin{center}
\begin{tikzpicture}[node distance=3cm,auto,>=latex']
\pomsetstyle
\node[name=a'_1]{$a_1$};
\picstoreT[right =.3cm of a'_1]{t_{11}}{\Locs}{p}{\&f}
\picstoreT[right of = t_{11}]{t_{12}}{\Pc}{}{a_2}
\node[name=a'_2,below =\pomsetnewline of a'_1]{$a_2$};
\picloadT[right =.3cm of a'_2]{t_{21}}{\Locs}{p}
\picexpT[right of = t_{21}]{t_{22}}{t_{21}}
\picstoreT[right of = t_{22}]{t_{23}}{\Pc}{}{t_{22}}
\draw[->]  (t_{21}) edge (t_{22}) ;
\draw[->]  (t_{22}) edge (t_{23}) ;
\end{tikzpicture}
\end{center}
\caption{\stylecode{*p:=\&f; (*p)()}}
\label{example:spectre:v2}
\end{pomset}


Because our semantics can predict only internal operations (see rule $\prd$),  
the translation function introduces an additional internal operation, i.e., $t_{22}$ which allows 
predicting the value of the load $t_{21}$.

Suppose that the function $f$ simply returns and 
the security policy labels all data, except 
the program counter, as sensitive. The program is secure (at the
ISA level) 
as it always transfers control to $f$,
producing the sequence of observations
$\il{a_1}::\ds{p}::\il{a_2}::\dl{p}::\il{\&f}$ independently of the
initial state.

Jump target prediction produces a different behavior.
Let $\newstate_0$ be the state containing only the translation of the
instruction in $a_1$.
Initially, $\prediction_{BTB}(\newstate_0)$ is empty since the state
contains no PC updates (e.g. $t_{12}$) that result from
translating indirect jumps.
The CPU may execute and fetch $t_{12}$, thus adding
$t_{21}$, $t_{22}$, and $t_{23}$ to the set of microinstructions $\Inst$.
In the resulting state $\prediction_{BTB}$ is $\{t_{22} \mapsto v
\mid v \in \Val\}$, since $t_{23}$ models an indirect jump and
$t_{22}$ has not been executed.
The CPU can therefore predict the value of $t_{22}$ without
waiting for the result of the load $t_{21}$.
If the predicted value is the address $g$ of the instruction
\stylecode{$r_1$:=*($r_2$)} the misprediction can use $g$ as 
gadget to leak sensitive information.
\begin{center}
\begin{tikzpicture}[node distance=3cm,auto,>=latex']
\pomsetstyle
\node[name=g]{$g$};
\picloadT[right =.3cm of g]{t_{31}}{\Regs}{r_2}
\picloadT[right of = t_{31}]{t_{32}}{\Locs}{r_1}
\picstoreT[right of = t_{32}]{t_{33}}{\Regs}{r_1}{t_{32}}
\draw[->]  (t_{31}) edge (t_{32}) ;
\draw[->]  (t_{32}) edge (t_{33}) ;
\end{tikzpicture}
\end{center}
In fact, the speculative semantics can produce the sequence of
observations $\il{a_1}::\ds{p}::\il{a_2}::\dl{p}::\il{g}::\dl{v}$,
where $v$ is the initial value of register $r_2$. The last
observation  of the sequence allows an attacker to learn sensitive data.
Observe that this leak is readily captured by the security condition, since
such observation sequence is not possible in the sequential semantics. 

\subsubsection{Countermeasure: Retpoline}
A known countermeasure to Spectre-BTB is the \emph{Retpoline} technique developed by Google~\cite{retpoline}. 
In a nutshell, retpolines are instruction snippets that isolate indirect jumps from speculative execution
via call and return instructions.
Retpoline has the effect of  transforming  indirect jumps at address
$a_2$ of Example~\ref{example:spectre:v2} as:\\
\begin{tikzpicture}[node distance=3cm,auto,>=latex']
\pomsetstyle
\node[name=a_2]{$a_2$};
\picloadT[right =.3cm of a_2]{t_{21}}{\Regs}{sp}
\picexpT[right of = t_{21}]{t_{22}}{t_{21} - 4}
\picstoreT[right of = t_{22}]{t_{23}}{\Regs}{sp}{t_{22}}
\picexpT[below = \pomsetnewline of t_{21}]{t_{24}}{a_{3}}
\picstoreT[right of = t_{24}]{t_{25}}{\Locs}{t_{21}}{t_{24}}
\picstoreT[right of = t_{25}]{t_{26}}{\pc}{}{b_1}
\draw[->]  (t_{21}) edge (t_{22}) ;
\draw[->]  (t_{22}) edge (t_{23}) ;
\draw[->]  (t_{21}) edge (t_{25}) ;
\draw[->]  (t_{24}) edge (t_{25}) ;
\end{tikzpicture}\\
\begin{tikzpicture}[node distance=3.5cm,auto,>=latex']
\pomsetstyle
\node[name=a_1]{$a_3$};
\picstoreT[right =.3cm of a_1]{t_{31}}{\Pc}{}{a_3}
\end{tikzpicture}\\
\begin{tikzpicture}[node distance=3cm,auto,>=latex']
\pomsetstyle
\node[name=b_1]{$b_1$};
\picloadT[right =.3cm of b_1]{t_{41}}{\Regs}{sp}
\picexpT[right of = t_{41}]{t_{42}}{t_{41} + 4}
\picloadT[below = \pomsetnewline of t_{41}]{t_{43}}{\Locs}{p}
\picexpT[right of = t_{43}]{t_{44}}{t_{43}}
\picstoreT[right of = t_{42}]{t_{45}}{\Locs}{t_{42}}{t_{44}}
\picstoreT[right of = t_{44}]{t_{46}}{\Pc}{}{b_2}
\draw[->]  (t_{41}) edge (t_{42}) ;
\draw[->]  (t_{42}) edge (t_{45}) ;
\draw[->]  (t_{43}) edge (t_{44}) ;
\draw[->]  (t_{44}) edge (t_{45}) ;
\end{tikzpicture}\\
\begin{tikzpicture}[node distance=3cm,auto,>=latex']
\pomsetstyle
\node[name=b_2]{$b_2$};
\picloadT[right =.3cm of a_1]{t_{51}}{\Regs}{sp}
\picexpT[right of = t_{51}]{t_{52}}{t_{51} + 4}
\picstoreT[right of = t_{52}]{t_{53}}{\Regs}{sp}{t_{52}}
\picloadT[below = \pomsetnewline of t_{51}]{t_{54}}{\Locs}{t_{52}}
\picexpT[right of = t_{54}]{t_{55}}{t_{54}}
\picstoreT[right of = t_{55}]{t_{56}}{\Pc}{}{t_{55}}
\draw[->]  (t_{51}) edge (t_{52}) ;
\draw[->]  (t_{52}) edge (t_{53}) ;
\draw[->]  (t_{52}) edge (t_{54}) ;
\draw[->]  (t_{54}) edge (t_{55}) ;
\draw[->]  (t_{55}) edge (t_{56}) ;
\end{tikzpicture}\\
Instruction at $a_2$ calls a \emph{trampoline} starting at address $b_1$ and
instruction at $a_3$ loops indefinitely.
The first instruction of the trampoline overwrites the return address on the stack with
the value at address $p$ and its second instruction at $b_2$
returns.

We leverage our model to analyze the effectiveness of Retpoline for
indirect jumps.
Since address $b_1$ is known at compile time, $t_{26}$ does not trigger a
jump target prediction.
While executing the trampoline, the value of $t_{55}$ may be mispredicted,
especially if the load from $p$ has not been executed and the store $t_{45}$
is postponed.
However, $b_2$ is a \verb|ret|, hence
the value of $t_{55}$ is predicted via $\prediction_{RSB}$. Since there is no
call between $a_1$ and $b_2$, then prediction can only assign the address $a_3$ to
$t_{55}$ (i.e., $\proj{\prediction_{RSB}}{t_{55}} \subseteq \{t_{55} \rightarrow a_3\}$).
Therefore, the RSB entry generated by $a_2$ is used
and mispredictions are captured with the infinite loop in $a_3$. Ultimately, when
the value of $t_{55}$ is resolved, the correct return address is used
and the control flow is redirected to the value of $*p$, as expected.


%% file: proof_t_eq.tex
\subsection{Correctness of $t$-equivalence: Lemma~\ref{lem:deps-eq}}\label{proof:teq}
  If $\newstate_1  \sim_t  \newstate_2$ and the $t$'s microinstruction  
  in $\newstate_1$ is $\instr = \ass{c}{t}{o}$, then
  $\deps(t, \newstate_1) = \deps(t, \newstate_2)$,
  $\den{c}\newstate_1 = \den{c}\newstate_2$,
  and if $\den{\inst}\newstate_1 = (v_1, l_1)$ and $\den{\inst}\newstate_2 =
  (v_2, l_2)$ then $v_1 = v_2$. 

  \begin{proof}
    For non-load operations and guards the proof is trivial, since their semantics only
  depends on the value of the bound names of the operation and
  guard. These names are statically identified from $c$ and $o$ and their value is the same in $\newstate_1$ and
  $\newstate_2$ by definition of $\sim_t$.
  For loads, the proof relies on
  showing that for every state $\newstate$, 
  $\stract(\proj{\newstate}{T}, t)
  =
  \stract(\proj{\newstate}{T}, t)
  $, where $T=\deps(t, \newstate)$.
  Let $t' \in \stract(\newstate, t)$. By
  definition, $t'$ and its bound names are in $T$, therefore their
  values are equal in $\newstate$ and $\proj{\newstate}{T}$ and hence $t'
  \in \strfor(\proj{\newstate}{T},t)$. Also, by definition of
  $\deps$, all names referred to by conditions and addresses of
  subsequent stores of $t'$ are in $T$. Therefore
  if there is no subsequent store
  that overwrites $t'$ (i.e.,  $\ass{c''}{t''}{\store{\restype}{t''_a}{t''_v}}$ such that 
  $[c'']\newstate$ and $(\newstate(t''_a) = \newstate(t) \vee
  \newstate(t''_a) = \newstate(t_a))$) then
  there is no store overwriting $t'$ in  $\proj{\newstate}{T}$ and
  hence $t' \in \stract(\proj{\newstate}{T}, t)$.
\end{proof}


%% file: proof_ooo_consistency.tex
\subsection{OoO Memory Consistency: Theorem~\ref{thm:ooo:co}}\label{prf:oooco}
To prove that $\doublestep{}{×}$ and $\singlestep{}{}$ are memory consistent we
demonstrate a reordering lemma, which allows to
commute transitions
if the $(n+1)$-th transition modified  
name $t_2$, 
$n$-th transition modified name $t_1$, and $t_2 < t_1$.
\begin{center}
\begin{tikzpicture}[node distance=2cm,auto,>=latex']

\tikzstyle{every node}=[font=\small]
\tikzset{pblock/.style = {rectangle, draw=white!50, top
                      color=white,bottom color=white, align=center, minimum width=1cm}}
                      
\node[pblock,name=a1]{$\newstate_1$};                      
\node[pblock,name=a2, right of  = a1]{$\newstate_2$};
\node[pblock,name=a3, right of  = a2]{$\newstate_n$};
\node[pblock,name=a4, right of  = a3]{$\newstate_{n+1}$};
\node[pblock,name=a4', below = 0.5cm of a4]{$\newstate'_{n+1}$};
\node[pblock,name=a5, right of  = a4]{$\newstate_{n+2}$};
                
\draw[->]  (a1) edge  node[above] {} (a2);
\draw[->]  (a2) edge[dashed]  node[midway,sloped,rotate=90]{} (a3); 
\draw[->]  (a3) edge node[above,sloped] {$(\alpha_1, t_1)$}  (a4);
\draw[->]  (a4) edge node[above,sloped] {$(\alpha_2, t_2)$}  (a5);
\draw[->]  (a3) edge node[below left] {$(\alpha_2, t_2)$}  (a4');
\draw[->]  (a4') edge node[below right] {$(\alpha_1, t_1)$}  (a5);
\end{tikzpicture}
\end{center}

We use the following notation. Let
$\newstate_0 :: \dots ::\newstate_n$ a sequence of states,
we define $(\Inst_i, \storage_i, \commits_i, \decodes_i) = \newstate_i$
for $i \in \{0 \dots n\}$,
$(\alpha_i, t_i) = \stepparam(\newstate_{i-1}, \newstate_{i})$
and $(\Inst_{i-1} \cup \hat{\Inst}_{i}, \storage_{i-1} \cup \hat \storage_{i},
\commits_{i-1} \cup \hat \commits_{i}, \decodes_{i-1} \cup \hat \decodes_{i}) =
(\Inst_{i}, \storage_{i}, \commits_{i}, \decodes_{i})$ for $i \in
\{1 \dots n\}$.

We first demonstrate that $\strfor$ and $\stract$ of a microinstruction
$t$ do not depend on names bigger that $t$ and that they are monotonic.
\begin{lemma}\label{lemma:activestore:reorder1}
  Let $\newstate_0$ and $\newstate_1$ be two states, 
  if $bn(\hat \Inst_1) \geq t$ and $\dom{\hat \storage_i} \geq t$ then
  $\strfor(\newstate_1,t) = \strfor(\newstate_0,t)$ and
  $\stract(\newstate_1, t) =  \stract(\newstate_0, t)$
\end{lemma}
\begin{proof}
Let $t$ be a load or store accessing address $t_a$ (the other cases are
trivial, since $\stract$ is undefined), hence $t_a < t$.\\
(1) The set of instructions that precedes $t$ is the same in $\newstate_0$ and
$\newstate_1$. In fact, since $bn(\hat \Inst_1) \geq t$ then
$\{\ass{c'}{t'}{o} \in \Inst_0 \cup \hat \Inst_1 \mid  t' < t\} =
\{\ass{c'}{t'}{o} \in \Inst_0 \mid  t' < t\}$.\\
(2) For every store that precedes $t$, the evaluation of condition and
address is the same in $\newstate_0$ and $\newstate_1$. In fact,
let $\iota' = \ass{c'}{t'}{\store{\restype}{t'_a}{t'_v}}$ and $t' < t$
then $n(c') \cup \{t'_a\} < t$. Therefore, $\den{c'}{\storage_0 \cup
\hat \storage_1} = \den{c'}{\storage_0}$ and $(\storage_0 \cup
\hat \storage_1)(t'_a) = \storage_0(t'_a)$.\\
(3) Similarly, since $t_a < t$ then $(\storage_0 \cup
\hat \storage_1)(t_a) = \storage_0(t_a)$.\\
Properties (1, 2, 3) guarantee that $\strfor(\newstate_1,t) =
\strfor(\newstate_0,t)$. Similarly, since $\stract$ depends on the
addresses and conditions of stores in $\strfor(\newstate_1,t)$ and these have 
names smaller than $t$ then
$\stract(\newstate_1,t) =
\stract(\newstate_0,t)$.
\end{proof}
\begin{lemma}\label{lemma:pas:mono}
  if $\newstate_0 \doublestep{l}{} \newstate_1$ and $\ass{c}{t}{o} \in \newstate$ then
  $\strfor(\newstate_1, t) \subseteq \strfor(\newstate_0, t)$ and
  $\stract(\newstate_1, t) \subseteq \stract(\newstate_0, t)$
\end{lemma}
\begin{proof}
  The proof is done by case analysis on $\alpha_1$.
  For commits the proof is trivial, since $\hat \Inst_1 =
  \emptyset$ and $\hat \storage_1 = \emptyset$.
  For fetches, $\hat \storage_1 = \emptyset$ and the transition may decode new
  stores in $\hat \Inst_1$. However, these new stores
  have names greater than $max(\Inst_0)$, hence their names are greater than
  $t$. Therefore the additional stores do not affect $\strfor$ and 
  $\stract$.
  
  For executions, $\hat \Inst_1 = \emptyset$, $\hat \storage_1 =
  \{t_1 \mapsto v\}$, and $\fundefined{\storage_0}{t_1}$ for some $v$.
  This store update can make defined the evaluation of the
  condition or expression of a store. In this case, if a store is in
  $\strfor(\newstate_1, t)$ it must also be in $\strfor(\newstate_0, t)$.
  Stores that are in $\strfor(\newstate_0, t)$ but are not in
  $\strfor(\newstate_1, t)$ have undefined conditions or addresses in
  $\newstate_0$ and false condition or non matching address in $\newstate_1$.\\
  To show that $\stract$ does not increase we proceed as follows.
  Let $t'$ be a store in $\strfor(\newstate_0, t) \setminus
  \stract(\newstate_0, t)$.
  There must be a subsequent overwriting store $t''$ in $\strfor(\newstate_0, t)$
  whose condition holds in $\newstate_0$ and address is defined in
  $\newstate_0$.
  Such store cannot have $t_1$ in its free names, hence it is also in
  $\strfor(\newstate_1, t)$. Therefore the store $t''$ overwrites $t'$ in
  $\newstate_1$ too.
\end{proof}

Proof of Theorem~\ref{thm:ooo:co} is done by induction on the length
of traces and relies on Lemma~\ref{lemma:switch-pre} to demonstrate
that $(\alpha_2, t_2)$ can be applied in $\newstate_n$ and
Lemma~\ref{lemma:switch-post} to show that $(\alpha_1, t_1)$ can be
applied in the resulting state $\newstate'_{n+1}$ to obtain
$\newstate_{n+2}$.
\begin{lemma}\label{lemma:switch-pre}
If $\newstate_0 \doublestep{l_1}{} \newstate_1 \doublestep{l_2}{} \newstate_2$ and 
$t_2 < t_1$ then exists $l'_{2}$ such that $\newstate_0 \doublestep{l'_2}{}
(\Inst_{0} \cup \hat \Inst_{2}, \storage_{0} \cup \hat \storage_{2},
\commits_{0} \cup \hat \commits_{2}, \decodes_{0} \cup \hat \decodes_{2}) =
\newstate'$,
$\stepparam(\newstate_0, \newstate') = (\alpha_2, t_2)$, and if $\alpha_2 =
\cmt(a, v)$ then $\alpha_1 \neq \cmt(a, v')$.
\end{lemma}
\begin{proof}
  We fist bound the effects of the transitions that modified $t_1$.\\
(1) If $\newstate_1 \doublestep{l_2}{} \newstate_2$ and 
$\stepparam(\newstate_1, \newstate_2) = (\alpha_2, t_2)$,
then exists  $\inst_2 = \ass{c}{t_2}{o} \in I_1$. Since
$\stepparam(\newstate_0,\newstate_1) = (\alpha_1, t_1)$
then $bn(\hat I_{1}) > t_1 > t_2$. Therefore $\ass{c}{t_2}{o} \in I_0$.\\
(2) Similarly, $\dom{\hat \storage_{1}} \subseteq \{t_1\}$, $\hat \commits_{1} \subseteq
  \{t_1\}$, and $\hat \decodes_{1} \subseteq \{t_1\}$.
  
The proof continues by case analysis over the transition rule $\alpha_2$.
\\
\textbf{(Case $\exe$)}
The hypothesis of the rule ensure that 
$\fundefined{\storage_1}{t_2}$,  $\den{c}{\storage_1}$, and
$\den{\inst_2}\newstate_1 = (v, l_2)$.
The conclusion of the rule ensures that $\hat \storage_{2} = \{t_2 \mapsto v\}$,
  $\hat \commits_{2}=\emptyset$, $\hat \decodes_{2}=\emptyset$, and
  $\hat I_{2} = \emptyset$.
  \\
The proof that $t_2$ can be executed in $\newstate_0$ relies on
the fact that all free names of the instruction $t_2$ must be smaller
than $t_2$.
Property (2), $fn(\inst_2) < t_2$, and $t_2 < t_1$ ensure that
$\fundefined{\storage}{t_2}$ and  $\den{c}{\storage}$.\\
The same properties guarantee that 
$\den{\inst_2}\newstate_0 = (v, l_2)$. For internal operations and
stores the proof is trivial,
since $fn(\inst_2) < t_2$, and $t_2 <
t_1$.   The proof for loads uses Lemma~\ref{lemma:activestore:reorder1}
to guarantee that $\stract(\newstate_0, t_2) = \stract(\newstate_1, t_2)$.
  \\
  Hence we can apply rule (exec) to show that
  exists $l'_{2}$ such that
  $\newstate_0 \doublestep{l_2}{} (\Inst_0, \storage_0 \cup \{t_2 \mapsto v\}, \commits_0,
  \decodes_0) = \newstate'$.
  Notice that $l'_{2} \neq l_{2}$. In fact, if $t_{2}$ is a load, $t_{1}$ is the corresponding active store,
  and $\alpha_{1} = \cmt(a,v)$ then the execution of $t_{2}$ needs to access the memory subsystem in $\newstate_{1}$ while it can
  simply forward the value of $t_{1}$ in $\newstate_{0}$: i.e., an observable load becomes silent.
  \\
  \textbf{(Case $\cmt(a,v)$)}
  In this case $o = \store{\Locs}{t_a}{t_v}$.
  The hypothesis of the rule ensure that
  $\storage_1(t_2) = v$, $t_2 \not \in \commits_1$,
  $\bn(\strfor(\newstate_1, t_2)) \subseteq \commits_1$,
  and $\storage_1(t_a) = a$.
  The conclusion of the rule ensures that $\hat \storage_2 = \emptyset$,
  $\hat \commits_2= \{t_2\}$, $\hat \decodes_2=\emptyset$, 
  $\hat I_{2} = \emptyset$, and $l_2 = \ds{a}$.\\
Property (2) and $t_a < t_2 < t_1$ ensure that
$\storage_0(t_2) = v$, $t_2 \not \in \commits_0$, and
$\storage_0(t_a) = a$.
Similarly to the case $\exe$-load,  Lemma~\ref{lemma:activestore:reorder1}
guarantees that $\bn(\strfor(\newstate_1, t_2)) = \bn(\strfor(\newstate,
t_2))$. Since the $\strfor$ are smaller than $t_2$ then 
$\bn(\strfor(\newstate_0, t_2)) \subseteq \commits_0$.
\\
  Hence we can apply rule (commit) to show that
  $\newstate_0 \doublestep{l_2}{} (\Inst_0, \storage_0, \commits_0 \cup \{t_2\},
  \decodes_0) = \newstate'$. Finally, to show that $\alpha_1 \neq \cmt(a,
  v')$ we proceed by contradiction. If $\alpha_1 = \cmt(a, v')$, then
  $\bn(\strfor(\newstate_0, t_1)) \subseteq \commits_0$. However, $t_2
  \in \bn(\strfor(\newstate_0, t_1))$, because they write the same
  address $a$ and $t_2 < t_1$. This contradict that $t_2 \not \in \commits_0$.
  \\
  \textbf{(Case $\ftc$)}
  In this case $o = \store{\Pc}{}{t_v}$.
  The hypothesis of the rule ensure that
  $\storage_1(t_2) = v$, $t_2 \not \in \decodes_1$, and
  $\bn(\strfor(\newstate_1, t_2)) \subseteq \decodes_1$.
  The conclusion of the rule ensures that $\hat \storage_2 = \emptyset$,
  $\hat \commits_2= \emptyset$, $\hat \decodes_2= \{t_2\}$, 
  $\hat I_{2} = \translate(v, max(\Inst_1))$, and $l_2 = \il{a}$.\\
Property (2) and $fn(c) \cup \{t_a\} < t_2 < t_1$ ensure that
$\storage_0(t_2) = v$ and $t_2 \not \in \decodes_0$.
Similarly to the case commit,  Lemma~\ref{lemma:activestore:reorder1}
guarantee that $\bn(\strfor(\newstate_1, t_2)) = \bn(\strfor(\newstate,
t_2))$. Since the $\strfor$ are smaller than $t_2$ then 
$\bn(\strfor(\newstate_0, t_2)) \subseteq \decodes_0$.
\\
To complete the proof we must show that $\hat \Inst_1 = \emptyset$.
We proceed by contradiction: if $\hat \Inst_1 \neq \emptyset$ then
$\alpha_1 = \ftc(\hat \Inst_1)$, hence this transition
fetched $t_1$ and $\bn(\strfor(\newstate_0, t_1)) \subseteq
\decodes_0$. However, $t_2 < t_1$ and both update the program counter,
therefore $t_2 \in \bn(\strfor(\newstate_0, t_1))$. This contradicts
$t_2 \not \in \decodes_0$.\\
Finally,  we can apply rule $\ftc$ to show that
$\newstate_0 \doublestep{l_2}{}$\\$(\Inst_0 \cup \translate(v, max(\Inst_0), \storage_0, \commits_0,
\decodes_0 \cup \{t_2\}) = $\\$(\Inst_0 \cup \translate(v, max(\Inst_1), \storage_0, \commits_0,
\decodes_0 \cup \{t_2\}) = \newstate'$.
\end{proof}

\begin{lemma}
  \label{lemma:switch-post}
  If $\newstate_0 \doublestep{l_1}{} \newstate_1 \doublestep{l_2}{} \newstate_2$
  and $t_2 < t_1$ then there exists $\newstate'$, $l'_{1}$, and $l'_{2}$
 such that
 $\newstate_0 \doublestep{l'_2}{} \newstate' \doublestep{l'_1}{} \newstate_2$,
 $\stepparam(\newstate_0, \newstate') = (\alpha_2, t_2)$,
  $\stepparam(\newstate', \newstate_2) = (\alpha_1, t_1)$,
  and if $\alpha_2 = \cmt(a, v)$ then $l_1 \neq \cmt(a, v')$.
\end{lemma}
\begin{proof}
The existence of $\newstate'$ and $l'_{2}$ is given by Lemma~\ref{lemma:switch-pre}.
For transition $\newstate' \doublestep{l'_1}{} \newstate_2$
  we first bound the effects of $\newstate_0 \doublestep{l_1}{} \newstate_1$.\\
(1) If $\newstate_0 \doublestep{l_1}{} \newstate_1$ and 
$\stepparam(\newstate_0, \newstate_1) = (\alpha_1, t_1)$,
then there exists  $\inst_1 = \ass{c}{t_1}{o} \in \Inst_0$.  Therefore
$\ass{c}{t_1}{o} \in \Inst_0 \cup \hat \Inst_2 = \Inst'$.\\
(2) $\dom{\hat \storage_2} \subseteq \{t_2\}$, $\hat \commits_2 \subseteq
  \{t_2\}$, and $\hat \decodes_2 \subseteq \{t_2\}$.\\
We continue the proof by case analysis on
  $\alpha_2$.
  \\
  \textbf{(Case $\exe$)}
  The hypothesis of the rule ensure that 
$\fundefined{\storage_0}{t_1}$,  $\den{c}{\storage_0}$, and
$\den{\inst_1}\newstate_0 = (v, l_1)$.
The conclusion of the rule ensures that $\hat \storage_1 = \{t_1 \mapsto v\}$,
  $\hat \commits_1=\emptyset$, $\hat \decodes_1=\emptyset$, and
  $\hat I_{1} = \emptyset$.\\
  (3) Property (2) and $t_2 < t_1$ ensure that
  $\fundefined{\storage'}{t_1}$.
  Moreover, if $\hat \storage_2 = \emptyset$ then $\den{c}{\storage'}
  = \den{c}{\storage_0}$.
  Otherwise, $\hat \storage_2 = \{ t_2 \mapsto v_2\}$ for some $v_2$. In
  this case, $\alpha_2 = \exe$ and $\fundefined{\storage_0}{t_2}$.
  Since $\fdefined{\den{c}}{\storage_0}$ then
  $t_2 \not \in n(c)$, hence $\den{c}{\storage'}
  = \den{c}{\storage_0}$.
  \\
  The same argument is used to guarantee that 
  $\den{\inst_1}\newstate' = (v, l'_1)$. For internal operations and
  stores the proof follows the same approach of (3).  The proof for loads uses
  Lemma~\ref{lemma:pas:mono}. Notice that for internal operations and
  stores $l'_1 = l_1 = \cdot$. For loads, either 
  $l'_1= l_1 = \dl a$ or $l'_1= \dl a$ and $l'_1= \cdot$. The latter
  happens when $\stepparam(\newstate_0, \newstate') = (\cmt(a, v), t_2)$. In
  this case, we have reordered a memory commit before a load thus
  making it not possible to forward the value of the store to the load
  and requiring a new memory interaction. This shows that the
  observations of the OoO model are a subset of the
  observations of in-order model, since the OoO model can
  execute loads before the corresponding stores are committed.
  \\
  Finally, we can apply rule $\exe$ to show that
  $\newstate' \doublestep{l'_1}{} (\Inst', \storage' \cup \{t_1 \mapsto v\}, \commits',
  \decodes') = \newstate_2$.
  \\
  \textbf{(Case $\cmt(a,v)$)}
  In this case $o = \store{\Locs}{t_a}{t_v}$.
  The hypothesis of the rule ensure that
  $\storage_0(t_1) = v$, $t_1 \not \in \commits_0$,
  $\bn(\strfor(\newstate_0, t_1)) \subseteq \commits_0$,
  and $\storage_0(t_a) = a$.
  The conclusion of the rule ensures that $\hat \storage_1 = \emptyset$,
  $\hat \commits_1= \{t_1\}$, $\hat \decodes_1=\emptyset$, 
  $I_{t_1} = \emptyset$, and $l_1 = \ds{a}$.\\
Property (2) and $t_2 < t_1$ ensure that
$\storage'(t_1) = v$ and $\storage'(t_a) = a$.
To show that  $\bn(\strfor(\newstate', t_1)) \subseteq \commits'$ we use
Lemma~\ref{lemma:pas:mono}.
Finally, $t_1 \not\in \commits'$, since $t_1 > t_2$.
To prove that $\alpha_2 \neq \cmt(a,  v')$ we proceed by contradiction. If $\alpha_2 = \cmt(a, v')$ then
  $t_2 \not \in \commits_0$. However, $t_2
  \in \bn(\strfor(\newstate_0, t_1))$, because they write the same
  address $a$ and $t_2 < t_1$. This contradict that $\bn(\strfor(\newstate_0, t_1)) \subseteq \commits_0$. 
\\
Therefore we can apply rule $\cmt$ to show that
  $\newstate' \doublestep{l_1}{} (\Inst', \storage', \commits' \cup \{t_1\},
  \decodes') = \newstate_2$. 
  \\
  \textbf{(Case $\ftc$)}
  In this case $o = \store{\Pc}{}{t_v}$.
  The hypothesis of the rule ensure that
  $\storage_0(t_1) = v$, $t_1 \not \in \decodes_0$, and
  $\bn(\strfor(\newstate_0, t_1)) \subseteq \decodes_0$.
  The conclusion of the rule ensures that $\hat \storage_1 = \emptyset$,
  $\hat \commits_1= \emptyset$, $\hat \decodes_1= \{t_1\}$, 
  $I_{t_1} = \translate(v, max(\Inst_1)$, and $l_1 = \il{a}$.
  \\
  Property (2) and $t_2 < t_1$ ensure that
$\storage'(t_1) = v$.
To show that  $\bn(\strfor(\newstate', t_1)) \subseteq \decodes'$ we use
Lemma~\ref{lemma:pas:mono}.
\\
Finally, $t_1 \not\in \decodes'$, since $t_1 > t_2$.
\\
To complete the proof we must show that $\hat \Inst_2 = \emptyset$.
We proceed by contradiction: if $\hat \Inst_2 \neq \emptyset$ then
$\alpha_2 = \ftc$, hence this transition
fetched $t_2$ and $t_2 \not \in \decodes_0$. However, $t_2 < t_1$ and both update the program counter,
therefore $t_2 \in \bn(\strfor(\newstate_0, t_1))$. This contradicts
the hypothesis that $\bn(\strfor(\newstate_0, t_1)) \subseteq \decodes_0$.\\
Finally,  we can apply rule $\ftc$ to show that
$\newstate' \step{l_1}{}$\\$ (\Inst' \cup \translate(v, max(\Inst'), \storage', \commits',
\decodes' \cup \{t_1\}) = $\\$(\Inst' \cup \translate(v, max(\Inst_0), \storage', \commits',
\decodes' \cup \{t_1\}) = \newstate'$.

\end{proof}


%% file: proof-spec-consistency.tex
\subsection{Memory Consistency of Speculative Semantics: Theorem~\ref{thm:spec:co}}\label{prf:speco}
\label{appendix:proof:spec:cho}
We reduce memory consistency for the speculation model to the OoO case using Theorem~\ref{thm:ooo:co}. Since the OoO semantics
already takes care of reordering, to prove Theorem \ref{thm:spec:co} a bisimulation argument suffices. Intuitively, referring to Figure \ref{fig:spec:lifecycle},  
the states ``decoded'', ``predicted'', ``speculated'' and ``speculatively fetched'' in the speculative semantics all correspond 
in some sense to the state ``decoded'' in the OoO semantics, in that any progress can still be undone to return to the ``decoded'' state.
In a similar vein, the state ``retired'' corresponds to ``executed''
in the OoO semantics, ``fetched'' to ``fetched'' and ``committed'' to
``committed''. The only exception is states that are speculatively
fetched. In this case there is an option to directly retire the
fetched state, without passing through ``retired'' first. The proof
reflects this intuition. 
 
 The main challenges in defining the bisimulation are i) to pin down
 the non-speculated instructions in the speculative semantics and
 relate them correctly to instructions in the OoO semantics, and ii)
 account for speculatively fetched instructions. The latter issue
 arises when retiring an instruction in the speculative semantics that
 has earlier been speculatively fetched. In that case, the
 corresponding decoded microinstructions
 are already in flight, although the bisimilar OoO state will have no
 trace of this. 
 Then the OoO microinstruction will have to be first executed and then
 fetched. The following definitions make this intuition precise.

First say that a name $t'$ is \emph{produced by} the PC store
microinstruction $\ass{c}{t}{\store{\Pc}{}{t_v}} \in \Inst$, if $t'\prec
t$, i.e. $t\in\dom{\spec(t')}$. We would like to conclude that $t$ is
uniquely determined, as we need this to properly relate the
speculative and OoO states.  However, this 
does not hold in general. For a counterexample consider the PC store
microinstruction $\ass{c}{t}{\store{\Pc}{}{t_v}}$. Suppose that the fetch
from $t$ causes a new instruction $\Inst'$ to be allocated with
another PC store instruction $\ass{c'}{t'}{\store{\Pc}{}{t'_v}}$
followed by a PC load $\ass{c''}{t''}{\load{\Pc}{}}$, i.e. such that $t'< t''$. 
At this point, $\spec(t') = [t\mapsto v]$ and $\spec(t'') = [t\mapsto
v]$. After executing the fetched PC store $t'$ and then the PC load,
$\storage(t')=v'$ and  
$\storage(t'') = v'$. At this point, $\spec(t'')$ will map $t$ to $v$
and $t'$ to $v'$. But then $t''$ is produced by both $t$ and $t'$. The
same property holds 
if $t'$ is used as an argument to operations other than a PC load.
This causes us to impose the following wellformedness condition on
instruction translations:
\begin{definition}[Wellformed instruction translation]
\label{def:wft}
The translation function $\translate$ is wellformed if $\translate(v,t)=\Inst$ implies:
\begin{enumerate}
\item $\ass{c}{t}{\store{\Pc}{}{t_v}}, \ass{c'}{t'}{\load{\Pc}{}} \in\Inst$ implies $t'< t$.
\item $\ass{c}{t}{\store{\Pc}{}{t_v}}, \inst\ \in\Inst$ implies $t\notin \fn(\inst)$.
\item For all $\storage$ there is a unique $\ass{c}{t}{\store{\Pc}{}{t_v}}$ such that $\den{c}\storage$.
\end{enumerate} 
\end{definition}
Condition \ref{def:wft}.1 and 2 can be imposed without loss of
generality since any occurrence of $t$ bound to the microinstruction
$\ass{c}{t}{\store{\Pc}{}{v}}$ can be replaced by $v$ itself, and condition
\ref{def:wft}.3 is natural to ensure that any linear control flow
gives rise to a correspondingly linear flow of instructions. We
obtain: 
\begin{proposition}
If $t'$ is produced by $t_1$ and $t'$ is produced by $t_2$ then $t_1=t_2$. \hfill $\Box$
\end{proposition}

Consider now microinstructions (images of $\translate$)
$\Inst_1$ and $\Inst_2$ such that
$bn{(\Inst_1)}\cap bn{(\Inst_2)}=\emptyset$. Say that $\Inst_1$
\emph{produces} $\Inst_2$, 
$\Inst_1<\Inst_2$, if there is $\ass{c}{t}{\store{\Pc}{}{t_v}} \in \Inst_1$ such that for each $t'\in\Inst_2$, $t'$ is produced by $t$.  Clearly, if $\Inst'$ is added to the set of microinstructions due to a fetch from $\Inst$ then $\Inst<\Inst'$. Say then that $\Inst$ (and by extension states containing $\Inst$) is \emph{wellformed}
by the partitioning $\Inst_{retired},\Inst_1,\ldots,\Inst_n$, if
$\Inst=\bigcup\{\Inst_{retired},\Inst_1,\ldots,\Inst_n\}$,
$\Inst_{retired}$ is retired, and for each $\Inst_i$, $1\leq i\leq n$ there is
$\Inst\in\{\Inst_{retired},\Inst_1,\ldots,\Inst_n\}$ such that $\Inst
< \Inst_i$. Moreover we require that $<^*$ on the partitions
$\{\Inst_{retired},\Inst_1,\ldots,\Inst_n\}$ is well-founded and that
the partitions are maximal. Note
that if $\Inst$ is wellformed by
$\Inst_{retired},\Inst_1,\ldots,\Inst_n$ 
then the partitioning is unique. We note also that all reachable states  in the
speculative semantics are wellformed and each partition corresponds
to the translation of one single ISA instruction.
We say that an ISA instruction $\Inst_i$ in the
partitioning $\Inst_1,\ldots,\Inst_n$ is \emph{unconditionally
  fetched}, if $\Inst_{retired}< \Inst_i$ and let $\Inst_{uf}$ be the
union of $\Inst_{retired}$ and the 
instructions that have been unconditionally fetched.

We can now proceed to define the bisimulation $R$. We restrict
attention to reachable states in both the OoO and speculative
semantics in order to keep the definition of $R$ manageable and be
able to implicitly make use of simple invariant properties such as
 $\dom{\spec}\cap\dom{\spec(t)}\Rightarrow 
t\notin\commits$ (no instruction with a speculated dependency is committed).
Let $(\Inst_1,  \storage_1, \commits_1, \decodes_1)\ R\  (\Inst_2,
\storage_2,  \commits_2, \decodes_2, \spec_2, \guesses_2) $ if 
\begin{enumerate}
\item $\Inst_2$ is wellformed by the partitioning $\Inst_{2,retired},\Inst_{2,1},\ldots,\Inst_{2,n}$.
\item There is a bijection $\widetilde{\cdot}$ from $\Inst_{2,uf}$ to $\Inst_1$.
\item $\commits_2 = \widetilde{\commits_1}$,
\item $\mapminus{\storage_2}{\dom{\spec_2}} = \widetilde{\storage_1}$,
\item $\decodes_2 \setminus\dom{\spec_2} = \widetilde{\decodes_1}$,
 \end{enumerate}
In 3.-5. the bijection $\widetilde{\cdot}$ is pointwise extended to sets and expressions.

Note that, from 2. and 4. we get that a microinstruction $\widetilde{t}$ in $\Inst_1$ has been executed iff $t\notin\dom{\spec_2}$.
 

We prove that $R$ is a weak bisimulation in two steps. We first show
that all speculative transitions up until retire or non-speculative
fetch are reversible. To prove this it is sufficient to show that each
of the conditions 1.--4. is invariant under $\prd$, $\exe$, $\pexe$, $\rbk$, and $\ftc$, the latter under the condition that
the fetched instruction is in $\spec_2$. These transitions are
identified by $T1$ in the following picture:
\begin{center}
\begin{tikzpicture}[node distance=2.5cm,auto,>=latex']
\tikzstyle{every node}=[font=\small]
\tikzset{pblock/.style = {rectangle, draw=white!50, top
                      color=white,bottom color=white, align=center, minimum width=0.75cm}}
                      
\node[pblock,name=a1]{$(\newstate_2, \spec_2, \guesses_2)$};                      
\node[pblock,name=a2, right of  = a1]{$(\newstate'_2, \spec'_2, \guesses'_2)$};
\node[pblock,name=a3, right of  = a2]{$(\newstate''_2, \spec''_2, \guesses''_2)$};
\node[pblock,name=a4, right of  = a3]{$(\newstate'''_2, \spec'''_2, \guesses'''_2)$};
                
\node[pblock,name=a'1, below = 1cm of  a1]{$\newstate_1$};                      
\node[pblock,name=a'2, right of  = a'1]{$\newstate'_1$};
\node[pblock,name=a'3, right of  = a'2]{$\newstate''_1$};
\node[pblock,name=a'4, right of  = a'3]{$\newstate'''_1$};

\draw[->]  (a1) edge  node[above] {T1} (a2);
\draw[->]  (a2) edge  node[above]{T2} (a3); 
\draw[->]  (a3) edge  node[above]{T3} (a4); 

\draw[->]  (a'1) edge  node[above] {T2} (a'2);
\draw[->]  (a'2) edge  node[above] {\exe} (a'3);
\draw[->]  (a'3) edge  node[above] {\ftc} (a'4);

\draw[-]  (a1) edge  node[left] {$R$} (a'1);
\draw[-]  (a2) edge  node[left] {$R$} (a'1);
\draw[-]  (a3) edge  node[left] {$R$} (a'2);
\draw[-]  (a4) edge  node[left] {$R$} (a'4);
\end{tikzpicture}
\end{center}


\begin{lemma}
\label{lem:bisim:one}
If $\newstate_1\ R\ (\newstate_2,\spec_2,\guesses_2)$ and 
$(\newstate_2,\spec_2,\guesses_2)\triplestep{}{} (\newstate_2',\spec_2',\guesses_2')$ is an instance of $\prd$, $\exe$, $\pexe$,$\rbk$, or speculative $\ftc$ then $\newstate_1\ R\ (\newstate_2',\spec_2',\guesses_2')$.
\end{lemma}
\begin{proof}
Let $\newstate_2 = (\Inst_2,\storage_2,\commits_2,\decodes_2)$ and $(\newstate_2',\spec_2',\guesses_2')=(\Inst_2',\storage_2',\commits_2',\decodes_2',\spec_2',\guesses_2')$.

\noindent \textbf{(Case \prd)} We get $\newstate_2' = (\Inst_2,\storage_2[t\mapsto v],\commits_2,\decodes_2,\spec_2,\guesses_2\cup\{t\})$ and note that conditions 1.--5. are trivially satisfied by the assumptions.
\\
\textbf{(Case \exe)} We get $\newstate_2\doublestep{}{} \newstate_2'$
with $\stepparam(\newstate_2,\newstate_2')=(\exe, t)$, $t\not\in \guesses_2$,
$\spec_2' = \spec_2\cup\{t\mapsto \proj{\storage_2}{\deps(t,\newstate_2)}\}$ and $\guesses_2'=\guesses_2$. Cond. 1 and 2 are straightforward since $\Inst_{2,uf}$ and
$\Inst_1$ are not affected by the transition.
For cond. 3 and 5 we get $\commits_2' = \commits_2$ and
\[
\decodes_2'\cap \overline{\dom{\spec_2'}} = 
\decodes_2\cap \overline{\dom{\spec_2}}\ ,
\]
since $\dom{\spec_2'}\setminus \dom{\spec_2}=\{t\}$ and $t\not\in\dom{\decodes_2}$. For cond. 4,
\begin{eqnarray*}
\mapminus{\storage_2'}{\dom{\spec_2'}}
& = & \mapminus{\storage_2[t\mapsto v]}{\dom{\spec_2\cup\{t\mapsto \proj{\storage_2}{\deps(t,\newstate_2)}}} \\
& = & \mapminus{\storage_2[t\mapsto v]}{\dom{\spec_2}\cup\{t\}} \\
& = & \mapminus{\storage_2}{\dom{\spec_2}}\ .
\end{eqnarray*} 
%

\noindent {\textbf{(Case \pexe)}} In this case
$\proj{\newstate_2}{\overline\{t\}} \doublestep{}{} \newstate_2'$ with
$\stepparam(\newstate_2,\newstate_2')=(\exe, t)$, $t\in \guesses_2$,
$\spec_2' = \spec_2\cup\{t\mapsto \proj{\storage_2}{\deps(t,\newstate_2)}\}$ and $\guesses_2'=\guesses_2\setminus\{t\}$. Cond. 1 and 2 again are immediate. For cond. 3, $\commits_2'=\commits_2$ and for cond. 5, 
\[
\decodes_2'\cap \overline{\dom{\spec_2'}} = \decodes_2\cap \overline{\dom{\spec_2}}
\]
since, again, $\dom{\spec_2'}\setminus \dom{\spec_2}=\{t\}$ and $t\not\in\dom{\decodes_2}$.
Finally for cond. 2,
\begin{eqnarray*}
\mapminus{\storage_2'}{\dom{\spec_2'}}
& = & \mapminus{\storage_2[t\mapsto v]}{\dom{\spec_2}\cup\{t\}} \\
& = & \mapminus{\storage_2[t\mapsto v]}{\dom{\spec_2}\cup\{t\}}  \\
& = & \mapminus{\storage_2}{\dom{\spec_2}} 
\end{eqnarray*}

\noindent \textbf{(Case \rbk)} We get that  $(\Inst_2,\storage_2,\commits_2,\decodes_2)\not\sim_{t}
 (\Inst_2,\spec_2(t),\commits_2,\decodes_2)$ and $t\not\in\guesses_2$. 
 We get $\Inst_2'= \Inst_2\setminus \Delta^+$, $\storage_2' = \proj{\storage_2}{\overline{\Delta^*}}$,
 $\commits_2'=\commits_2$, $\decodes_2'=\decodes_2\setminus\Delta^*$, $\spec_2'=\spec_2\setminus\Delta^*$, and $\guesses_2'=\guesses_2\setminus\Delta^*$. 
 For cond. 1 and 2 first note that $t\in\dom{\spec_2}$. If $t$ is not a PC store the effect of $\rbk$ is to remove $t$ from $\storage_2$, $\decodes_2$, $\spec_2$, $\guesses_2$. This does not affect the bijection $\widetilde{\cdot}$, so 1 and 2 remain valid also for $(\newstate_2',\spec_2',\guesses_2')$. If $t$ \emph{is} a PC store
 then we need to observe the following: Since $t$ is speculated, $t$ is a member of some "macro"-instruction (= partition) $\Inst_{2,i}$. The set
 $\Delta^+$ contains all instructions/partitions $\Inst_{2,j}$ such that $\Inst_{2,i}<^+ \Inst_{2,j}$. In particular, no such $\Inst_{2,j}$ is in $\Inst_{2,fu}$, since
 otherwise $\Inst_{2,j}$ would have been added by a retired PC store microinstruction. It follows that the bijection $\widetilde{\cdot}$ is not affected by
 the removal of $\Delta^+$, and 1 and 2 are reestablished for the new speculative state. 
 
 For cond. 3,  $\commits_2' = \commits_2$. For cond. 4, we calculate:
\begin{eqnarray*}
\proj{\storage_2'}{\overline{\dom{\spec_2'}}} 
 & = & \proj{(\proj{\storage_2}{\overline{\Delta^*}})}{\overline{\dom{\spec_2'}}} \\
  & = & \proj{\storage_2}{\overline{\dom{\spec_2'}\cup\guesses_2'}\cup\overline{\Delta^*}} \\
   & = & \proj{\storage_2}{\overline{\dom{\spec_2\setminus\Delta^*}}\cup\overline{\Delta^*}} \\
 & = & \proj{\storage_2}{\overline{\dom{\spec_2}}\cup\overline{\Delta^*}} \\
 & = & \proj{\storage_2}{\overline{\dom{\spec_2}}}\ .
\end{eqnarray*}
Note that the final step uses prop. \ref{prop:dep:po}.
\begin{proposition}
\label{prop:dep:po}
The relation $\prec^*$ is a partial order.
\end{proposition}
\begin{proof}
By induction in the length of derivation.
\end{proof}

 Finally for cond. 5:
 \begin{eqnarray*}
\decodes_2'\setminus\dom{\spec_2'}
 & = & (\decodes_2\setminus\Delta^*)\setminus(\dom{\spec_2}\setminus\Delta^*) \\
 & = & \decodes_2\setminus(\dom{\spec_2}) \ .
  \end{eqnarray*}

\noindent \textbf{(Case speculative \ftc)} 
We get that $\newstate_2\doublestep{}{} \newstate_2'$, $\stepparam(\newstate_2,\newstate_2')=(\ftc(\Inst_2''),t)$,
$\guesses_2'=\guesses_2$, $\spec_2'=\spec_2\cup\{t'\mapsto \proj{\storage}{\{t\}}\mid t'\in\Inst_2''\}$, and, since $t$ is speculated,
$t\in\dom{\spec_2}$. Also, we find $\ass{c}{t}{\store{\Pc}{}{t_v}} \in
\Inst_2$, $\storage_2(t) = v$, $t\notin\decodes_2$,
$ \bn(\strfor(\newstate_2, t)) \subseteq \decodes_2$, $\Inst_2'=\Inst_2\cup\Inst_2''$, $\storage_2'=\storage$,
$\commits_2'=\commits_2$, and $\decodes_2' = \decodes\cup\{t\}$. For cond. 1 and 2 observe that no instruction  $\Inst'$ added by
the fetch can belong to $\Inst_{2,fu}'$, since all such instructions are produced by a retired PC store instruction.
\\
For cond. 3, $\commits_2'=\commits_2$ is immediate.
\\
For cond. 4 we calculate:
\begin{eqnarray*}
\proj{\storage_2'}{\overline{\dom{\spec_2'}}} 
 & = & \proj{\storage_2}{\overline{\dom{\spec_2}\cup\{t'\mid t'\in\Inst_2''\}}} \\
 & = & \proj{\storage_2}{\overline{\dom{\spec_2}}}
 \end{eqnarray*}
 \\
 Finally for cond. 5:
 \begin{eqnarray*}
 \decodes_2'\setminus\dom{\spec_2'} 
  & = & (\decodes_2\cup\{t\})\setminus(\dom{\spec_2}\cup\{t'\mid t'\in\Inst_2''\}) \\
  & = &  \decodes_2 \setminus\dom{\spec_2}
 \end{eqnarray*}
%
%
\end{proof}

The following Lemma handle cases for $\ret$ when $t \not \in
\decodes_2$, $\cmt$, $\ftc$ when $t \not \in \dom{\delta_2}$ (which
are identified by $T2$ in Figure), and the cases for 
$\ret$ when $t \in \decodes_2$ (which
are identified by $T3$ in Figure).

\begin{lemma}
\label{lem:bisim:two}
Assume that $\newstate_1\ R\ (\newstate_2,\spec_2,\guesses_2)$.
\begin{enumerate}
\item
If $(\newstate_2,\spec_2,\guesses_2)\triplestep{}{}(\newstate_2',\spec_2',\guesses_2')$ is an instance of $\cmt$, $\ftc$, or $\ret$ then 
$\newstate_1 \doublestep{}{}^{\ast}\newstate_1'$ such that $\newstate_1'\ R\ (\newstate_2',\spec_2',\guesses_2')$.
\item
If $\newstate_1 \doublestep{}{}\newstate_1'$ then $(\newstate_2,\spec_2,\guesses_2)\triplestep{}{}(\newstate_2',\spec_2',\guesses_2')$
such that $\newstate_1'\ R\ (\newstate_2',\spec_2',\guesses_2')$.
\end{enumerate}
\end{lemma}
\begin{proof}
  Assume first that 
  $(\newstate_2,\spec_2,\guesses_2)\triplestep{}{}(\newstate_2',\spec_2',\guesses_2')$.
Let $\newstate_i = (\Inst_i,\storage_i,\commits_i,\decodes_i)$ and $\newstate_i'=(\Inst_i',\storage_i',\commits_i',\decodes_i')$. We prove that
$\newstate_1'\ R\ (\newstate_2',\spec_2',\guesses_2')$ and proceed by cases, first from the speculative to the OoO semantics. 

\noindent \textbf{(Case $\cmt$)} We get
$\newstate_2\doublestep{}{}\newstate_2'$, and $t\notin\dom{\spec_2}$,
$\ass{c}{t}{\store{\Locs}{t_a}{t_v}} \in \Inst_2$, $\storage_2(t) = v$, $t\notin\commits_2$ and  $\bn(\strfor(\newstate_2, t)) \subseteq \commits_2$. Since $t$ is not speculated we get $t\in\Inst_{2,retired}$. Since $t\notin\dom{\spec_2}$ we get $\storage_1(\widetilde{t}) = \widetilde{v}$ by cond. 4, $\widetilde{t}\notin\widetilde{\commits_1}$ by cond. 3,  $\widetilde{t}\in\widetilde\Inst_2$ by cond. 2, and,
since all members of $\bn(\strfor(\newstate_2, t))$ are non-speculated, by cond. 2, 3, 4, $\bn(\strfor(\newstate_1, \widetilde{t})) \subseteq \commits_1$.
It follows that $\newstate_1' = (\Inst_1,\storage_1,\commits_1\cup\{\widetilde{t},\decodes_1)$, and the conditions 1.-5. for $\newstate_1'$ and $\newstate_2'$ are easily verified.

\noindent \textbf{(Case non-speculative $\ftc$)} For a non-speculative fetch we get $t\notin\dom{\spec_2}$. Also, $\newstate_2\doublestep{}{} \newstate_2'$,
$\stepparam(\newstate_2,\newstate_2') = (\ftc(\Inst_2'),t)$,
$\ass{c}{t}{\store{\Pc}{}{t_v}} \in\Inst_2$, $\storage_2(t) = v$, $t\notin\decodes$,
$\bn(\strfor(\newstate_2, t)) \subseteq \commits_2$, $\Inst_2' = \Inst_2\cup\Inst_2''$, $\storage_2'=\storage_2$, $\commits_2'=\commits_2$, and $\decodes_2'=\decodes_2\cup\{t\}$. Since $\ass{c}{t}{\store{\Pc}{}{t_v}}\in\Inst_2$ and $t\notin\dom{\spec_2}$ we obtain
that $\ass{c}{t}{\store{\Pc}{}{t_v}} \in\Inst_{2,retired}$ and hence that $\widetilde{t}\leftarrow\widetilde{c}?st\ \Pc\ \widetilde{t_v}\in\Inst_1$. Also, $\storage_1(\widetilde{t})=\widetilde{v}$ by cond. 4, 
$t\notin \decodes_1$ by cond. 5, and $\bn(\strfor(\newstate_1, \widetilde{t})) \subseteq \decodes_1$. But then
$\newstate_1\doublestep{}{} \newstate_1'$ where $\Inst_1' = \Inst_1\cup \widetilde{\Inst_2'}$ (by wellformedness), $\storage_1' =\storage_1$,
$\commits_1'=\commits_1$, and $\decodes_1' = \decodes_1\cup\{\widetilde{t}\}$.

\noindent \textbf{(Case $\ret$ and $t \not \in \decodes_2$)} Assume that $\newstate_2\doublestep{}{}\newstate_2'$, $t\in\dom{\storage_2}$, $\dom{\spec_2(t)}\cap \dom{\spec_2} =\emptyset$, $t\notin\guesses_2$, and $(\Inst_2,\storage_2,\commits_2,\decodes_2) \sim_t (\Inst_2,\spec_2(t),\commits_2,\decodes_2)$ such that $\Inst_2'=\Inst_2$, 
$\storage_2' =\storage_2$, $\commits_2'=\commits_2$, $\decodes_2'=\decodes_2$, $\spec_2'=\proj{\spec_2}{\overline{\{t\}}}$, and
$\guesses_2'=\guesses_2$.
By Lemma~\ref{lem:deps-eq} $\den{\inst}{(\Inst_2,\storage_2,\commits_2,\decodes_2)} = \storage_2(t)$.
By cond. 4, whenever $t'\notin\dom{\spec_2}$, $\storage_1(\widetilde{t'}) = \storage_2(t')$. We know that $\storage_1(\widetilde{t_i}) = \widetilde{\storage_2(t_i)}$ by 4. Let $\ass{c}{t}{o} \in \Inst_2$ be the
microinstruction bound to $t$ in $\Inst_2$. Since $\dom{\spec_2}\cap\dom{\spec_2(t)}=\emptyset$ we know that $t\in\Inst_{2,uf}$ and hence $\widetilde{t}\in\dom{\Inst_1}$ by 2. It follows that $\widetilde{t}$ can be executed resulting in $\Inst_1'=\Inst_1$, $\storage_1' = \storage_1[\widetilde{t}\mapsto \storage_2'(t)]$
such that $\storage_1'(\widetilde{t}) = \widetilde{\storage_2'(t)}$, $\commits_1'=\commits_1$, and $\decodes_1'=\decodes_1$. This is sufficient to
reestablish $R$.

\noindent \textbf{(Case $\ret$ and $t \in \decodes_2$)}
In this case we must additionally
account for the microinstructions produced by $t$: i.e. the partition
$\hat \Inst_{2}$ of $\Inst_2$ such that $\{\ass{c}{t}{o}\} < \hat
\Inst_{2}$. The microsintructions are not in $\Inst_{2,uf}$, since $t
\in \dom{\spec_2}$, hence are not covered by the bijection $\widetilde \cdot$.
However, microinstructions in $\hat \Inst_{2}$ are in $\Inst'_{2,uf}$.
For this reason, in order to restablish the bisimulation, the OoO must
perform a further step from $\newstate'_1$ and apply $\ftc$ to $t$. 
This allows to extend the bijection  
$\widetilde \cdot$ to relate the newly decoded
microinstructions to $\hat \Inst_{2}$.

\noindent \textbf{(Converse direction)}
For the converse direction, from the OoO semantics to the speculative semantics the steps for commits and non-speculative retires follow the commit case above closely. The only delicate case is for $\exe$. So assume $\newstate_1\step{}{}_{OoO}\newstate_1'$ such that $\widetilde{\instr} = \widetilde{t}\leftarrow \widetilde{c}?\widetilde{o}\in\Inst_1$, $\widetilde{t}\notin\dom{\storage_1}$, $[\widetilde{c}]\storage_1$ is true, $[\widetilde{\instr}]\newstate_1 = (\widetilde{v},\widetilde{l})$, $\Inst_1'=\Inst_1$, $\storage_1' = \storage_1[\widetilde{t}\mapsto\widetilde{v}]$, $\commits_1'=\commits_1$, and $\decodes_1'=\decodes_1$. We get that $\instr = t\leftarrow c?o\in\Inst_{2,uf}$. There are two cases. Either $t\notin\dom{\storage_2}$ ($t$ has not yet been executed), or $t\in\dom{\storage_2}\cap\dom{\spec_2}$.
In the former case, the execution step can be mirrored in the speculative semantics and then retired. In the latter case, the execution step can be retired directly.
This completes the proof of lemma \ref{lem:bisim:two}.
\end{proof}

We now obtain theorem \ref{thm:spec:co} as a corollary of lemma \ref{lem:bisim:one} and \ref{lem:bisim:two}.
\hfill $\Box$


%% file: proof_constant_time.tex
\newcommand{\ctbisimrel}{\mathbf{R}}
\subsection{MIL Constant Time Security:
  Theorem~\ref{thm:ct}}
The proof is done by showing that the relation $\ctbisimrel$ is a bisimulation for the
OoO transition relation, where $\newstate \ctbisimrel \newstate'$
iff $\newstate \bisim \newstate'$ and there exist $\newstate_0 \sim_L
  \newstate'_0$ and $n$ such that $\newstate_0 \doublestep{}{}^n \newstate$ and $\newstate'_0
  \doublestep{}{}^n \newstate'$.

Let $(\Inst, \storage, \commits. \decodes) = \newstate$,
$(\Inst', \storage', \commits'. \decodes') = \newstate'$,
$\newstate \doublestep{}{} \newstate_1 = (I \cup I_t, \storage \cup \storage_t, \commits
\cup \commits_t, \decodes \cup \decodes_t)$, and $\stepparam(\newstate,
\newstate_1) = (\alpha, t)$.
The proof is done by case analysis on $\alpha$.

\textbf{(Case $\exe$)}
The hypothesis of the rule ensure that 
$\inst = \ass{c}{t}{o} \in \Inst$,
  $\fundefined{\storage}{t}$,  $\den{c}{\storage}$, and
  $\den{\inst}\newstate = (v, l)$.
  The conclusion  ensures that $\hat \storage = \{t \mapsto v\}$,
  $\hat \commits=\emptyset$, $\hat \decodes=\emptyset$, and
  $\hat \Inst = \emptyset$.
  Relation $\bisim$ ensures that  $\inst \in \Inst'$,
  $\fundefined{\storage'}{t}$,  and $\den{c}{\storage'}$.

  We must show that exists $v'$ such that $\den{\inst}\newstate' = (v', l)$.
  For $o = e$ and $o = \store{\restype}{t_a}{t_v}$ the proof is trivial. In fact,
  since $\den{\inst}\newstate$ then all free names of $o$ are defined in
  $\newstate$ and $\bisim$ ensures that the same names are defined in
  $\newstate'$.
  For $o = \load{\restype}{t_a}$, $\den{\inst}\newstate = (v, l)$ ensures
  that $\bn(\stract(\newstate, t)) = \{t_s\}$, $\fdefined{\newstate}{t_a}$, and
  $\fdefined{\newstate}{t_s}$. Relation $\bisim$ ensures that
  $\newstate'(t_a)=\newstate(t_a)$, $\fdefined{\newstate'}{t_s}$,
  $\Inst=\Inst'$ (hence there are the same store
  instructions), and that for every store
  $\ass{c'}{t'}{\store{\restype}{t'_a}{t'_v}}$,
  $\den{c'}\newstate=\den{c'}{\newstate'}$, and
  $\newstate(t'_a)=\newstate'(t'_a)$. Therefore
  $\bn(\stract(\newstate', t)) = \bn(\stract(\newstate, t))$
  and $\den{\inst}\newstate' = (\newstate'(t_s), l')$. Finally, since
  relation $\bisim$ guarantee that $(t_s \in \commits) \Leftrightarrow
  (t_s \in \commits')$ then $l'=l$.
  \\
  These properties enable applying rule ($\exe$) to show that
  $\newstate' \doublestep{l}{} (\Inst', \storage' \cup \{t \mapsto \storage'(t_s)\}, \commits',
  \decodes') = \newstate'_1$.
  \\
  To prove that $\bisim$ is preserved we use 
  Theorem~\ref{thm:ooo:co}.
  Let $\ass{c'}{t'}{o'} \in \Inst$. Let $\ANames'$ be $fn(c')$ if 
  $o'$ is neither a load or a store; $fn(c') \cup
  \{t'_a\}$ if $o'$ is a memory or register access and $t'_a$ is the
  corresponding address;
  $fn(c') \cup
  \{t'_v\}$ if $o'$ is a program counter update and $t'_v$ is the
  corresponding value.
  If $t \not \in \ANames'$ then the proof is trivial, since
  $\den{c}{\storage \cup \{t  \mapsto v\}} = \den{c}{\storage} =
  \den{c}{\storage'} = \den{c}{\storage' \cup \{t  \mapsto
    \storage'(t_s)\}}$ (the same holds for the address in case of a
  resource accesses or program counter update).

\begin{center}
\begin{tikzpicture}[node distance=2cm,auto,>=latex']

\tikzstyle{every node}=[font=\small]
\tikzset{pblock/.style = {rectangle, draw=white!50, top
                      color=white,bottom color=white, align=center, minimum width=0.75cm}}
                      
\node[pblock,name=a1]{$\newstate_0$};                      
\node[pblock,name=a2, right of  = a1]{$\newstate$};
\node[pblock,name=a3, right of  = a2]{$\newstate_1$};
\node[pblock,name=a4, right of  = a3]{$\newstate_{S}$};
                
\node[pblock,name=a'1, below = 0.5cm of a1]{$\newstate'_0$};                      
\node[pblock,name=a'2, right of  = a'1]{$\newstate'$};
\node[pblock,name=a'3, right of  = a'2]{$\newstate'_1$};
\node[pblock,name=a'4, right of  = a'3]{$\newstate'_{S}$};

\draw[->]  (a1) edge[dashed]  node[above] {$n$} (a2);
\draw[->]  (a2) edge  node[above]{$(\alpha, t)$} (a3); 
\draw[->]  (a3) edge[dashed] node[above,sloped] {$m$}  (a4);
\draw[->]  (a1) edge[dashed,bend left]  node[above] {$n+1+m$} (a4);

\draw[->]  (a'1) edge[dashed]  node[above] {$n$} (a'2);
\draw[->]  (a'2) edge  node[midway,sloped]{} (a'3); 
\draw[->]  (a'3) edge[dashed] node[above,sloped] {$m$}  (a'4);
\draw[->]  (a'1) edge[dashed,bend right]  node[below] {$n+1+m$} (a'4);

\draw[-]  (a1) edge  node[left] {$\sim_L$} (a'1);
\draw[-]  (a2) edge  node[left] {$\ctbisimrel$} (a'2);
\draw[-]  (a4) edge  node[left] {$\bisim$} (a'4);
\draw[-]  (a3) edge  node[left] {$\ctbisimrel$} (a'3);
\end{tikzpicture}
\end{center}
For $t \in fn(c')$ we reason as follows.
States $\newstate_1$ and
  $\newstate'_1$ are the $(n+1)$-th states of two
  OoO traces $\rho = \newstate_0 \doublestep{}{}^{n+1} \newstate_1$ and $\rho' = \newstate'_0 \doublestep{}{}^{n+1} \newstate'_1$ such that $\newstate_0
  \sim_{L} \newstate'_0$.
  There is a trace $\rho_1 = \newstate_0 \doublestep{}{}^{n+1}
  \newstate_1 \doublestep{}{}^{m} \newstate_s$ that has
  prefix $\rho$, such that $\completed{\newstate_s, t''}$
  for every $t'' \leq max(bn(\Inst'))$. Notice that this state is ``sequential''.
  Since in the OoO semantics the storage is monotonic  then
  $\newstate_1(t) = \newstate_s(t)$.
  Theorem~\ref{thm:ooo:co} permits to connect this trace to a
  sequential trace, which enables to use the MIL constant-time hypothesis.
  In fact, there exists an ordered
  execution $\pi$ of $\rho_1$ that ends in $\newstate_s$: $\pi = \newstate_0 \singlestep{}{}^{n+1+m}
  \newstate_s$.
  For the same reason, $\rho'$ is a prefix
  of a trace $\rho'_1$ that ends in a sequential state $\newstate'_s$,
  $\newstate'_s(t) = \newstate'_1(t)$, and there exists a sequential trace $\pi'$
  of $n+1+m$ steps  that
  ends in $\newstate'_s$.
  Since $\completed{\newstate_s, t''}$ and $\completed{\newstate'_s, t''}$
  then $\fdefined{\den{c'}}{\newstate_s}$ and
  $\fdefined{\den{c'}}{\newstate'_s}$.
  Therefore, we can use the assumption on MIL constant-time to show
  that
  $\newstate_s \bisim \newstate'_s$, hence $\den{c'}{\newstate_s} = \den{c'}{\newstate'_s}$. Since $t \in fn(c)$,
  either $\fundefined{\den{c'}}{\newstate_s}$ and
  $\fundefined{\den{c'}}{\newstate'_s}$ or $\den{c'}{\newstate_1} =
  \den{c'}{\newstate'_1}$.
  The same reasoning is used if $t'$ is a resource accesses and $t
  = t'_a$, or if $t'$ is program counter update and $t
  = t'_v$.

\textbf{(Case \cmt)}
The hypothesis  ensure that
$\ass{c}{t}{\store{\Locs}{t_a}{t_v}} \in \Inst$,
  $\storage(t) = v$, $t \not \in \commits$,
  $\bn(\strfor(\newstate, t)) \subseteq \commits$,
  and $\storage(t_a) = a$.
  The conclusion  ensures that $\storage_{t} = \emptyset$,
  $\commits_{t}= \{t\}$, $\decodes_{t}=\emptyset$, 
  $I_{t} = \emptyset$, and $l = \ds{a}$.
  Also, the invariant guarantees that $\den{c}{\storage}$.\\
  The relation $\bisim$ ensures that   $\ass{c}{t}{\store{\Locs}{t_a}{t_v}} \in \Inst'$,
  $\exists v' . \storage'(t) = v'$, $t \not \in \commits'$,
  $\storage'(t_a) = a$, and $\den\storage'[c]$.\\
  To show that $\bn(\strfor(\newstate', t)) = \bn(\strfor(\newstate, t))$ we use
  the same reasoning used to prove that $\bn(\stract(\newstate', t)) =
  \{t_s\}$  of case $\exe$ when $o = \load{t_a}$.\\
  Relation $\bisim$ ensures that $\bn(\strfor(\newstate', t)) \subseteq
  \commits'$. Therefore  we can apply rule ($\cmt$) to show that
  $\newstate' \doublestep{l}{} (\Inst', \storage', \commits' \cup \{t\},
  \decodes') = \newstate'_1$ hence $\newstate_1 \bisim \newstate'_1$.

\textbf{(Case $\ftc$)}
The hypothesis of the rule ensure that
$\ass{c}{t}{\store{\Pc}{}{t_v}} \in \Inst$,
  $\storage(t) = v$, $t \not \in \decodes$,
  $\bn(\strfor(\newstate, t)) \subseteq \decodes$.
  The conclusion  ensures that $\storage_{t} = \emptyset$,
  $\commits_{t}= \emptyset$, $\decodes_{t}=\{t\}$, 
  $I_{t} = \translate(v, max(\Inst))$, and $l = \il{v}$.
  Also, the invariant guarantees that $\den{c}{\storage}$.\\
  The relation $\bisim$ ensures that   $\ass{c}{t}{\store{\Pc}{t_v}} \in \Inst'$,
  $\exists v' . \storage'(t) = v'$, $t \not \in \decodes'$,
  and $\den{c}{\storage'}$.\\
  To show that $\bn(\strfor(\newstate', t)) = \bn(\strfor(\newstate, t))$ we use
  the same reasoning used to prove that $\bn(\stract(\newstate', t)) =
  \{t_s\}$  of case $\exe$ when $o = \load{t_a}$.\\
  Relation $\bisim$ ensures that $\bn(\strfor(\newstate', t)) \subseteq
  \decodes'$.
  We can apply rule ($\ftc$) to show that
  $\newstate' \doublestep{\il{v'}}{} (\Inst' \cup   \translate(v', max(\Inst))), \storage', \commits',
  \decodes' \cup \{t\}) = \newstate'_1$.
  \\
  To show that $v'=v$ and  $\bisim$ is reestablished for $
  \translate(v, max(\Inst))$ we use a similar reasoning to case $\exe$. We
  find two sequential traces that end with the fetch of $t$ and 
  use MIL constant time to show that the value used for the PC
  update must be the same and that the parameter and conditions of the
  newly decoded microinstructions are equivalent in the two states.